\documentclass[sts,preprint]{imsart}

\input{./Definitions}
\usepackage{enumerate}
\usepackage[shortlabels]{enumitem}
\usepackage{fixltx2e}
\usepackage{graphicx}
\usepackage{longtable}
\usepackage{float}
\usepackage{wrapfig}
\usepackage{soul}
\usepackage{amsmath}
\usepackage{mathtools}
\usepackage{mathrsfs}
\usepackage{textcomp}
\usepackage{marvosym}
\usepackage{wasysym}
\usepackage{latexsym}
\usepackage{amssymb}
\usepackage[colorlinks,citecolor=blue,urlcolor=blue]{hyperref}
\usepackage{graphicx}
\usepackage{subfig}
\usepackage{longtable}
\usepackage{float}
\usepackage{array}
\usepackage[ruled]{algorithm}
\usepackage[noend]{algpseudocode}
\usepackage[table]{xcolor}
\usepackage{arydshln}
\usepackage[normalem]{ulem}
\usepackage{url}
\usepackage{comment}
\usepackage{color}
\usepackage{amsmath}
\usepackage[round, comma]{natbib}
\usepackage{graphicx,epsf}
\usepackage[margin=1in]{geometry}
\newtheorem{proposition}{Proposition}[section]

\renewcommand{\mb}{\mathbf{m}}

\newcommand{\var}{\text{var}}

\newcommand{\cov}{\text{cov}}

\newcommand{\smax}{\text{s}_{\text{max}}}
\newcommand{\vertiii}[1]{{\left\vert\kern-0.25ex\left\vert\kern-0.25ex\left\vert #1
    \right\vert\kern-0.25ex\right\vert\kern-0.25ex\right\vert}}

\begin{document}

\begin{frontmatter}

\title{A Divide-and-Conquer Bayesian Approach to Large-Scale Kriging}
\runtitle{Distributed Kriging}


\author{\fnms{Rajarshi} \snm{Guhaniyogi}\thanksref{T2}\ead[label=e1]{rguhaniy@ucsc.edu}}
\address{\printead{e1}}
\affiliation{Department of Statistics, University of California, Santa Cruz}
\author{\fnms{Cheng} \snm{Li}\thanksref{T2}\ead[label=e2]{stalic@nus.edu.sg}}
\address{\printead{e2}}
\affiliation{Department of Statistics and Applied Probability, National University of Singapore}
\author{\fnms{Terrance} \snm{Savitsky}\thanksref{T2}\ead[label=e3]{savitsky.terrance@bls.gov}}
\address{\printead{e3}}
\affiliation{U.~S.~Bureau of Labor Statistics}
\author{\fnms{Sanvesh} \snm{Srivastava}\thanksref{T1} \thanksref{T2}\ead[label=e4]{sanvesh-srivastava@uiowa.edu}}
\address{\printead{e4}}
\affiliation{Department of Statistics and Actuarial Science, The University of Iowa}

\thankstext{T2}{Equally contributing authors.}
\thankstext{T1}{Corresponding author.}

\runauthor{Guhaniyogi, Li, Savitsky, Srivastava}

\begin{abstract}
We propose a three-step divide-and-conquer strategy within the Bayesian paradigm {that delivers} massive scalability for \emph{any} spatial process model. We partition the data into a large number of subsets, apply a readily available Bayesian spatial process model on every subset, in parallel, and optimally combine the posterior distributions estimated across all the subsets into a pseudo posterior distribution that conditions on the entire data. The combined pseudo posterior distribution {replaces the full data posterior distribution for} predicting the responses at arbitrary locations and for inference on the model parameters and spatial surface. {Based on distributed Bayesian inference,} our approach is called ``Distributed Kriging'' (DISK) and offers significant advantages in {massive data} applications where the full data are stored across multiple machines. {We show theoretically that the Bayes $L_2$-risk of the DISK posterior distribution achieves the near optimal convergence rate in estimating the true spatial surface with various types of covariance functions, and provide upper bounds for the number of subsets as a function of the full sample size.} The {model-free feature of DISK is demonstrated by scaling posterior computations in spatial process models with a stationary full-rank and a nonstationary low-rank Gaussian process (GP) prior}. A variety of simulations and a geostatistical analysis of the Pacific Ocean sea surface temperature data validate our theoretical results.
\end{abstract}


\end{frontmatter}

\section{Introduction}\label{sec:intro}

A fundamental challenge in geostatistics is the analysis of massive spatially-referenced data. Massive spatial data provide scientists with an unprecedented opportunity to hypothesize and test complex theories, see for example \citet{Geletal10,CreWik11,Banetal14}. This has led to the development of complex and flexible hierarchical GP-based models that are computationally intractable for {a large number of spatial locations, denoted as $n$, }due to the $O(n^3)$ computational cost and the $O(n^2)$ storage cost. {We develop a three-step general distributed Bayesian approach, called Distributed Kriging (DISK), for {boosting} the scalability of any state-of-the-art spatial process model based on GP prior or its variants to multiple folds using the divide-and-conquer technique.}

{There is an extensive literature on {scalable} Gaussian process (GP)-based modeling of massive spatial data due to its great practical importance \citep{heaton2017methods}. We provide a brief overview of basic ideas, deferring detailed comparisons of the existing literature with DISK to Section \ref{disk-comp}. A common idea in GP-based modeling is to seek dimension-reduction by endowing the spatial covariance matrix either with a low-rank or a sparse structure. Low-rank structures represent {a} spatial surface using a small number of  \emph{a priori} chosen basis functions such that the posterior computations scale in the cubic order to the number of chosen basis functions {(rather than the number of spatial locations)}, resulting in reduced storage and computational costs. Sparse {structured} models assume that the spatial correlation between two distantly located observations is nearly zero. If the assumption is true, then little information is lost by assuming independence between data at distant locations. Another approach introduces sparsity in the inverse covariance matrix using conditional independence assumptions or composite likelihoods. Some variants of dimension-reduction methods partition the spatial domain into sub-regions containing fewer spatial locations. Each of these sub-regions is modeled using a GP which are then hierarchically combined by borrowing information across the sub-regions. }

{The proposed DISK framework does not belong to any of these classes of methods, but {it enhances the scalability of any of these methods by embedding each within the three-step DISK framework.  The outline of the DISK framework is as follows.}} {First, the $n$ spatial locations are divided into $k$ subsets such that each subset has representative data samples from all regions of the spatial domain with the $j$th subset containing $m_j$ data samples.} {Second, posterior computations are implemented in parallel on the $k$ subsets using any chosen spatial process model after raising the model likelihood to a power of $n/m_j$ in the $j$th subset.} The pseudo posterior distribution obtained using the modified likelihood is called the ``subset posterior distribution.'' {Since $j$th subset posterior distribution conditions on $(m_j/n)$-fraction of the full data, the modification of the likelihood by raising it to the power of $n/m_j$ ensures that variance of each subset posterior is of the same order (as a function of $n$) as that of the full data posterior distribution.} Third, the $k$ subset posterior distributions are combined into a single pseudo probability distribution, called the DISK pseudo posterior (henceforth,  DISK posterior), that conditions on the full data and replaces the computationally expensive full data posterior distribution for prediction and inference.

{Our novel contributions to the growing literature on distributed Bayesian inference are two-fold.} Computationally, the main innovations are in the second and third steps because {the literature on} general sampling and combination schemes {is sparse} in process-based modeling of spatial data using the divide-and-conquer technique. No restrictive data- or model-specific assumptions, such as the independence between data subsets or independence between blocks of parameters, are adopted and the DISK framework still allows principled Bayesian inference with parameter estimation, surface interpolation, and prediction. {Theoretically, we provide guarantees on the accuracy of performance in estimating the true spatial surface using the DISK posterior as a function of $n$, $k$, and analytic properties of the true spatial surface.
We show that when $k$ is controlled to increase in some proper order of $n$ as $n$ tends to infinity, the Bayes $L_2$-risk of the DISK posterior achieves near minimax optimal convergence rates under different types of covariance functions.
There are some theoretical results in this direction \citep{ShaChe15,CheSha15,SzaVan17}, but DISK is the first general Bayesian framework addressing these theoretical problems with a focus on computationally efficient posterior computations in massive data applications with complex nonparametric models, while avoiding restrictive assumptions that limit wide applicability.}


{We illustrate the application of DISK for enhancing the scalability of two representative GP priors. One is the usual full-rank GP prior with a stationary covariance kernel and the other is a low-rank GP prior with a nonstationary covariance kernel called the modified predictive process (MPP) prior \citep{Finetal09}. The latter prior is commonly used for estimating nonstationary surfaces in large spatial data. MPP constructs a low-rank approximation of covariance matrix for the generating distribution of the spatial surface to reduce computation time, but if the rank is moderately large, then MPP struggles to provide accurate inference in a manageable time even for $10^4$ observations.} Our numerical results presented later establish that DISK with MPP prior scales to $10^6$ observations without compromising on either computational efficiency or accuracy in inference and prediction. An interesting empirical observation is that under a fixed computation budget the accuracy of MPP prior in detecting local surface features is enhanced by embedding it within the DISK framework in the sense that we are able to increase the spatial resolution.  We expect this conclusion to hold for all of the popular structured GP priors.


\subsection{DISK and Existing Methods for GP-Based Modeling of Massive Spatial Data}
\label{disk-comp}

{The DISK framework does not compete with existing methods for analyzing massive spatial data, but aims to boost their scalability using the divide-and-conquer technique. With this in mind, we compare DISK with existing approaches for GP-based spatial modeling based on variants of dimension-reduction technique and refer to \citet{heaton2017methods} for a more comprehensive review.} Low-rank structures on the spatial covariance matrix {are the most widely used tool for computationally efficient spatial computation.} They represent the spatial surface using $r$  \emph{apriori} chosen basis functions with associated computational complexity of $O(nr^2+r^3)$ \citep{CreJoh08,Banetal08,Finetal09,Guhetal11,Banetal10,SanHua12,wikle2010low}; however, practical considerations entail that $r$ grows roughly as $O(\sqrt{n})$ for accurate estimation, implying that $O(nr^2)$ flops are also expensive in low-rank structures. In fact, with a small $(r/n)$-ratio, scientists have observed shortcomings in many of the above methods for approximating GPs such as the propensity to oversmooth the data \citep{stein2014limitations,simpson2012order}. {DISK offers a solution to this problem. If $m_j \ll n$, then $(r/m_j)$-ratio is relatively large on the subsets, yielding accurate and computationally efficient inference using subset posteriors. Our theoretical results guarantee that the DISK posterior has better accuracy than any subset posterior, which can potentially outperform the full data posterior estimated using the same prior. Our simulations empirically confirm  this claim for the MPP prior.}

A specific form of sparse structure uses compactly supported covariance functions to create sparse spatial covariance matrices that approximate the full covariance matrix \citep{Kauetal08,Furetal12}. They are useful for parameter estimation and interpolation of the response (``kriging''), but not for more general inference on the latent processes due to an expensive determinant computation of the massive covariance matrix. An alternative approach is to introduce sparsity in the inverse covariance (precision) matrix of the GP likelihoods using products of lower dimensional conditional distributions \citep{Vec88, Rueetal09, Steetal04}, or via composite likelihoods \citep{Eidetal14, bai2012joint}.
There are recent approaches, extending these ideas, that can introduce sparsity in the inverse covariance (precision) matrix of process realizations and hence enable ``kriging'' at arbitrary locations  \citep{Datetal15, guinness2016permutation}. In related literature on computer experiments, localized approximations of GP models are proposed, see, for example,  \cite{gramacy2015local}. DISK relaxes the trade-off between computation time and the accuracy in modeling a spatial surface.  In current practice, approximation methods are used with the intent to make the computations feasible at the expense of accuracy.  Reduced rank simplifications of the covariance matrix may produce over-smoothing that limits the ability to detect local features, while sparse covariance structures may underestimate correlations.  Yet, both reduced rank and sparse covariance structures may be easily embedded in our DISK framework to dramatically scale the computations such that the degree of approximation required may be notably reduced, which we demonstrate in the sequel.

{The remaining variants of dimension-reduction methods combine the benefits of low-rank and sparse structure covariance functions.} Examples include non-stationary models \citep{Banetal14} and multi-level and multi-resolution models \citep{gelfand2007multilevel, Nycetal15, Kat16, GuhanSanso2017}. These models usually achieve scalability by assuming block-independence at some level of the hierarchy, usually across sub-regions, but may lose scalability when they borrow information across sub-regions. Multi-resolution models are in general difficult to implement,  do not generally come with desirable theoretical guarantees concerning large sample behavior, and may become less amenable to various modification to suit different applications. In contrast, DISK makes no independence assumptions across subregions to accomplish predictions at new locations on a spatial surface and can fit a multiresolution model in each subset for enhancing its scalability. \citet{lindgren2011explicit} proposed an approximation based on viewing a GP with Mat\'ern
covariance as the solution to the corresponding stochastic partial differential equation, but
this approach is only applicable to covariance functions of Mat\'ern type and may not be applicable in scaling GP with low-rank kernels.

There is a class of methods, of which DISK is a member, that divide the data into a large number of subsets, draw inference in parallel on the subsets, and combine the inferences by some mechanism that approximates the inference conditional on the full data. \citet{barbian2017spatial} propose combining point estimates of spatial parameters obtained from different subsamples, but they do not provide combined inference on the spatial processes or predictions. Similarly, \citet{heaton2017nonstationary} partition the spatial domain and assume independence between the data in different partitions. Although computationally attractive, assuming independence across subdomains may trigger loss in predictive uncertainty as demonstrated in \cite{heaton2017methods}. In a similar effort to the DISK posterior, \cite{guhaniyogi2017meta} propose drawing subset inferences and combine the posterior distributions in subsets using the idea of ``meta-posterior''. This approach has an added advantage over that of \citet{heaton2017nonstationary} in that it does not assume independence across data blocks and enables prediction with accurate characterization of uncertainty \citep{heaton2017methods}; however, it produces desirable inference \emph{only} when a stationary GP model is fitted in each subset and is not accurate in estimation of the spatial surface when nonstationary low-rank models (e.g. MPP) are fitted in each subset. This limits the applicability of the {meta-posterior}. Also, \citet{guhaniyogi2017meta} do not offer any theoretical guidance on choosing the number of subsets for optimal inference on the  spatial surface. The proposed DISK framework fills both these gaps. Our experiments also demonstrate that the DISK posterior provides accurate uncertainty quantification unlike
some of the divide-and-conquer approaches popularly used in the machine learning literature such as Consensus Monte Carlo \citep{Scoetal16}.

{Sampling algorithms are computationally inefficient in massive data settings, so this has motivated significant interest in developing general approaches to scalable Bayesian inference using the divide-and-conquer technique. The DISK framework builds on the recent works that combine the subset posterior distributions through their geometric centers, such as the mean or the median, and guarantee wide applicability under general assumptions \citep{Minetal14,Srietal15,Lietal16,Minetal17,SavSri16,Srietal17}. A major limitation of the current distributed approaches is that the theory and practice is limited to parametric models. {By contrast}, the DISK framework is tuned for accurate and computationally efficient posterior inference in nonparametric Bayesian models based on GP priors. In particular, we develop (a) a new approach to modify the likelihood for computing the subset posterior distribution of an unknown function, an infinite-dimensional parameter, (b) generalizations of existing algorithms for a full-rank and a low-rank GP prior to general MCMC samples from a subset distribution with modified likelihood, {and (c) theoretical guarantees on the convergence rate of the DISK posterior to the true function, and guidance on choosing $k$ depending on the covariance function and $n$, such that the DISK posterior maintains near minimax optimal performance as $n$ tends to infinity}.

}


The remainder of the manuscript evolves as follows. In Section 2 we outline a Bayesian hierarchical mixed model framework that incorporates models based on both the full-rank and the low-rank GP priors. Our DISK approach will work with posterior MCMC samples from such models. Section 3 develops the framework for DISK, discusses how to compute the DISK posterior distribution, and offers theoretical insights into the DISK for general GPs and their approximations. A detailed simulation study followed by an analysis of the Pacific ocean sea surface temperature data are illustrated in Section 4 to justify the use of DISK for real data. Finally, Section 5 discusses what DISK achieves, and proposes a number of future directions to explore. Proofs of the theoretical results in Section 3 are offered in the supplementary material. It also offers additional theoretical results concerning convergence rate of the DISK posterior.

\section{Bayesian inference in GP-based spatial models}\label{Sec: Pooled_Bayesian_Spatial_Regression}

Consider the univariate spatial regression model for the data observed at location $\sbb$ in a compact domain $\Dcal$,
\begin{align}\label{parent_proc}
  y(\sbb)=\xb(\sbb)^T \betab + w(\sbb) + \epsilon(\sbb),
\end{align}
where $y(\sbb)$ and $\xb(\sbb)$ are the response and a $p\times 1$ predictor vector respectively at $\sbb$, $\betab$ is a $p\times 1$ predictor coefficient, $w(\sbb)$ is an unknown spatial function $w(\cdot)$ at $\sbb$, and {$\epsilon(\sbb)$ is the white-noise process $\epsilon(\cdot)$ at $\sbb$, which is independent of $w(\cdot)$. The Bayesian implementation of the model in  \eqref{parent_proc} customarily assumes  (a) that $\betab$ {apriori} follows N($\mub_{\betab}$, ${\boldsymbol \Sigma}_{\betab}$) and (b) that $w(\cdot)$ and $\epsilon(\cdot)$ apriori follow mean 0 GPs with covariance functions $C_{\alphab}(\sbb_1,\sbb_2)$ and $D_{\alphab}(\sbb_1,\sbb_2)$ that model  $\cov\{w(\sbb_1),w(\sbb_2)\}$ and $\cov\{\epsilon(\sbb_1), \epsilon(\sbb_2)\}$, respectively, where
$\alphab$ are the process parameters indexing the two families of covariance functions and $\sbb_1, \sbb_2 \in \Dcal$; therefore, the model parameters are $\Omegab = \{\alphab, \betab\}$.
The training data consists of predictors and responses observed at $n$ spatial locations, denoted as $\Scal = \{\sbb_1, \ldots, \sbb_n\}$. }

Standard Markov chain Monte Carlo (MCMC) algorithms exist for performing posterior inference on $\Omegab$ and the values of $w(\cdot)$ at a given set of locations $\Scal^* = \{\sbb_1^*, \ldots, \sbb_l^*\}$, where $\Scal^* \cap \Scal = \emptyset$, and for predicting $y(\sbb^*)$ for any $\sbb^* \in \Scal^*$ \citep{Banetal14}. Given $\Scal$, the prior assumptions on $w(\cdot)$ and $\epsilon(\cdot)$ imply that $\wb^T = \{w(\sbb_1), \ldots, w(\sbb_n)\}$ and $\epsilonb^T = \{\epsilon(\sbb_1), \ldots, \epsilon(\sbb_n)\}$   are independent and follow $N \left\{ \zero, \Cb({\alphab}) \right\}$ and $N \left\{ \zero, \Db({\alphab}) \right\}$, respectively, with the $(i,j)$th entries of $\Cb(\alphab)$ and $\Db(\alphab)$ are $C_{\alphab}(\sbb_i,\sbb_j)$ and $D_{\alphab}(\sbb_i,\sbb_j)$, respectively.
The hierarchy in \eqref{parent_proc} is completed by assuming that $ \alphab $ {apriori} follows a distribution with density $\pi(\alphab)$. The MCMC algorithm for sampling $\Omegab$, $\wb^{*T} = \{w(\sbb_1^*), \ldots, w(\sbb_l^*)\}$, and $\yb^{*T} = \{y(\sbb_1^*), \ldots, y(\sbb_l^*)\}$ cycle through the following three steps until sufficient MCMC samples are drawn post convergence:
\begin{enumerate}
\item Integrate over $\wb$ in \eqref{parent_proc} and
  \begin{enumerate}
  \item sample $\betab$ given  $\yb,\Xb,\alphab$ from $N(\mb_{\betab}, \Vb_{\betab})$, where
    \begin{align}
      \label{eq:samp-beta}
        \Vb_{\betab} = \left\{ \Xb^T \Vb(\alphab)^{-1} \Xb + \Sigmab_{\betab}^{-1}   \right\}^{-1}, \quad \mb_{\betab} = \Vb_{\betab} \left\{ \Xb^{T} \Vb(\alphab)^{-1} \yb +  \Sigmab_{\betab}^{-1} \mub_{\betab}   \right\},
    \end{align}
    where $\Xb=[\xb(\sbb_1): \cdots :\xb(\sbb_n)]^T$ is the $n\times p$ matrix of predictors, with $p < n$, and $ \Vb(\alphab) = \Cb(\alphab) + \Db(\alphab)$; and
  \item sample $\alphab$ given $\yb,\Xb,\betab$ using the Metropolis-Hastings algorithm with a normal random walk proposal.
  \end{enumerate}
\end{enumerate}
\begin{enumerate}
 \setcounter{enumi}{1}
\item Sample $\wb^*$ given $\yb,\Xb,\alphab,\betab$ from $N(\mb_{*}, \Vb_*)$, where
  \begin{align}
    \label{eq:samp-w}
    \Vb_* = \Cb_{*,*}(\alphab) - \Cb_{*}(\alphab) \Vb(\alphab)^{-1} \Cb_{*}(\alphab)^T, \quad \mb_* = \Cb_{*}(\alphab) \Vb(\alphab)^{-1} (\yb - \Xb \betab),
  \end{align}
  $\Cb_*(\alphab)$  and $\Cb_{*,*}(\alphab)$ are $l \times n$ and $l \times l$ matrices, respectively, and the  $(i,j)$th entries of $\Cb_{*,*}(\alphab)$ and $\Cb_*(\alphab)$ are $C_{\alphab}(\sbb_i^*,\sbb_j^*)$ and $C_{\alphab}(\sbb_i^*,\sbb_j)$, respectively.
\item Sample $\yb^*$ given $\alphab,\betab,\wb^*$ from $N\left\{ \Xb^* \betab + \wb^*, \Db(\alphab) \right\}$, where
  $\Xb^{*T} = [\xb(\sbb_1^*) : \cdots : \xb(\sbb_l^*)]$.
\end{enumerate}
Many Bayesian spatial models can be formulated in terms of \eqref{parent_proc} by assuming different forms of $C_{\alphab}(\sbb_1,\sbb_2)$ and $D_{\alphab}(\sbb_1,\sbb_2)$; see \cite{Banetal14} and supplementary material for details on the MCMC algorithm. Irrespective of the form of $\Db(\alphab)$, if no additional assumptions are made on the structure of $\Cb(\alphab)$, then the three steps require $O(n^3)$ flops in computation and $O(n^2)$ memory units in storage in every MCMC iteration. Spatial models with this form of posterior computations are based on a \emph{full-rank} GP prior. In practice, if $n \geq 10^4$, then posterior computations in a model based on a full-rank GP prior are infeasible due to numerical issues in matrix inversions involving an unstructured $\Cb(\alphab)$.

{There are methods which either impose a low-rank structure or a sparse structure on $\Cb(\alphab)$ to address this computational issue \citep{Banetal14}. Methods with a low-rank structure on $\Cb(\alphab)$ expresses $\Cb(\alphab)$ in terms of $r \ll n$ basis functions (with $r=O(\sqrt{n})$ is desirable for accurate inference), in turn inducing a \emph{low-rank} GP prior. Again,
a class of sparse structure uses compactly supported covariance functions to create $\Cb(\thetab)$ with overwhelming zero entries \citep{Kauetal08,Furetal12}, where as another variety of sparse structure imposes a Markov random field model on the joint distribution of $\yb$ \citep{Vec88, Rueetal09, Steetal04} or $\wb$ \citep{Datetal15, guinness2016permutation}. We use the MPP prior as a representative example of this broad class of computationally efficient methods.} Let $\mathcal{S}^{(0)}=\{\sbb_{1}^{(0)},...,\sbb_{r}^{(0)}\}$ be a set of $r$ locations, known as the ``knots,'' which may or may not intersect with $\Scal$. Let $\cb(\sbb, \mathcal{S}^{(0)})= \{ C_{\alphab}(\sbb,\sbb_{1}^{(0)}), \ldots,C_{\alphab}(\sbb,\sbb_{r}^{(0)}) \}^T$ be an $r\times 1$ vector and $\Cb(\mathcal{S}^{(0)})$ be an $r\times r$ matrix whose $(i,j)$th entry is $C_{\alphab}(\sbb_{i}^{(0)},\sbb_{j}^{(0)})$. Using $\cb(\sbb_{1}, \mathcal{S}^{(0)}), \ldots, \cb(\sbb_{n}, \mathcal{S}^{(0)})$ and $\Cb(\mathcal{S}^{(0)})$, define the diagonal matrix $\deltab =\diag\{\delta(\sbb_1), \ldots, \delta(\sbb_n)\}$ with
$\delta(\sbb_i) =  C_{\alphab}(\sbb_i,\sbb_i) - \cb^T(\sbb_i, \mathcal{S}^{(0)}) \Cb(\mathcal{S}^{(0)})^{-1}\cb(\sbb_i, \Scal^{(0)})$,  $i = 1, \ldots, n$.
Let $\one({\ab=\bb}) = 1$ if $\ab=\bb$ and 0 otherwise. Then, MPP is a GP with  covariance function
\begin{align}
  \label{eq:mpp-mdl-cov}
  \tilde C_{\alphab}(\sbb_1,\sbb_2) = \cb^T(\sbb_1, \mathcal{S}^{(0)}) \Cb(\mathcal{S}^{(0)})^{-1}\cb(\sbb_2, \mathcal{S}^{(0)}) + \delta(\sbb_1) \one({\sbb_1=\sbb_2}), \quad \sbb_1, \sbb_2 \in \Dcal,
\end{align}
where $\tilde C_{\alphab}(\sbb_1,\sbb_2)$ depends on the covariance function of the parent GP and the selected $r$  knots, which define $\Cb(\mathcal{S}^{(0)})$, $\cb^T(\sbb_1, \mathcal{S}^{(0)})$, and $\cb^T(\sbb_2, \mathcal{S}^{(0)})$. We have used a  $\, \tilde {   } \,$ in \eqref{eq:mpp-mdl-cov} to distinguish the covariance function of a low-rank GP prior from that of its parent full-rank GP. If $\tilde \Cb(\alphab)$ is a matrix with $(i,j)$th entry $\tilde C_{\alphab}(\sbb_i,\sbb_j)$, then the posterior computations using MPP, a low-rank GP prior, replace $\Cb(\alphab)$ by $\tilde \Cb(\alphab)$ in the steps 1(a), 1(b), and 2. The (low) rank $r$ structure imposed by $\Cb(\mathcal{S}^{(0)})$ implies that $\tilde \Cb(\alphab)^{-1}$ computation requires $O(nr^2)$ flops using the Woodbury formula \citep{Har97}; however, massive spatial data require that $r = O(\sqrt{n})$, leading to the computational inefficiency of low-rank methods.
The next section develops our DISK framework, which uses the divide-and-conquer technique to scale the posterior computations using full-rank and low-rank GP priors.

\section{Distributed Kriging}

\subsection{First step: partitioning of spatial locations}
\label{sec:partitioning-step}

We partition the $n$ spatial locations into $k$ subsets. The value of $k$ depends on the chosen covariance function used in the spatial model, and it is set to be large enough to ensure computationally efficient posterior computations on any subset. The default partitioning scheme is to randomly allocate the locations into $k$ {possibly overlapping} subsets (referred to as the random partitioning scheme hereon) to ensure that each subset has representative data samples from all subregions of the domain.
{Let $\Scal_j= \{\sbb_{j1}, \ldots, \sbb_{jm_j}\}$ denote the set of $m_j$ spatial locations in subset $j$ ($j=1, \ldots, k$). A spatial location can belong to multiple subsets so that  $\sum_{j=1}^{k} m_j \geq n$ but $\cup_{j=1}^k \mathcal{S}_j = \mathcal{S}$, where $\sbb_{ji} = \sbb_{i'}$ for some $\sbb_{i'} \in \Scal$ and for every $i=1, \ldots, m_j$ and $j=1, \ldots, k$. Denote the data in the $j$th partition as  $\{\yb_j, \Xb_j\}$  ($j = 1,\ldots, k$), where $\yb_j = \{y(\sbb_{j1}), \ldots, y(\sbb_{jm_j})\}^T$ is a $ m_j \times 1$ vector and $\Xb_j =
[\xb(\sbb_{j1}): \cdots :\xb(\sbb_{jm_j})]^T$ is a $m_j \times p$ matrix of predictors corresponding to the spatial locations in $\Scal_j$ with $p < m_j$.}  {In modern grid or cluster computing environments, all the machines in the network have similar computational power, so the performance of DISK is optimized  by choosing similar values of $m_1, \ldots, m_k$. }

{One can choose more sophisticated partitioning schemes than random partitioning. For example, it is possible to cluster the data based on centroid clustering \citep{knorr2000bayesian} or hierarchical clustering based on spatial
gradients \citep{anderson2014identifying,heaton2017nonstationary}, and then construct subsets such that each subsets contains representative data samples from each cluster. Detailed exploration later shows that even random partitioning leads to desirable inference in the various simulation settings and in the sea surface data example, hence inferential improvement with any other sophisticated partitioning should be marginal in these examples. {Perhaps more sophisticated blocking methods may provide further improvement in the cases where spatial locations are drawn based on specific designs; for example, sophisticated partitioning schemes have inferential benefits when a sub-domain shows substantial local behavior compared to the others \citep{GuhanSanso2017}, or sampled locations are chosen based on a specific survey design. Since they are {atypical} examples in the spatial context, we will pursue them elsewhere in greater detail.}}


The univariate spatial regression models using either a full-rank or a low-rank GP prior for the data observed at any location $\sbb_{ji} \in \Scal_j \subset \Dcal$ is given by
\begin{align}\label{eq:mdl-j}
  y(\sbb_{ji}) = \xb(\sbb_{ji})^T \betab + w(\sbb_{ji}) + \epsilon(\sbb_{ji}), \quad i=1, \ldots, m_j.
\end{align}
Let $\wb_j^T = \{w(\sbb_{j1}), \ldots, w(\sbb_{jm_j})\}$ and $\epsilonb_j^T = \{\epsilon(\sbb_{j1}), \ldots, \epsilon(\sbb_{jm_j})\}$ be the realizations of GP $w(\cdot)$ and white-noise process $\epsilon(\cdot)$, respectively, in the $j$th subset.
{After marginalizing over $\wb_j$ in the GP-based model for the $j$th subset, the likelihood of $\Omegab = \{\alphab, \betab\}$ is given by
  $\ell_{j}(\Omegab) = N \{\yb_j \mid \Xb_j \betab, \Vb_j(\alphab)\}$,
where $\Vb_j(\alphab) = \Cb_j(\alphab) + \Db_j(\alphab)$ and  $\Vb_j(\alphab) = \tilde \Cb_j(\alphab) + \Db_j(\alphab)$ for full-rank and low-rank GP priors, respectively, and $\Cb_j(\alphab), \tilde \Cb_j(\alphab), \Db_j(\alphab)$ are obtained by extending the definitions of $\Cb(\alphab), \tilde \Cb(\alphab), \Db(\alphab)$ to the $j$th subset. In a model based on full-rank or low-rank GP prior, the likelihood of $\wb_j$ given $\yb_j$, $\Xb_j$, and $\Omegab$  is
  $\ell_{j }(\wb_j) = N \{\yb_j - \Xb_j \betab \mid \wb_j, \Db_j(\alphab) \}.$
The likelihoods in $\ell_{j}(\Omegab)$ and $\ell_{j}(\wb_j)$ are used to define the
posterior distributions for $\betab, \alphab, \wb^*$,  $\yb^*$ based on a full-rank or a low-rank GP prior in subset $j$ and are called $j$th subset posterior distributions.}

\subsection{Second step: sampling from subset posterior distributions}
\label{sec:stoch-appr}

We define subset posterior distributions by modifying the likelihoods in $\ell_{j}(\Omegab)$ and $\ell_{j}(\wb_j)$. More precisely, the density of the $j$th subset posterior distribution of $\Omegab$ is given by
\begin{align}\label{eq:fullsub-j}
\pi_{m_j} (\Omegab \mid \yb_j) &= \frac{\{\ell_{j }(\Omegab) \}^{n/m_j}\pi(\Omegab) } {\int \{\ell_{j }(\Omegab) \}^{n/m_j}\pi(\Omegab) d \Omegab},
\end{align}
where we assume that $\int \{\ell_{j }(\Omegab) \}^{n/m_j}\pi(\Omegab) d \Omegab < \infty$,
and the subscript `$m_j$' denotes that the density conditions on $m_j$ data samples in the $j$th subset. The modification of likelihood to yield the subset posterior density in \eqref{eq:fullsub-j} is called \emph{stochastic approximation} \citep{Minetal14}. Raising the likelihood to the power of $n/m_j$ is equivalent to replicating every $y(\sbb_{ji})$ $n/m_j$ times ($i=1, \ldots, m_j$), so stochastic approximation accounts for the fact that the $j$th subset posterior distribution conditions on a $(m_j/n)$-fraction of the full data and ensures that its variance is of the same order (as a function of $n$) as that of the full data posterior distribution. Unlike parametric models, stochastic approximation in spatial regression models has not been studied previously in the literature. We address this gap next.

With the proposed stochastic approximation in (\ref{eq:fullsub-j}), the full conditional densities of $j$th subset posterior distributions for prediction and inference follow from their full data counterparts. The $j$th full conditional densities of $\betab$ and $\alphab$ in the GP-based models are
\begin{align}
  \label{eq:sub-j}
  \pi_{m_j} (\betab \mid \yb_j, \alphab) &= \frac{\{\ell_{j }(\Omegab) \}^{n/m_j}\pi(\betab) } {\int \{\ell_{j }(\Omegab) \}^{n/m_j}\pi(\betab) d \betab}, \quad
  \pi_{m_j} (\alphab \mid \yb_j, \betab) = \frac{\{\ell_{j}(\Omegab) \}^{n/m_j}\pi(\alphab) } {\int \{\ell_{j }(\Omegab) \}^{n/m_j}\pi(\alphab) d \alphab},
\end{align}
where $\pi(\betab) = N(\mub_{\betab}, {\boldsymbol \Sigma}_{\betab})$, $\pi(\alphab)$ is the prior density of $\alphab$, and we assume that $\int \{\ell_{j}(\Omegab) \}^{n/m_j}\pi(\betab) d \betab$
and $\int \{\ell_{j }(\Omegab) \}^{n/m_j}\pi(\alphab) d \alphab$ respectively are finite.
The $j$th full conditional densities of $\yb^*$ and $\wb^*$ are calculated after modifying the likelihood of $\wb_j$ using stochastic approximation. Given $\yb_j$, $\Xb_j$, and $\Omegab$, straightforward calculation yields that the $j$th subset posterior predictive density of $\wb^*$ is $\pi_{m_j} (\wb^* \mid \yb_j, \Omegab) = N(\wb^* \mid \mb_{j*}, \Vb_{j*})$, with
  \begin{align}
    \label{eq:sub-w-j}
    &\Vb_{j*} = \Cb_{*,*}(\alphab) - \Cb_{*j}(\alphab) \Vb_j(\alphab)^{-1} \Cb_{*j}(\alphab)^T, \quad \mb_{j*} = \Cb_{*j}(\alphab) \Vb_j(\alphab)^{-1} (\yb_j - \Xb_j \betab),
  \end{align}
  where $\Vb_j(\alphab) = \Cb_j(\alphab) + (n/m_j)^{-1} \Db_j(\alphab)$ and $\Vb_j(\alphab) = \tilde \Cb_j(\alphab) + (n/m_j)^{-1} \Db_j(\alphab)$ for full-rank and low-rank GP priors, respectively,  and
  $\Cb_{*,*}(\alphab), \Cb_{*j}(\alphab)$ are $l \times l$, $l \times m_j$ matrices obtained by extending the definition in \eqref{eq:samp-w} to subset $j$ for full-rank and low-rank GP priors with covariance functions $C_{\alphab}(\cdot, \cdot)$ and $\tilde C_{\alphab}(\cdot, \cdot)$, respectively. We note that the stochastic approximation exponent, $n/m_j$, scales $\Db_j(\alphab)$ in $\Vb_j(\alphab)$
  so that the uncertainty in subset and full data posterior distributions are of the same order (as a function of $n$). The $j$th subset posterior predictive density of $\yb^*$ given the MCMC samples of $\wb^*$  and $\Omegab$ in the $j$th subset is
 $N\{\yb^* \mid \Xb^* \betab + \wb^*, \Db_j(\alphab)\}$.
  We employ the same three-step sampling algorithm, as earlier introduced, specialized to subset $j$  ($j = 1,\ldots,k$), sampling $\{\betab, \alphab, \yb^*, \wb^*\}$ in each subset across multiple MCMC iterations; see supplementary material for detailed derivations of subset posterior sampling algorithms in the full-rank and low-rank GP priors. The computational complexity of $j$th subset posterior computations follows from their full data counterparts if we replace $n$ by $m_j$. Specifically, the computational complexities
for sampling a subset posterior distribution are $O(m^3)$ and $O(mr^2)$ flops per iteration if the model in \eqref{eq:mdl-j} uses a full-rank or a low-rank GP prior, respectively, where $m = \max_{j} m_j$. Performing subset posterior computations in parallel across $k$ servers also alleviates the need to store large covariance matrices.

The combination of subset posteriors outlined below is more widely applicable compared to other divide-and-conquer type approaches as it is free of model- or data-specific assumptions, such as independence of samples in training data, except that every subset posterior distribution has a density and has finite second moments.

\subsection{Third step: combination of subset posterior distributions}
\label{combinescheme}

{The combination step relies on the notion of Wasserstein barycenter, as used in some related scalable Bayes literature for independent data \citep{Lietal16,Srietal17}. We first provide some background on this topic.} Let $(\Theta, \rho)$ be a complete separable metric space and $\Pcal(\Theta)$ be the space of all probability measures on $\Theta$. The Wasserstein space of order $2$ is a set of probability distributions defined as
$\Pcal_2(\Theta) = \{\mu \in \Pcal(\Theta): \int_{\Theta} \rho^2(\theta, \theta_0) \mu(d \theta) < \infty \}$,
where $\theta_0 \in \Theta$ is arbitrary and $\Pcal_2(\Theta) $ does not depend on the choice of $\theta_0$. The Wasserstein distance of order 2, denoted as $W_2$, is a metric on $\Pcal_2(\Theta)$. Let $\mu, \nu$ be two probability measures in $\Pcal_2(\Theta)$ and $\Pi (\mu, \nu)$ be the set of all probability measures on $\Theta \times \Theta$ with marginals $\mu$ and $\nu$, then $W_2$ distance between $\mu$ and $\nu$ is defined as
$W_2(\mu, \nu) = \{\underset{\pi \in \Pi (\mu, \nu)} {\mathrm{inf}} \, \int_{\Theta \times \Theta} \rho^2(x, y) \, d \pi(x, y) \}^{1/2}.$
Let $\nu_1, \ldots, \nu_k \in \Pcal_2(\Theta)$, then the Wasserstein barycenter of $\nu_1, \ldots, \nu_k$ is defined as
\begin{align}
  \overline \nu = \underset{ \nu \in \Pcal_2(\Theta)} {\argmin} \, \frac{1}{k} \sum_{j=1}^{k} W_2^2 (\nu, \nu_j). \label{eqn:w-bary}
\end{align}
It is known that $\overline \nu$ exists and is unique \citep{AguCar11}.


{
In the DISK framework, for any parameter of interest $\theta$, either a scalar or a vector, the DISK posterior is defined to be the Wasserstein barycenter of the $k$ subset posterior distributions of $\theta$. Here, $\theta$ can be taken as $\betab$, $\alphab$, $\wb^*$, $\yb^*$, their individual components, or any functionals of these parameters. In other words, for our DISK approach, $\nu_1, \ldots, \nu_k$ in \eqref{eqn:w-bary} are taken as the $k$ subset posterior distributions of $\theta$. Hence the DISK posterior, mathematically computed from the Wasserstein barycenter $\overline \nu$ in \eqref{eqn:w-bary}, provides a general notion of obtaining the mean of $k$ possibly dependent subset posterior distributions. For Bayesian inference, the exact subset posteriors of $\theta$ ($\nu_1, \ldots, \nu_k$ in \eqref{eqn:w-bary}) are analytically intractable in general, but they can be well approximated by the subset posterior MCMC samples of $\theta$, and we can conveniently estimate the empirical version of the Wasserstein barycenter $\overline \nu$ by efficiently solving a sparse linear program as described in \citep{CutDou14,Srietal15,Staetal17}. It has been shown that for independent data, the Wasserstein barycenter is a preferable choice to several other combination methods \citep{Lietal16,Srietal17}; for example, directly averaging over many subset posterior densities with different means can usually result in an undesirable multimodal pseudo posterior distribution, but the Wasserstein barycenter does not have this problem and can recover a unimodal posterior; see, for example, Figure 1 in \citet{Srietal17}. Besides, it does not rely on the asymptotic normality of
the subset posterior distributions as in other approches, such as consensus Monte Carlo \citep{Scoetal16}.

If $\theta$ represents a one-dimensional functional of interest (a functional of $\betab$, $\alphab$, $\wb^*$, or $\yb^*$), then the DISK posterior of $\theta$ can be easily obtained by averaging empirical subset posterior quantiles \citep{Lietal16}. This is because the $W_2$ distance between two univariate distributions is the same as the $L_2$ distance between their quantile functions (Lemma 8.2 of \citealt{BicFre81}). In particular, let $\nu$ and $\nu_j$ be the full posterior and $j$th subset posterior distribution of $\theta$, and $\overline \nu$ be the Wasserstein barycenter of $\nu_1, \ldots, \nu_k$ as in \eqref{eqn:w-bary}. For any $q\in (0,1)$, let $\hat {\nu}_{j}^q$ be the $q$th empirical quantile of $\nu_j$ based on the MCMC samples from $\nu_j$, and $\hat {\overline \nu}^q$ be the $q$th quantile of the empirical version of $\overline \nu$. Then, $\hat {\overline \nu}^q$ can be computed as
\begin{align}
  \label{eq:emp-quant}
  \hat {\overline \nu}^q = \frac{1}{k} \sum_{j=1}^k \hat {\nu}_{j}^q, \quad q = \xi, 2\xi, \ldots, 1 - \xi,
\end{align}
where $\xi$ is the grid-size of the quantiles \citep{Lietal16}. If the $\xi$-grid is fine enough in \eqref{eq:emp-quant}, then the parameter MCMC samples from the marginal DISK distribution are obtained by inverting the empirical distribution function supported on the quantile estimates.

In practice, the primary interest often lies in the marginal distributions of model parameters and predicted values; that is, the posterior of some one-dimensional functional $\theta$. Therefore, the univariate Wasserstein barycenter obtained by averaging quantiles in \eqref{eq:emp-quant} accomplishes this with great generality and convenient implementation. For this reason, in the following sections, we only focus on the case where $\theta$ is one-dimensional and use \eqref{eq:emp-quant} to compute the DISK posterior through its empirical quantiles. Nonetheless, we emphasize that the DISK posterior for a multivariate $\theta$ can still be efficiently computed using the sparse linear program for Wasserstein barycenters as described in \citet{CutDou14,Srietal15,Staetal17}; however,  these methods are computationally expensive and lead to only a small amount of the improvement over the univariate quantile combination in \eqref{eq:emp-quant} as revealed by our experiments.
}


A key feature of the DISK combination scheme is that given the subset posterior MCMC samples, the combination step is agnostic to the choice of a model. Specifically,  given MCMC samples from the $k$ subset posterior distributions,
\eqref{eq:emp-quant} remains the same for models based on a full-rank GP prior, a low-rank GP prior, such as MPP, or any other model described in Section \ref{disk-comp}. Since the averaging over $k$ subsets takes $O(k)$ flops and $k < n$, the total time for computing the empirical quantile estimates of the DISK posterior in inference or prediction requires $O(k) + O(m^3)$ and $O(k) + O({rm^2})$  flops in models based on full-rank and low-rank GP priors, {respectively}. Assuming that we have abundant computational resources, $k$ is chosen large enough so that $O(m^3)$ computations are feasible. This would enable applications of the DISK framework in models based on both full-rank and low-rank GP priors in massive $n$ settings.

\subsection{{Bayes $L_2$-risk of DISK: {convergence rates} and the choice of $k$}}
\label{sec:risk-fixed-n}

In the divide-and-conquer Bayesian setup, it is already known that when the data are independent and identically distributed (i.i.d.), the combined posterior distribution using the Wasserstein barycenter of subset posteriors approximates the full data posterior distribution at a near optimal parametric rate, under certain conditions as $n, k,  m_1, \ldots, m_k \rightarrow \infty$ \citep{Lietal16,Srietal17}; however, in models based on spatial process, data are not i.i.d. and inference on the infinite dimensional true spatial surface is of primary importance. Few formal theoretical results are available in this nonparametric divide-and-conquer Bayes setup. A notable exception is the recent paper \citep{SzaVan17}, which shows that combination using Wasserstein barycenter has optimal Bayes risk and adapts to the smoothness of $w_0(\cdot)$, the true but unknown $w(\cdot)$, in the Gaussian white noise model. The Gaussian white noise model is a special case of \eqref{parent_proc} with additional smoothness assumptions on $w_0(\cdot)$.

{We investigate the theoretical properties of the DISK predictive posterior of the mean surface $\xb(\cdot)^T \betab + w(\cdot)$. For ease of presentation, we assume that $m_1 = \cdots = m_k=m$. Determining the appropriate order for $k$ in terms of $n$ is one of the key issues for all divide-and-conquer statistical methods. If $k$ is too small, then the biases in subset posterior distributions are small due to a large subset size $m$, while the overall variance of the DISK posterior is large due to too few subsets. In contrast, if $k$ is very large, then the biases in subset posterior distributions are large due to a small subset size $m$, while the variance of the DISK posterior can be small due to the large number of subsets.
Ideally, $k$ should be controlled to increase in some order of $n$, such that the bias and the variance can be balanced and the Bayes $L_2$-risk of the DISK posterior can be minimized.
}

We formally explain the model setup for our theory development. Suppose that the data generation process follows the model \eqref{parent_proc} with the true parameter value $\Omegab_0=(\alphab_0,\betab_0)$ and the true spatial surface $w_0(\cdot)$. We focus on the Bayes $L_2$-risk of the DISK predictive posterior for the mean function in \eqref{parent_proc}; that is, $\xb(\sbb^*)^T \betab + w(\sbb^*)$ for any testing location $\sbb^* \in \Scal$. To ease the complexity of our theory, we first present two theorems below for the simplified model
\begin{align} \label{theory-mdl}
y(\sbb_i)= w(\sbb_i) + \epsilon(\sbb_i), \quad \epsilon(\sbb_i) \sim N(0, \tau^2), \quad i = 1, \ldots, n, \quad  w(\cdot) \sim \text{GP}\{0, C_{{\alphab}}(\cdot, \cdot) \}.
\end{align}
Compared to the spatial model \eqref{parent_proc}, the model \eqref{theory-mdl} does not contain the regression term $\xb(\sbb)^T \betab$; however, our theory includes this regression term later by modifying the covariance function; see our Corollary \ref{cor-splm} below. {Our theoretical setup is a general one that subsumes GP priors with Matern covariance kernel \citep{stein2012interpolation} and the wide class of low-rank GP priors. While $\tau^2$ and $\alphab$ are unknown in practice and are assigned prior distributions, our setup in \eqref{theory-mdl} assumes that $\tau^2$ and $\alphab$ are known because this is a common assumption in the asymptotic study of GP-based regression; see \citet{VarZan11} and the references therein. Furthermore, it is known that the theoretical results obtained by assuming $\tau^2$ and $\alphab$ as known constants are equivalent to those obtained by assigning priors with bounded supports on these parameters \citep{van2008rates,van2009adaptive}.} Generalization to unboundedly supported priors is well known to be a difficult in Bayesian GP theory, especially for $\alphab$, and only partial solutions are available \citep{van2009adaptive,yang2017supnorm}.

We introduce some definitions used in stating the results in this section. Let $\PP_{\sbb}$ be a design distribution of $\sbb$ over $\Dcal$, $L_2(\PP_{\sbb})$ be the $L_2$ space under $\PP_{\sbb}$, the inner product in $L_2(\PP_{\sbb})$ is defined as $\langle f, g\rangle_{L_2(\PP_{\sbb})} = \EE_{\PP_{\sbb}} (fg)$ for any $f, g \in L_2(\PP_{\sbb})$. {For any $f\in L_2(\PP_{\sbb})$ and $\sbb \in \Dcal$, define the linear operator $(T_{\alphab}f)(\sbb)=\int_{\Dcal} C_{\alphab}(\sbb,\sbb')f(\sbb')d\PP_{\sbb}(\sbb')$. According to the Mercer's theorem, there exists an orthonormal basis $\left\{\phi_i(\sbb)\right\}_{i=1}^\infty$ in $L_2(\PP_{\sbb})$, such that $C_{\alphab}(\sbb, \sbb') = \sum_{i=1}^{\infty} \mu_i \phi_i(\sbb) \phi_i(\sbb')$, where $\mu_1\geq \mu_2 \geq \ldots\geq 0$ are the eigenvalues and $\{\phi_i(\sbb)\}_{i=1}^\infty$ are the eigenfunctions of $T_{\alphab}$. }
The trace of the kernel $C_{\alphab}$ is defined as $\tr(C_{\alphab})=\sum_{i=1}^{\infty} \mu_i$.
Any $f\in L_2(\PP_{\sbb})$ has the series expansion $f(\sbb) = \sum_{i=1}^{\infty} \theta_i \phi_i(\sbb)$, where $\theta_i = \langle f,\phi_i\rangle_{L_2(\PP_{\sbb})} $. The reproducing kernel Hilbert space (RKHS) $\HH$ attached to $C_{\alphab}$ is the space of all functions $f \in L_2(\PP_{\sbb})$ such that the $\HH$-norm $\|f\|_{\HH} = \sum_{i=1}^{\infty} \theta_i^2 / \mu_i<\infty$. The RKHS $\HH$ is the completion of the linear space of functions defined as $\sum_{i=1}^I a_i C_{\alphab}(\sbb_i, \cdot)$, where $I$ is a positive integer, $\sbb_i \in \Dcal$, and  $ a_i \in \RR$ ($i=1, \ldots, I$); see \cite{VarZan08b} for more details on RKHS.

For theory development, we consider a random design scheme with the observed locations $\Scal=\left\{\sbb_1,\ldots,\sbb_n\right\}$ and $\Scal^*=\{\sbb^*\}$. We assume that the locations $\sbb_1, \ldots, \sbb_n, \sbb^*$ are independently drawn from the same sampling distribution $\PP_{\sbb}$. We further impose the following assumptions.

\begin{enumerate}[label=A.\arabic*]
\item \label{r1} (RKHS) The true function $w_0$ is an element of the RKHS $\HH$ attached to the kernel $C_{\alphab}$.
\item \label{r2} (Trace class kernel) $\tr(C_{\alphab})<\infty$.
\item \label{r3} (Moment condition) There are positive constants $\rho$ and, with a slight abuse of notation, $r \geq 2$ such that $\EE_{\PP_{\sbb}} \{\phi_i^{2r}(\sbb)\}  \leq \rho^{2r}$ for every $i$, and $\text{var} \left\{\epsilon(\sbb)\right\} = \tau^2 <\infty$ for any $\sbb\in \Dcal$.
\end{enumerate}
In general, the RKHS $\HH$ can be a much smaller space relative to the support of the GP prior. While we use $w_0\in \HH$ in Assumption \ref{r1} mainly for technical simplicity, this assumption can be possibly relaxed by considering sieves with increasing $\HH$-norms; {in the same vein as} Assumption B$'$ and Theorem 2 in \citet{Zhaetal15}. We expect that similar convergence rate results to our Theorems \ref{bayes-risk} and \ref{bv3cases} can be shown for much larger classes of functions than $\HH$; see the additional posterior convergence theory in Section 2 of supplementary material. In Assumption \ref{r2}, $\tr(C_{\alphab})$ measures the size of the covariance kernel and imposes conditions on the regularity of functions that DISK can learn. Assumption \ref{r3} controls the error in approximating $C_{\alphab}(\sbb, \sbb')$ by a finite sum, and the superscript $r$ here should not be confused with the number of knots in low-rank GP priors. Our results are valid for any error distribution with a finite variance, and it holds trivially for the model in \eqref{theory-mdl}.

We begin by examining the Bayes $L_2$-risk of the DISK posterior for estimating $w_0$ in \eqref{theory-mdl}. Under the setup of \eqref{theory-mdl}, let $\EE_{\sbb^*}$, $\EE_{0}$, $\EE_{\Scal}$, and $\EE_{0 \mid \Scal}$  respectively be the expectations with respect to the distributions of $\sbb^*$, $(\Scal, \yb)$, $\Scal$, and $\yb$ given $\Scal$. If $\overline w(\sbb^*)$ is a random variable that follows the DISK posterior for estimating $w_0(\sbb^*)$, then $\overline w(\sbb^*)$ has the density $N(\overline m, \overline v)$, where
\begin{align}
  \label{eq:e3}
  \overline m = \frac{1}{k} \sum_{j=1}^k \cb_{j,*}^T (\Cb_{j, j} + \tfrac{\tau^2}{k} \Ib)^{-1}\yb_j , \;
  \overline v^{1/2} = \frac{1}{k}\sum_{j=1}^k v_j^{1/2}, \;
  v_j = c_{*,*} - \cb_{j,*}^T (\Cb_{j,j} + \tfrac{\tau^2}{k} \Ib)^{-1} \cb_{j, *},
\end{align}
$c_{*,*} = \cov\{w(\sbb^*), w(\sbb^*)\}$, and $\cb_{j,*}^T = [\cov\{w(\sbb_{j1}), w(\sbb^*)\}, \ldots, \cov\{w(\sbb_{jm}), w(\sbb^*)\}]$. The Bayes $L_2$-risk of the DISK posterior
in estimating $w_0$ is $\EE_{0} \left[ \EE_{\sbb^*} \{\overline w(\sbb^*) - w_0(\sbb^*)\}^2 \right]$. {This risk can be used to quantify how quickly the DISK posterior concentrates around the unknown true surface $w_0(\cdot)$ as the total sample size $n$ increases to infinity.} It is straightforward to show (see the proof of Theorem \ref{bayes-risk} in the Supplementary Material) that this Bayes $L_2$-risk can be decomposed into the squared bias, the variance of subset posterior means, and the variance of DISK posterior terms as
{\begin{align}
&\text{bias}^2 =  \EE_{\sbb^*} \EE_{\Scal} \{  \cb_{*}^T(k \Lb + \tau^2 \Ib)^{-1}\wb_0 - w_0(\sbb^*)\}^2, \quad \nonumber \\
&\text{var}_{\text{mean}} =  \tau^2 \EE_{\sbb^*} \EE_{\Scal} \left\{\cb^T_* (k \Lb + \tau^2 \Ib)^{-2} \cb_* \right\}, \quad \nonumber \\
&\text{var}_{\text{DISK}} =  \EE_{\sbb^*}\EE_{\Scal}  (\overline v(\sbb^*)), \label{asymp-mse}
\end{align}}
where  $\overline v(\sbb^*) = \EE_{0 \mid \Scal} \left[ \var \{\overline w(\sbb^*) \mid \yb \} \right]$, $\cb^T_* = (\cb^T_{1,*}, \ldots, \cb^T_{k,*})$,  $\wb_{0j} = \{w_0(\sbb_{j1}), \ldots, w_0(\sbb_{jk})\}$ ($j=1, \ldots, k$), $\wb_0^T = (\wb_{01}, \ldots, \wb_{0k})$, and $\Lb$ is a block-diagonal matrix with $\Cb_{1,1}, \ldots, \Cb_{k,k}$ along the diagonal. The next theorem provides theoretical upper bounds for each of the three terms in \eqref{asymp-mse}.
\begin{theorem}\label{bayes-risk}
If Assumptions \ref{r1}--\ref{r3} hold, then
\begin{align}
  \label{eq:bayesrisk-ineq}
  \text{Bayes } L_2\text{ risk} &= \EE_{\Scal} \EE_{0|\Scal} \EE_{\sbb^*} \{\overline w(\sbb_*) - w_0(\sbb_*)\}^2 =  \text{bias}^2 + \text{var}_{\text{mean}} + \text{var}_{\text{DISK}} ,\nonumber \\
  \text{bias}^2 & \leq \frac{8\tau^2}{n} \| w_0 \|_{\HH}^2 + \| w_0 \|_{\HH}^2 \; \underset{d \in \NN}{\inf} \, \left[ \frac{8n}{\tau^2} \rho^4 \tr(C_{\alphab}) \tr(C_{\alphab}^{d}) + \tr(C_{\alphab}) R(m,n,d,r) \right],\nonumber \\
  \text{var}_{\text{mean}} &\leq \frac{2 n + 4 \| w_0 \|^2_{\HH} }{k} \inf_{d\in \NN} \left[\left\{ \frac{12n}{\tau^2} \rho^4 \tr(C_{\alphab}) +1\right\}\tr(C_{\alphab}^d)  + R(m,n,d,r) \right]  \nonumber \\
  &~~~ +
                  12 \frac{\tau^2}{n} \gamma \left( \frac{\tau^2}{n} \right) +   \frac{12}{k} \frac{\tau^2}{n} \| w_0 \|_{\HH}^2 ,  \nonumber \\
  \text{var}_{\text{DISK}} &\leq 3 \frac{\tau^2}{n} \gamma \left( \frac{\tau^2}{n} \right) + \underset{d \in \NN}{\inf} \,\left[ \left\{ \frac{4n}{\tau^2}\tr(C_{\alphab}) + 1 \right\} \tr(C_{\alphab}^d) + \tr(C_{\alphab}) R(m,n,d,r) \right],
\end{align}
where $\NN$ is the set of all positive integers, $A$ is a global positive constant that does not depend on any of the quantities here, and
\begin{align*}
& b(m, d, r) = \max \left( \sqrt{\max(r, \log d)}, \frac{\max(r,\log d)}{m^{1/2 - 1/r}} \right), \\
& R(m,n,d,r)= \left\{  \frac{A\rho^2  b(m, d, r) \gamma(\tfrac{\tau^2}{n})}{\sqrt{m}} \right\}^r,\\
& \gamma(a) = \sum_{i=1}^{\infty} \frac{\mu_i}{\mu_i+a} \text{ for any } a>0, \quad \tr(C^d_{\alphab}) = \sum_{i=d+1}^{\infty} \mu_i.
\end{align*}
\end{theorem}


These bounds are similar to the bounds obtained in Theorem 1 of \citet{Zhaetal15} for the frequentist divide-and-conquer estimator in kernel ridge regression. Although the upper bounds in \eqref{eq:bayesrisk-ineq} appear very complicated and involve many terms, the dominant term among them is $\frac{\tau^2}{n} \gamma \left( \frac{\tau^2}{n} \right)$, where the function $\gamma(\cdot)$ is related to the ``effective dimensionality" of the covariance kernel $C_{\alphab}$ \citep{Zha05}. This term determines how fast the Bayes $L_2$-risk converges to zero, as long as $k$ is chosen to be some proper order of $n$ such that all the other terms in the upper bounds of \eqref{eq:bayesrisk-ineq} can be made negligible compared to $\frac{\tau^2}{n} \gamma \left( \frac{\tau^2}{n} \right)$.

In contrast to the frequentist literature such as \citet{Zhaetal15}, a significant difference in our Theorem \ref{bayes-risk} is that our risk bounds involve two different variance terms. While our analysis naturally introduces the variance term $\text{var}_{\text{DISK}}$ that corresponds to the variance of the DISK posterior distribution, any frequentist kernel regression method only finds a point estimate of $w_0$ and thus does not include this variance term.  As a by-product of the proof of Theorem \ref{bayes-risk}, the upper bound for $\text{var}_{\text{DISK}}$ can be used to show that the integrated predictive variance of GP decreases to zero as the subset sample size $m\to\infty$ for various types of covariance kernels. A closely related work by \citet{GraGar15} studies the asymptotic behavior for the mean squared error of GP, but unrealistically assumes that the error variance increases with the sample size, which prevents their predictive variance of GP from converging to zero.



Each of the three upper bounds in Theorem \ref{bayes-risk} can be made close to zero as $n$ increases to $\infty$ and $k$ is chosen to appropriately depend on $n$. {The next theorem finds the appropriate order for $k$ in terms of $n$, such that the DISK posterior achieves nearly minimax optimal rates in its Bayes $L_2$-risk \eqref{eq:bayesrisk-ineq}, for three types of commonly used covariance kernels, (i) degenerate covariance kernels, (ii) covariance kernels with exponentially decaying eigenvalues, and (iii) covariance kernels with polynomially decaying eigenvalues. }
The covariance kernel $C_{\alphab}$ is a degenerate kernel of rank $d^*$ if there is some constant positive integer $d^*$ such that $\mu_1\geq \mu_2 \geq \ldots \geq \mu_{d^*}>0$ and $\mu_{d^*+1}=\mu_{d^*+2}=\ldots=\mu_{\infty}=0$. The covariance kernels in subset of regressors approximation \citep{QuiRas05} and predictive process \citep{Banetal08} are degenerate with their ranks equaling the number of ``inducing points'' and knots, respectively. The squared exponential kernel is very popular in machine learning. Its RKHS belongs to the class of RKHSs of kernels with exponentially decaying eigenvalues. Similarly, the class of RKHSs of kernels with polynomially decaying eigenvalues includes the Sobolev spaces with different orders of smoothness and the RKHS of the Mat\'ern kernel. {This kernel is most relevant for spatial applications, but we provide the other two results for a more general audience.}
\begin{theorem}\label{bv3cases}
If Assumptions \ref{r1}--\ref{r3} hold and $r>4$ in Assumption \ref{r3}, then, as $n \rightarrow \infty$,
\begin{itemize}
\item[(i)] if $C_{\alphab}$ is a degenerate kernel of rank $d^*$ and $k \leq c n^{\frac{r-4}{r-2}}/(\log n)^{\frac{2r}{r - 2}}$ for some constant $c>0$, then the Bayes $L_2$-risk of DISK posterior satisfies $\EE_{\sbb^*} \EE_{\Scal} \EE_{0|\Scal} \{\overline w(\sbb^*) - w_0(\sbb^*)\}^2 = O\left( n^{-1}\right)$;
\item[(ii)] if $\mu_i \leq c_{1\mu} \exp\left(-c_{2\mu} i^\kappa \right)$ for some constants $c_{1\mu}>0, c_{2\mu}>0, \kappa>0$ and all $i\in \NN$, and for some constant $c>0$, $k\leq c n^{\frac{r-4}{r-2}}/(\log n)^{\frac{2(r\kappa+r-1)}{\kappa(r-2)}}$, then the Bayes $L_2$-risk of DISK posterior satisfies $\EE_{\sbb^*} \EE_{\Scal} \EE_{0|\Scal} \{\overline w(\sbb^*) - w_0(\sbb^*)\}^2 = O\left\{ (\log n)^{1/\kappa} /n \right\}$; and
\item[(iii)] if $\mu_i \leq c_{\mu} i^{-2\nu}$ for some constants $c_{\mu}>0,\nu>\tfrac{r-1}{r-4}$ and all $i\in \NN$, and for some constant $c>0$, $k\leq c n^{\frac{(r-4)\nu-(r-1)}{(r-2)\nu}} / (\log n)^{\frac{2r}{r-2}}$, then the Bayes $L_2$-risk of DISK posterior satisfies $\EE_{\sbb^*} \EE_{\Scal} \EE_{0|\Scal} \{\overline w(\sbb^*) - w_0(\sbb^*)\}^2= O\left( n^{-\frac{2\nu-1}{2\nu}}\right)$.
\end{itemize}
\end{theorem}
The rate of decay of the $L_2$-risks in (i) and (ii) with $\kappa=2$ are known to be minimax optimal \citep{Rasetal12,Yanetal17a}, whereas the rate of decay of the $L_2$-risk in (iii) is slightly larger than the minimax optimal rate by a factor of $n^{\frac{1}{2 \nu (2 \nu + 1)}}$. Since $\Dcal$ is compact in all spatial applications, $r$ in Assumption \ref{r3} can be chosen as large as possible. Letting $r\rightarrow \infty$, the upper bounds on $k$ in (i), (ii), and (iii) reduce to $k = O\{n / (\log n)^{2} \} $, $k = O\{n / (\log n)^{2/\kappa}\} $, and $k= O\{n^{\frac{\nu - 1}{\nu}} / (\log n)^2\}$, respectively.


Now we generalize the results in Theorems \ref{bayes-risk} and \ref{bv3cases} to the model \eqref{parent_proc}. Besides \ref{r1}--\ref{r3}, we further make the following assumption on $\xb(\cdot)$ and the prior on $\betab$:
\begin{enumerate}[label=B.\arabic*]
\item \label{sp} All $p$ components of $\xb(\cdot)$ are non-random functions in $\Scal$. The prior on $\betab$ is $N(\mub_{\betab}, \Sigmab_{\betab})$ and it is independent of the prior on $w(\cdot)$, which is $\text{GP}\{0,C_{\alphab}(\cdot,\cdot)\}$.
\end{enumerate}
By the normality and joint independence in Assumption \ref{sp}, it is straightforward to show that the mean function $\xb(\sbb)^T \betab + w(\sbb)$ has a GP prior $\text{GP}\{\xb(\cdot)^T \mub_{\betab}, \check C_{\alphab}(\cdot,\cdot)\}$, where the modified covariance kernel $\check C_{\alphab}$ is given by
\begin{align} \label{modified-kernel}
\check C_{\alphab}(\sbb_1, \sbb_2) &= \cov \left\{\xb(\sbb_1)^T \betab + w(\sbb_1), ~ \xb(\sbb_2)^T \betab + w(\sbb_2) \right\} \nonumber \\
&= \xb (\sbb_1)^T \Sigmab_{\betab}  \xb (\sbb_2) + C_{\alphab}(\sbb_1, \sbb_2),
\end{align}
for any $\sbb_1,\sbb_2\in \Scal$. With this modified covariance kernel, we have the following corollary:
\begin{corollary}\label{cor-splm}
If Assumption \ref{sp} holds, and Assumptions \ref{r1}--\ref{r3} also hold with all $C_{\alphab}$ replaced by $\check C_{\alphab}$ defined in \eqref{modified-kernel}, then the conclusions of Theorems \ref{bayes-risk} and \ref{bv3cases} hold for the Bayes $L_2$-risk of the mean surface $\xb(\cdot)^T \betab + w(\cdot)$ in the model \eqref{parent_proc}.
\end{corollary}

\subsection{Inference for spatial range autocorrelation}
\label{sec:spatialcorrelation}
 
Besides the estimation of mean surface $w_0(\cdot)$, it is also important to have valid Bayesian inference for the spatial correlation. We demonstrate that the DISK estimate of the correlation structure is close to the true posterior correlation structure in the simulation example of Section \ref{sec:simulation-2-GP}; see Figure \ref{fig:cov_fun}. In the following, we describe some heuristics; a detailed study is left for future research. To explain why our combination scheme preserves the spatial range autocorrelation, let us consider a special case of model \eqref{parent_proc} where $\betab=0$, $\sigma^2$ and $\tau^2$ are known and fixed, and only the spatial range parameter $\phi$ is unknown and needs to be estimated. In this case, $\phi$ completely determines the correlation structure; therefore, for each data subset, the log likelihood function for the model in \eqref{parent_proc} can be written as
\begin{align} \label{eq:prop1}
    \log L_j(\phi) = k \log \ell_j(\phi) &= -\frac{km}{2} \log 2\pi - \frac{k}{2} \log \vert \Rb_j \vert - \frac{k}{2} \tr \left( \Rb_j^{-1} \yb_j \yb_j^T \right),
\end{align}
where $\ell_j(\phi)$ is the likelihood function without stochastic approximation, $\Rb_j$ is the shorthand for $\Rb_j(\phi)= \Cb_j(\phi) + \tau^2\Ib$ and $\Cb_j(\phi)$ is $\Cb_j(\alphab)$ with a fixed and known $\sigma^2$. If $\sigma^2$ is fixed at the true value, $\tau^2=0$, and $\phi$ is assigned a prior on the closed interval $[\phi_l, \phi_u]$ ($0<\phi_l<\phi_u<\infty$) that includes the true value $\phi_0$, then the maximum likelihood estimator (MLE) of $\phi$ over $\phi\in [\phi_l, \phi_u]$ is consistent and asymptotically normal for some commonly used covariance functions such as Mat\'ern (\citealt{KauSha13}). Motivated by the existing frequentist theory on MLE, we conjecture that a Bernstein-von Mises (BvM) theorem under fixed $\sigma^2$ and $\tau^2$ possibly holds for $\phi$. Let $\Rb'_j(\phi) = \partial \Rb_j (\phi) / \partial \phi$ and $\Rb''_j(\phi) = \partial^2 \Rb_j(\phi) / \partial \phi^2$ be the matrices consisting of component-wise derivatives with respect to $\phi$. Let $\PP_{\phi_0}$ be the probability measure of the subset data $\yb_j$ for all $j=1,\ldots,k$ and also the full data $\yb$.

From \eqref{eq:prop1}, we can derive the derivatives of log-likelihood function:
\begin{align}
    \label{eq:prop2}
    \frac{d \log L_j(\phi)}{d \phi} &= k\frac{d \log \ell_j(\phi)}{d \phi}  = - \frac{k}{2} \tr \left( \Rb_j^{-1} \Rb_j' \right) + \frac{k}{2} \tr(\Rb_j^{-1} \Rb_j' \Rb^{-1}_j \yb_j \yb_j^T ),
    \nonumber \\
    \frac{d^2 \log L_j(\phi)}{d\phi^2} &= k\frac{d^2 \log \ell_j(\phi)}{d \phi^2} =  - \frac{k}{2} \tr \left( \Rb_j^{-1} \Rb_j'' - \Rb_j^{-1} \Rb_j' \Rb_j^{-1} \Rb_j'   \right)  \nonumber \\
   &\qquad \qquad \qquad ~~ - \frac{k}{2} \tr \{(2 \Rb^{-1}_j \Rb_j' \Rb_j^{-1} \Rb_j' \Rb^{-1}_j - \Rb_j^{-1} \Rb_j'' \Rb^{-1}_j ) \yb_j \yb_j^T\}.
\end{align}
At the true parameter $\phi=\phi_0$, the second derivative $d^2 \log \ell_j(\phi)/d \phi^2$ simplifies to $-\frac{1}{2}\tr \left\{ \Rb_j^{-1}(\phi_0) \Rb_j^{'}(\phi_0) \Rb_j^{-1}(\phi_0)\Rb_j^{'}(\phi_0)  \right\}$ because $\EE_{\phi_0}(\yb_j\yb_j^T)=\Rb_j$. In the special case where $\tau^2=0$ and $C_{\alphab}(\cdot,\cdot)$ is the Mat\'ern covariance function with smoothness parameter $\nu$, one can derive from Theorem 2 of \citet{KauSha13} that the asymptotic variance of the MLE of $\phi$ is $\phi_0^2/(2\nu^2)$. This implies that for the Mat\'ern covariance function,
$$\lim_{m\to\infty} \frac{1}{2m}\tr \left\{ \Rb_j^{-1}(\phi_0) \Rb_j^{'}(\phi_0) \Rb_j^{-1}(\phi_0)\Rb_j^{'}(\phi_0)  \right\} = 2\nu^2/\phi_0^2.$$

Let $\Pi_{\phi,j}$ ($j=1,\ldots,k$) be the subset posteriors of $\phi$ after stochastic approximation. Let $\overline \Pi_{\phi,n}$ be the DISK posterior of $\phi$, which is the Wasserstein barycenter of $\Pi_{\phi,j}$'s and can be obtained by averaging the quantiles of $\Pi_{\phi,j}$'s over $j=1,\ldots,k$ (similar to $\overline \nu$ in \eqref{eq:emp-quant}). Let $\Pi_{\phi,n}$ be the full data posterior of $\phi$. Then, for the special case where the covariance function is Mat\'ern and $\tau^2=0$, based on Theorem 2 of \citet{KauSha13}, we conjecture that under certain regularity conditions, one can follow the techniques used in \citet{Lietal16} and show that both $\overline \Pi_{\phi,n}$ and $\Pi_{\phi,n}$ follow the BvM theorem and are asymptotically close to two normal distributions, each with the variance $2\nu^2/\phi_0^2$ but possibly different means. We conjecture that a similar result to \citet{Lietal16} holds, in the sense that the means between the two normal limits are of the order $O_p(1/\sqrt{m})$, and hence $m^{1/2}\cdot W_2\left(\overline \Pi_{\phi,n}, \Pi_{\phi,n}\right)\to 0 $ as $m\to\infty$, in the $\PP_{\phi_0}$-probability if the data $\yb$ has the probability measure $\PP_{\phi_0}$. Furthermore, based on Theorem 3 of \citet{KauSha13}, we conjecture that with high $\PP_{\phi_0}$-probability, for any $\phi$ drawn from either the true posterior $\Pi_{\phi,n}$ or the DISK posterior $\overline \Pi_{\phi,n}$, the Gaussian process predictive variance at any testing location $\sbb^*\in \Scal$ is asymptotically close to the true Gaussian process predictive variance with the true range parameter $\phi_0$, in the sense that the ratio of these two predictive variances tends to 1 as $m\to\infty$. If $\sigma^2$ is unknown, it is well known that $\sigma^2$ and $\phi$ cannot be identified \citep{Zha04} but a microergodic parameter like $\sigma^2\phi^{2\nu}$ for the Mat\'ern covariance functions can be identified. In this case, we conjecture that a similar BvM theorem can be shown for the microergodic parameter and similar conclusions about the equivalence of prediction variances can be obtained. We leave the thorough investigation on these issues for general covariance functions to future research.

\section{Experiments}
\label{experiments}

\subsection{Simulation setup}

We compare DISK with its competitors on synthetic data based on its performance in learning the process parameters, interpolating the unobserved  spatial surface, and predicting at new locations. This section presents three simulation studies. The first (\emph{Simulation 1}) and second  (\emph{Simulation 2}) simulations represent moderately large dataset with $12,025$ locations and the third (\emph{Simulation 3}) simulation  analyzes a large dataset with $1,002,025$ locations. In first two simulations, we randomly select the data at $n=10^4$ locations for model fitting and $l=2025$ locations for predictions, while in Simulation 3 the training and test data are of size $n= 10^6$ and $l = 2025$.

For all three simulations, the response is simulated at $(n+l)$ locations as
\begin{align}
  \label{eq:sim1}
  y(\sbb_i) = \beta_0 + w_0(\sbb_i) + \epsilon_i, \quad \epsilon_i \sim N(0, \tau_0^2), \quad i = 1, \ldots, n+l.
\end{align}
Simulations 1 and 3 follow the data generation scheme described in \cite{gramacy2015local}. Specifically, we set $\Dcal = [-2, 2] \times [-2, 2] \subset \RR^2$ and uniformly sample $(n+l)$ spatial locations $\sbb_i = (s_{i1}, s_{i2})$ in $\Dcal$ ($i = 1, \ldots, n+l$). For any $s \in [-2, 2]$,  define the function $f_0(s) = e^{-(s-1)^2} + e^{-0.8(s+1)^2} - 0.05 \sin \{8(s + 0.1)\}$ and set $w_0(\sbb_i)= - f_0(s_{i1})  f_0(s_{i2})$. Although the function $w_0(\cdot)$ simulated in this way is theoretically infinitely smooth,
the response surface simulated from (\ref{eq:sim1}) exhibits complex local behavior, which is challenging to capture using spatial process-based models as we demonstrate later. Simulation 2 generates $\{w_0(\sbb_1),\ldots,w_0(\sbb_{n+l})\}$ from a GP(0,{$\sigma_0^2\exp\{-\phi_0 \|\sbb_i-\sbb_j\|)\}$}. This is a more familiar simulation scenario in the spatial context with the generated spatial surface being continuous everywhere but differentiable nowhere. As argued earlier, simulating data from an ordinary Gaussian process is expensive, hence we refrain from a massive simulation study under the scenario of \emph{Simulation 2}. {For  Simulations 1 and 3, the intercept $\beta_0$ and true error variance $\tau_0^2$ are set to $1$ and $0.01$, respectively.} For \emph{Simulation 2}, $\beta_0$, $\tau_0^2$, $\phi_0$ and $\sigma_0^2$ are taken to be $1$, $0.1$, $9$ and $1$ to keep long range spatial dependence and high spatial variance to nugget ratio. The inferential and predictive results for all simulations are based on $10$ replications. We compare DISK with a number of Bayesian and non-Bayesian spatial models in both simulations:
(i) LatticeKrig \citep{Nycetal15} using the \texttt{LatticeKrig} package in {R} \citep{R17}  with 3 resolutions \citep{Nycetal16};
(ii) nearest neighbor Gaussian process (NNGP) using the \texttt{spNNGP} package in {R} with the number of nearest neighbors (NN) as 5, 15, and 25 \citep{Datetal15};
(iii) full-rank Gaussian process (GP) using the \texttt{spBayes} package in {R} \citep{Finetal15} with the full data;
(iv) modified predictive process (MPP) using the \texttt{spBayes} package in {R} with the full data; and
(v) locally approximated Gaussian process (laGP) using the \texttt{laGP} package in {R} \citep{gramacy2015local}.

All five methods produce results in Simulation 1 and 2, but the first four methods fail due to numerical issues in Simulation 3. While laGP is not designed for full scale Bayesian inference, it is used as the benchmark for predictive point estimates and associated standard errors in \eqref{eq:sim1} due to its popularity in fitting computer models. For the spatial process-based Bayesian models, we employ \eqref{parent_proc} with only an intercept $\beta$, putting a $N(0, 100)$ prior on $\beta$, a GP prior on $w(\cdot)$, and $\text{IG}(2, 0.1)$ prior on $\tau^2$, where IG$(a, b)$  is the Inverse-Gamma distribution with mean $b/(a-1)$ and variance $b / \{(a-1)^2 (a-2)\}$ for $a > 2$. In model fitting, we assume an exponential correlation in the random field given by $\cov\{w(\sbb),w(\sbb')\} = \sigma^2 e^{-\phi \| \sbb-\sbb' \|}, \; \sbb, \sbb' \in \Dcal$ and put IG(2, 2) prior on $\sigma^2$ and a uniform prior on $\phi$. The MPP prior on $w(\cdot)$ is fitted by setting the rank $r$ as 200 and 400, respectively, where the $r$ knots are selected randomly from the domain $\Dcal$. NNGP is chosen as a representative example of the current state-of-the-art Bayesian method for inference and predictions in massive spatial data.


{The three-step DISK framework is applied using the full-rank GP and the low-rank MPP priors using the algorithm outlined in {Section \ref{combinescheme}}, yielding DISK (GP) and DISK (MPP) procedures, respectively. {For all our simulations, DISK combines the subset marginal posteriors by averaging their quantiles, as described in Section \ref{combinescheme}, and we set $\xi = 10^{-4}$ in Equation \eqref{eq:emp-quant}.} We use consensus Monte Carlo (CMC; \cite{Scoetal16}), semiparametric density product (SDP; \cite{NeiWanXin13}) and meta kriging (MK; \cite{guhaniyogi2017meta}) as representative  competitors for model-free subset posterior aggregation to highlight the advantages of DISK. Similar to DISK, these three approaches also operate in three steps. In {steps} 1 and 2, the MPP-based model in \eqref{eq:mdl-j} is fitted on every subset for CMC, SDP and MK using the \texttt{spBayes} package. {Unlike DISK, the methods of CMC, SDP and MK do not employ stochastic approximation in the sampling step}. Third, we use \texttt{parallelMCMCcombine} package with the default setting \citep{MirCon14} for combining subset posterior MCMC samples in CMC and SDP, yielding CMC(MPP) and SDP (MPP) procedures respectively and \texttt{Mposterior} package for combining subset inferences in MK to yield MK (MPP). SDP (MPP) fails due to numerical issues when the posterior MCMC samples for predictions and the surface are combined. Further, Simulation 1 shows less than competitive performance for MK (MPP) in estimating the spatial surface and posterior distribution of the parameters, hence Simulations 2 and 3 do not use it as a competitor. Identical priors, covariance functions, ranks, and knots are used for the non-distributed process models and their distributed counterparts for a fair comparison.}

All experiments are run on an Oracle Grid Engine cluster with 2.6GHz 16 core compute nodes. The non-distributed methods (LatticeKrig, GP, MPP, NNGP, and laGP) and the distributed methods (CMC, DISK, {MK}, and SDP) are allotted memory resources of 64GB and 16GB, respectively. Every MCMC sampling algorithm runs for 15,000 iterations, out of which the first 10,000 MCMC samples are discarded as burn-in MCMC samples and the rest of the chain is thinned by collecting every fifth MCMC sample. Convergence of the chains to their stationary distributions is confirmed using trace plots. All the interpolated spatial surfaces are obtained using the \texttt{MBA} package in \texttt{R}.

We compare the quality of prediction and estimation of  spatial surface at predictive locations $\Scal^*=\{\sbb_1^*,\ldots,\sbb_l^*\}$. If $w_0(\sbb_{i'}^*)$ is the value of the spatial surface at $\sbb_{i'}^* \in \Scal^*$, the estimates of bias, variance, and Bayes $L_2$-risk in estimating $w_0(\cdot)$ are defined as
\begin{align}
  \label{eq:sim2}
  {\text{bias}}^2 = \frac{1}{l} \sum_{i'=1}^l \{\hat w(\sbb_{i'}^*) - w_0(\sbb_{i'}^*)\}^2, ~~ {\text{var}} = \frac{1}{l} \sum_{i'=1}^l \hat {\text{var}} \{w(\sbb_{i'}^*)\}, ~~ {L_2\text{-risk}} = {\text{bias}}^2 +  {\text{var}},
\end{align}
where $\hat w(\sbb_{i'}^*)$ and $\hat {\text{var}} \{w(\sbb_{i'}^*)\}$ denote the estimate of $w_0(\sbb_{i'}^*)$  obtained using any distributed or non-distributed methods and its variance, respectively. For sampling-based methods, we set $\hat w(\sbb_{i'}^*)$ and $\hat {\text{var}} \{w(\sbb_{i'}^*)\}$ to be the median and the variance of the posterior MCMC samples for $ w(\sbb_{i'}^*)$, respectively, for $i'=1, \ldots, l$. We also estimate the point-wise 95\% credible or confidence intervals (CIs) of $w(\sbb^*_{i'})$  and predictive intervals (PIs) of $y(\sbb^*_{i'})$ for every $\sbb_{i'} \in \Scal^*$ and compare the CI and PI coverages and lengths for every method.  Also, the point predictive performance at the locations in  $\Scal^*$ are compared across competitors using the mean square prediction error (MSPE) defined as MSPE$= \sum_{i'=1}^l \{\hat y(\sbb_{i'}^*) - y(\sbb_{i'}^*)\}^2 / l$. Finally, we compare the performance of all the methods for parameter estimation using the posterior medians or point estimates and the 95\% CIs for $\beta$, $\sigma^2$, $\tau^2$, and $\phi$.

\subsection{Simulated data analysis}
\label{sim}

\subsubsection{Simulation 1: moderately large spatial data.}
\label{sec:simulation-1}

{We fit DISK (GP) for $k = 10, 20, 30, 40, 50$; CMC (MPP), SDP (MPP), MK (MPP) and DISK (MPP) for $k=10, 20$, along with other competitors. MPP using DISK, CMC, MK and SDP are fitted with the number of subsets below $20$ to ensure that the sample size in each subset $m$ is bigger than the number of knots $r$.} Focusing on the estimation of $w_0(\sbb^*)$ for $\sbb^* \in \Scal^*$, CMC and DISK have smaller biases and larger variances than their non-distributed counterparts if $k \leq 30$ (Table \ref{inf-w-e4}). The DISK estimator's variance decreases and bias increases with increasing $k$, resulting in a decreasing $L_{2}-$risk in the estimation of $w_0$ initially and increasing after $k = 20$ (Figure \ref{fig:bias-var-mse}), which empirically verifies Bayesian bias-variance trade-off revealed in {our theory}. The point-wise coverage of 95\% CIs in DISK are similar to that of the non-distributed methods, except MPP and NNGP, for $k \leq 30$ and are above the nominal value for all $k$. On the other hand, coverage of CMC and MPP are below the nominal value for every $k$ and $r$ and NNGP fails to cover $w_0$ across all replications. The length of 95\% CIs in DISK (GP) and DISK (MPP) are very close to that of their non-distributed version, whereas CMC's and NNGP's CIs greatly underestimate the posterior uncertainty. DISK (GP) and DISK (MPP) with $k=20$ are {among} the best performers, while MK (MPP) exhibits higher Bayes $L_2$-risk and wider 95\% CIs' with different choices of $r$ and $k$. {Since estimates of $w_0(\sbb^*)$ are not directly obtained for laGP or LatticeKrig from the \texttt{laGP} and \texttt{LatticeKrig} packages respectively, we subtract the estimated fixed effects from the predicted values to provide a rough estimate of $w_0(\sbb^*)$ for both of them. While such a strategy yields reasonable point estimates of $w_0(\sbb^*)$ from both of them, characterization of uncertainty of $w_0(\sbb^*)$ is perhaps less unreliable. Hence we refrain from interpreting coverage and length of laGP and LatticeKrig in Table \ref{inf-w-e4} any further, rather investigate the predictive performance of these methods vis-a-vis other competitors.} {All methods perform well in terms of MSPE and coverages and lengths of 95\% PIs (Table \ref{inf-p-e4a}). DISK (GP) and DISK (MPP) (with $k=10,20$) are more precise in the estimation of $\betab$ compared to their non-distributed versions. The most closely related competitor MK (MPP) shows competitive performance in terms of prediction, but much wider credible intervals for parameters, perhaps due to not accounting for stochastic approximation in subsets. Other competitors exhibit a similar degree of accuracy in estimating $\tau^2$, $\sigma^2$, and $\phi$ with the accuracy of parameter estimates and uncertainty quantification suffering beyond $k\geq 30$.}


An interesting feature of our comparisons is the dramatic difference between the performance of DISK (MPP) and MPP with the same choices of knots (Figure \ref{fig:disksurf10k-sim1}). The performance of MPP in spatial surface and parameter estimation  using full data is sensitive to the choice of $r$, suffering greatly when $r=200$. Contrary to this, DISK (MPP) has substantially smaller $L_2$-risk with the same number of knots used in each subset and its  $w_0$ estimate is almost indistinguishable from the true spatial surface. The performance of MPP is unstable when $r/n$ is low due to poorly conditioned covariance matrix. While the poor performance of full data MPP is attributed to this fact, using $r$ knots in \emph{each} subset of size $m$ results in relatively high $r/m$ ratio in each subset of DISK (MPP), which contributes to its strikingly superior performance.

While running NNGP with the \texttt{spNNGP} package, the inference from NNGP marginalized over $w(\sbb)$'s using the ``response'' option closely matches DISK in terms of inference on parameter estimates; however, $\wb^*$ cannot be estimated with ``response'' option, so we employ un-marginalized NNGP using the ``sequential'' option in the package. This results in severe auto-correlation among the latent variables and $\beta$, yielding NNGP's poor performance across all three choices of the number of nearest neighbors. In an ongoing work, this issue is addressed by employing conjugate gradient algorithms to estimate latent variables in NNGP.

\begin{figure}[h]
 \centering
 \includegraphics[width=\textwidth]{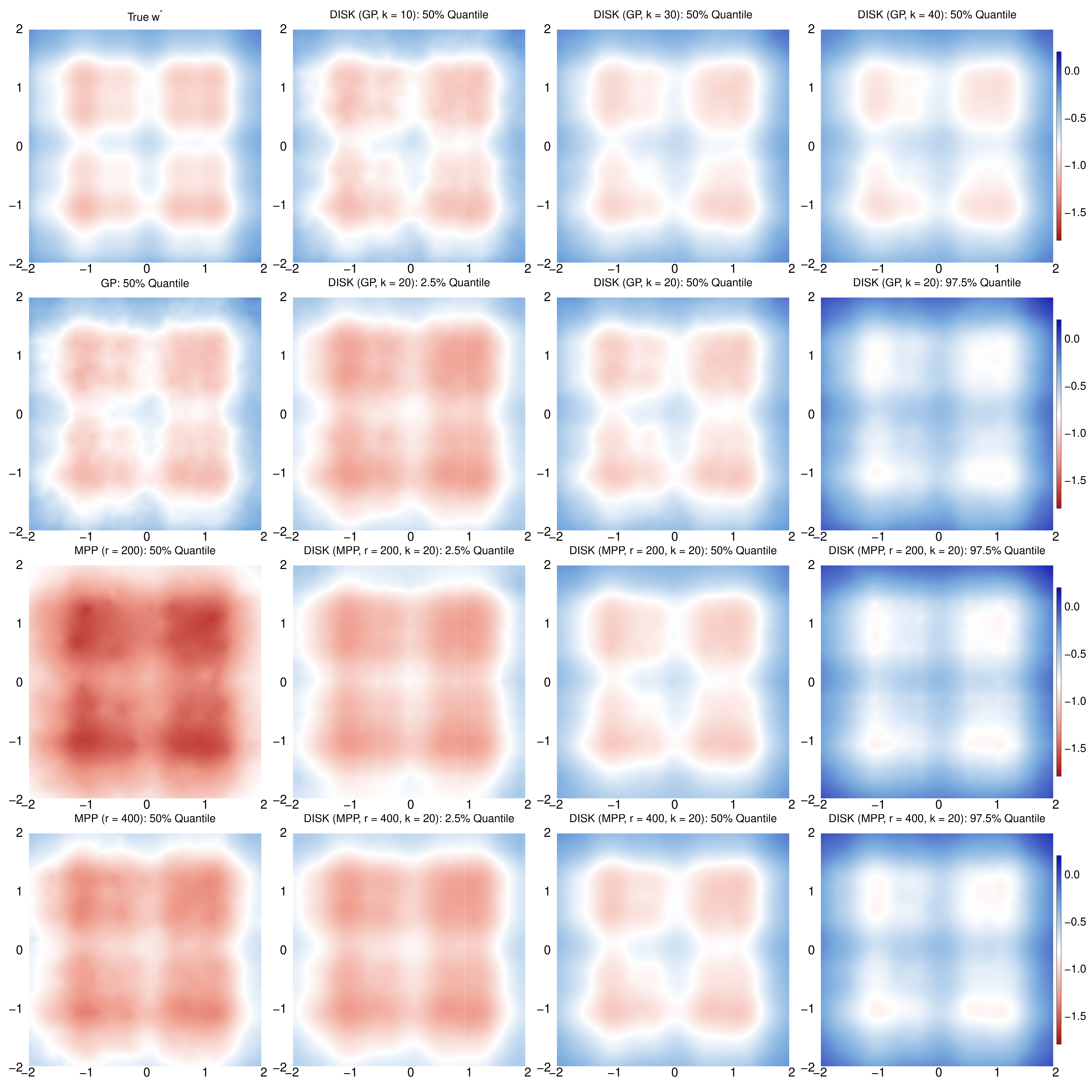}
 \caption{The spatial surface $w_0$ at the locations in $\Scal^*$ for all the competing full Bayesian methods (except NNGP) in {Simulation 1}.
   The 2.5\%, 50\%, and 97.5\% quantile surfaces, respectively, represent pointwise quantiles of the posterior distribution for $w(\sbb^*)$ for every $\sbb^* \in \Scal^*$, where
   the 50\% quantile of $w(\sbb^*)$ is the estimate of $w_0(\sbb^*)$ and the 2.5\% and 97.5\% quantiles quantify uncertainty. The true spatial surface $w_0(\sbb^*), \sbb^* \in \Scal^*$ is in the first row and column. The remaining entries in the first column are the estimates of $w_0(\sbb^*), \sbb^* \in \Scal^*$ obtained using the full data posterior distributions with full-rank GP prior and MPP prior with $r=200, 400$. The remaining entries in the first row are the estimates of $w_0(\sbb^*), \sbb^* \in \Scal^*$ obtained using DISK with GP prior and $k=10, 30, 40$, respectively. All other entries provide point estimates and quantify uncertainty in inference on $w_0(\sbb^*), \sbb^* \in \Scal^*$ with $k=20$ and GP prior (second row), MPP prior with rank 200 (third row) and MPP prior with rank 400 (fourth row). }
 \label{fig:disksurf10k-sim1}
\end{figure}

\begin{figure}[h]
 \centering
 \includegraphics[width=\textwidth]{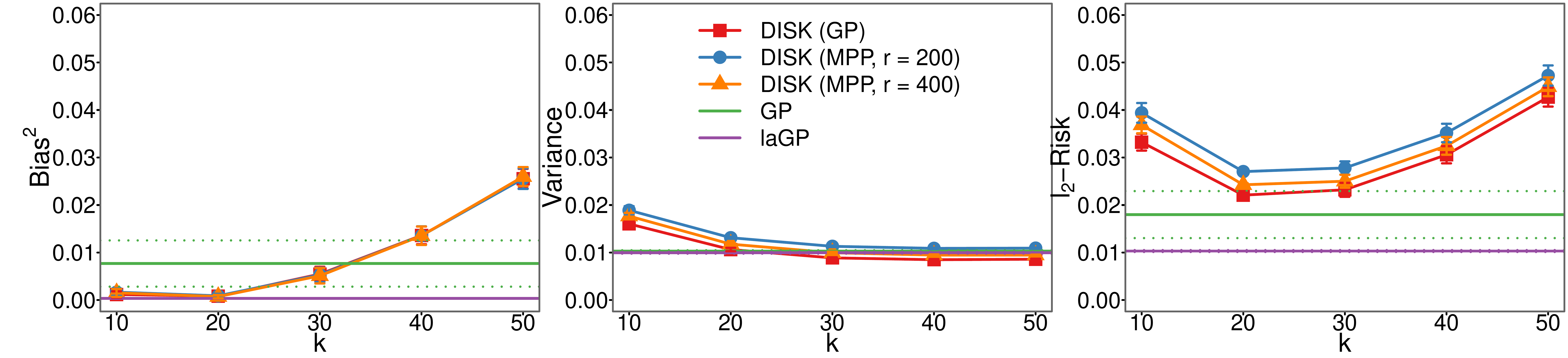}
 \caption{The empirical estimate of bias, variance, and Bayes $L_2$-risk in estimating the spatial surface $w_0$ at the locations in $\Scal^*$  in {Simulation 1}.
   The GP is the theoretical benchmark and laGP is the state-of-the-art method for estimation. The solid lines represent averaged values and the dotted lines and arrows represent one standard deviation error for the full-GP over 10 simulation replications.}
 \label{fig:bias-var-mse}
\end{figure}

\begin{table}[ht]
  \caption{Inference on the values of spatial surface at the locations in $\Scal^*$ in Simulation 1. The numbers in parentheses are standard deviations over 10 simulation replications. The bias, variance, and Bayes $L_2$-risk in the estimation of $w_0$ are defined in \eqref{eq:sim2} and the coverage and credible intervals are calculated pointwise for the locations in $\Scal^*$}
  \label{inf-w-e4}
  \centering
{\tiny
\begin{tabular}{|r|c|c|c|c|c|}
  \hline
 & Bias$^2$ & Variance& $L_2$-Risk & 95\% CI  Coverage & 95\% CI Length \\
  \hline
  laGP & 0.0004 (0.0000) & 0.0100 (0.0002) & 0.0103 (0.0002) & 1.0000 (0.0000) & 0.3890 (0.0036)  \\
  LatticeKrig & 0.0002 (0.0000) & 0.0003 (0.0000) & 0.0005 (0.0000) & 0.9867 (0.0033) & 0.0703 (0.0006)  \\
  GP & 0.0077 (0.0049) & 0.0103 (0.0002) & 0.0180 (0.0049) & 1.0000 (0.0002) & 0.3943 (0.0036)  \\
  MPP ($r=200$) & 0.3732 (0.3671) & 0.0110 (0.0002) & 0.3842 (0.3671) & 0.0000 (0.0000) & 0.4061 (0.0036)  \\
  MPP ($r=400$) & 0.0623 (0.0369) & 0.0105 (0.0002) & 0.0727 (0.0370) & 0.2946 (0.4662) & 0.3976 (0.0037)  \\
  \hline
   NNGP & \multicolumn{5}{c|}{} \\
  \hline
  NN$ = 5$ & 0.4213 (0.1373) & 0.0021 (0.0002) & 0.4233 (0.1373) & 0.0000 (0.0000) & 0.1778 (0.0079)  \\
  NN$ = 15$ & 0.4822 (0.0666) & 0.0013 (0.0001) & 0.4835 (0.0666) & 0.0000 (0.0000) & 0.1421 (0.0067) \\
  NN$ = 25$ & 0.4887 (0.0668) & 0.0013 (0.0001) & 0.4900 (0.0668) & 0.0000 (0.0000) & 0.1398 (0.0032) \\
  \hline
   CMC (MPP) & \multicolumn{5}{c|}{} \\
  \hline
  $r=200$, $k=10$ & 0.0020 (0.0006) & 0.0416 (0.0005) & 0.0436 (0.0006) & 0.8854 (0.0527) & 0.1429 (0.0010) \\
  $r=200$, $k=20$ & 0.0090 (0.0029) & 0.0402 (0.0006) & 0.0493 (0.0031) & 0.1265 (0.1027) & 0.1026 (0.0009) \\
  \hline
   CMC (MPP) & \multicolumn{5}{c|}{} \\
  \hline
  $r=400$, $k=10$ & 0.0031 (0.0013) & 0.0424 (0.0006) & 0.0455 (0.0014) & 0.7710 (0.1315) & 0.1398 (0.0013) \\
  $r=400$, $k=20$ & 0.0013 (0.0005) & 0.0409 (0.0006) & 0.0422 (0.0009) & 0.8255 (0.0987) & 0.1005 (0.0009) \\
  \hline
   DISK (GP) & \multicolumn{5}{c|}{} \\
  \hline
  $k=10$ & 0.0012 (0.0007) & 0.0160 (0.0007) & 0.0332 (0.0018) & 1.0000 (0.0000) & 0.4971 (0.0104) \\
  $k=20$ & 0.0008 (0.0005) & 0.0106 (0.0004) & 0.0221 (0.0005) & 1.0000 (0.0000) & 0.4041 (0.0070) \\
  $k=30$ & 0.0055 (0.0016) & 0.0089 (0.0001) & 0.0232 (0.0015) & 1.0000 (0.0000) & 0.3694 (0.0026) \\
  $k=40$ & 0.0136 (0.0019) & 0.0085 (0.0001) & 0.0306 (0.0018) & 0.9946 (0.0048) & 0.3612 (0.0026) \\
  $k=50$ & 0.0255 (0.0021) & 0.0086 (0.0001) & 0.0427 (0.0020) & 0.7949 (0.0572) & 0.3626 (0.0022) \\
  \hline
   DISK (MPP) & \multicolumn{5}{c|}{} \\
  \hline
  $r=200$, $k=10$ & 0.0017 (0.0008) & 0.0189 (0.0009) & 0.0394 (0.0021) & 1.0000 (0.0000) & 0.5388 (0.0122) \\
  $r=200$, $k=20$ & 0.0009 (0.0004) & 0.0131 (0.0002) & 0.0270 (0.0004) & 1.0000 (0.0000) & 0.4477 (0.0039) \\
  \hline
   DISK (MPP) & \multicolumn{5}{c|}{} \\
  \hline
  $r=400$, $k=10$ & 0.0015 (0.0008) & 0.0177 (0.0007) & 0.0369 (0.0017) & 1.0000 (0.0000) & 0.5211 (0.0099) \\
  $r=400$, $k=20$ & 0.0007 (0.0004) & 0.0118 (0.0002) & 0.0243 (0.0003) & 1.0000 (0.0000) & 0.4253 (0.0031) \\
   \hline
  MK (MPP) & \multicolumn{5}{c|}{} \\
  \hline
  $r=200$, $k=10$ & 0.0410 (0.0224) & 0.3777 (0.0293) & 0.4196 (0.0373) & 1.0000 (0.0000) & 2.4180 (0.0897) \\
  $r=200$, $k=20$ & 0.0234 (0.0081) & 0.4064 (0.0190) & 0.4298 (0.0209) & 1.0000 (0.0000) & 2.5139 (0.0628) \\
  \hline
   MK (MPP) & \multicolumn{5}{c|}{} \\
  \hline
  $r=400$, $k=10$ & 0.0247 (0.0113) & 0.3871 (0.0200) & 0.4118 (0.0225) & 1.0000 (0.0000) & 2.4576 (0.0641) \\
  $r=400$, $k=20$ & 0.0143 (0.0003) & 0.4254 (0.0148) & 0.4398 (0.0167) & 1.0000 (0.0000) & 2.5794 (0.0473) \\
   \hline
\end{tabular}
}
\end{table}

\begin{table}[h]
  \caption{
    Parametric inference and prediction in Simulation 1. For parametric inference posterior medians are provided along with
    the 95\% credible intervals (CIs) in the parentheses, where available. Similarly mean squared prediction errors (MSPEs) along with length and coverage of 95\% predictive intervals (PIs) are presented, where available. The upper and lower quantiles of 95\% CIs and PIs are averaged over 10 simulation replications, with the numbers in parentheses for the last three columns denoting standard deviations across replications;  `-' indicates that the parameter estimate or prediction is not provided by the software or the competitor}
  \label{inf-p-e4a}
\centering
{\tiny
\begin{tabular}{|r|c|c|c|c|c|c|c|}
  \hline
  & $\beta$ & $\sigma^2$ & $\tau^2$ & $\phi$ & MSPE & Coverage & Length \\
  \hline
  Truth &  1.00 &   -  &    0.01 &    - &  - & - & - \\
  laGP &   - &   -  &    - &    - & 0.010 (0.000) & 0.94 (0.01) & 0.39 (0.00) \\
  LatticeKrig   &   - &   - &    - & - & 0.010 (0.000) & 0.95 (0.01) & 0.39 (0.00) \\
  GP & 1.08 (0.50, 1.65) & 0.12 (0.10, 0.14) & 0.009 (0.009, 0.010) & 0.115 (0.107, 0.135) & 0.010 (0.000) & 0.95 (0.01) & 0.39 (0.00) \\
  MPP ($r=200$) & 1.56 (0.99, 2.15) & 0.15 (0.13, 0.18) & 0.008 (0.007, 0.008) & 0.119 (0.110, 0.133) & 0.010 (0.000) & 0.95 (0.01) & 0.41 (0.00) \\
  MPP ($r=200$) & 1.23 (0.61, 1.84) & 0.16 (0.13, 0.19) & 0.008 (0.008, 0.008) & 0.120 (0.110, 0.148) & 0.010 (0.000) & 0.95 (0.01) & 0.40 (0.00) \\
  \hline
  NNGP & \multicolumn{7}{c|}{} \\
  \hline
  NN$=5$ & 0.36 (0.36, 0.36) & 0.29 (0.29, 0.29) & 0.009 (0.009, 0.009) & 0.123 (0.123, 0.123) & 0.011 (0.000) & 0.94 (0.01) & 0.39 (0.02) \\
  NN$=15$ & 0.31 (0.31, 0.31) & 0.17 (0.17, 0.17) & 0.009 (0.009, 0.009) & 0.113 (0.113, 0.113) & 0.010 (0.000) & 0.95 (0.01) & 0.40 (0.01) \\
  NN$=25$ & 0.30 (0.30, 0.30) & 0.16 (0.16, 0.16) & 0.009 (0.009, 0.009) & 0.112 (0.112, 0.112) & 0.010 (0.000) & 0.95 (0.01) & 0.40 (0.01) \\
  \hline
  CMC (MPP) & \multicolumn{7}{c|}{} \\
  \hline
  $r=200$, $k=10$ & 0.99 (0.74, 1.23) & 0.23 (0.22, 0.25) & 0.006 (0.005, 0.006) & 0.112 (0.107, 0.119) & 0.011 (0.000) & 0.50 (0.01) & 0.14 (0.00) \\
  $r=200$, $k=20$ & 1.09 (0.89, 1.28) & 0.29 (0.27, 0.31) & 0.007 (0.006, 0.007) & 0.109 (0.106, 0.114) & 0.011 (0.000) & 0.38 (0.01) & 0.10 (0.00) \\
  \hline
  CMC (MPP) & \multicolumn{7}{c|}{} \\
  \hline
  $r=400$, $k=10$ & 1.04 (0.77, 1.30) & 0.26 (0.24, 0.28) & 0.006 (0.006, 0.007) & 0.108 (0.105, 0.114) & 0.011 (0.000) & 0.49 (0.01) & 0.14 (0.00) \\
  $r=400$, $k=20$ & 1.02 (0.82, 1.22) & 0.29 (0.27, 0.32) & 0.007 (0.007, 0.007) & 0.112 (0.109, 0.119) & 0.010 (0.000) & 0.38 (0.01) & 0.10 (0.00) \\
  \hline
  SDP (MPP) & \multicolumn{7}{c|}{} \\
  \hline
  $r=200$, $k=10$ & 0.98 (0.75, 1.23) & 0.23 (0.22, 0.25) & 0.006 (0.005, 0.006) & 0.112 (0.106, 0.118) & - & - & - \\
  $r=200$, $k=20$ & 1.08 (0.89, 1.27) & 0.29 (0.27, 0.31) & 0.007 (0.006, 0.007) & 0.109 (0.106, 0.113) & - & - & - \\
  \hline
  SDP (MPP) & \multicolumn{7}{c|}{} \\
  \hline
  $r=400$, $k=10$  & 1.04 (0.79, 1.29) & 0.26 (0.24, 0.28) & 0.006 (0.006, 0.007) & 0.109 (0.104, 0.113) & - & - & - \\
  $r=400$, $k=20$  & 1.02 (0.83, 1.21) & 0.30 (0.27, 0.32) & 0.007 (0.007, 0.007) & 0.113 (0.108, 0.118) & - & - & - \\
  \hline
  DISK (GP) & \multicolumn{7}{c|}{} \\
  \hline
  $k=10$ & 1.03 (0.80, 1.26) & 0.21 (0.17, 0.24) & 0.009 (0.008, 0.009) & 0.124 (0.111, 0.147) & 0.010 (0.000) & 0.95 (0.01) & 0.41 (0.00) \\
  $k=20$ & 0.98 (0.82, 1.15) & 0.22 (0.17, 0.26) & 0.008 (0.008, 0.009) & 0.142 (0.121, 0.179) & 0.010 (0.000) & 0.96 (0.01) & 0.42 (0.00) \\
  $k=30$ & 0.93 (0.80, 1.07) & 0.20 (0.16, 0.24) & 0.008 (0.008, 0.008) & 0.171 (0.140, 0.219) & 0.010 (0.000) & 0.96 (0.00) & 0.43 (0.00) \\
  $k=40$ & 0.88 (0.78, 1.00) & 0.18 (0.15, 0.23) & 0.008 (0.007, 0.008) & 0.201 (0.162, 0.252) & 0.010 (0.000) & 0.97 (0.01) & 0.44 (0.00) \\
  $k=50$ & 0.84 (0.75, 0.94) & 0.17 (0.14, 0.21) & 0.007 (0.007, 0.008) & 0.231 (0.186, 0.285) & 0.011 (0.000) & 0.97 (0.00) & 0.46 (0.00) \\
  \hline
  DISK (MPP) & \multicolumn{7}{c|}{} \\
  \hline
  $r=200$, $k=10$ & 1.03 (0.80, 1.27) & 0.21 (0.18, 0.24) & 0.009 (0.008, 0.009) & 0.120 (0.109, 0.144) & 0.010 (0.000) & 0.97 (0.01) & 0.44 (0.00) \\
  $r=200$, $k=20$ & 0.98 (0.82, 1.16) & 0.22 (0.17, 0.26) & 0.008 (0.008, 0.009) & 0.140 (0.119, 0.177) & 0.010 (0.000) & 0.97 (0.01) & 0.46 (0.00) \\
  \hline
  DISK (MPP) & \multicolumn{7}{c|}{} \\
  \hline
  $r=400$, $k=10$ & 1.03 (0.80, 1.27) & 0.21 (0.18, 0.24) & 0.009 (0.008, 0.009) & 0.119 (0.109, 0.143) & 0.010 (0.000) & 0.96 (0.01) & 0.42 (0.00) \\
  $r=400$, $k=20$ & 0.98 (0.82, 1.16) & 0.22 (0.17, 0.26) & 0.008 (0.008, 0.009) & 0.140 (0.119, 0.181) & 0.010 (0.000) & 0.97 (0.01) & 0.44 (0.00) \\
  \hline
  MK (MPP) & \multicolumn{7}{c|}{} \\
  \hline
  $r=200$, $k=10$ & 1.12 (0.03, 2.28) & 0.41 (0.21, 0.58) & 0.007 (0.004, 0.008) & 0.058 (0.051, 0.127) & 0.010 (0.000) & 0.96 (0.01) & 0.47 (0.00) \\
  $r=200$, $k=20$ & 1.16 (0.04, 2.33) & 0.46 (0.27, 0.66) & 0.006 (0.004, 0.008) & 0.060 (0.051, 0.112) & 0.011 (0.000) & 0.97 (0.01) & 0.49 (0.00) \\
  \hline
  MK (MPP) & \multicolumn{7}{c|}{} \\
  \hline
  $r=400$, $k=10$ & 1.12 (0.07, 2.27) & 0.41 (0.22, 0.60) & 0.007 (0.005, 0.008) & 0.068 (0.052, 0.129) & 0.010 (0.000) & 0.96 (0.00) & 0.44 (0.00) \\
  $r=400$, $k=20$ & 1.12 (-0.10, 2.38) & 0.49 (0.26, 0.71) & 0.007 (0.005, 0.009) & 0.063 (0.050, 0.125) & 0.010 (0.000) & 0.97 (0.00) & 0.48 (0.00) \\
  \hline
\end{tabular}
}%
\end{table}

\subsubsection{Simulation 2: moderately large data with a rough spatial surface.}
\label{sec:simulation-2-GP}

\begin{figure}[h]
  \centering
  \includegraphics[scale = 0.25]{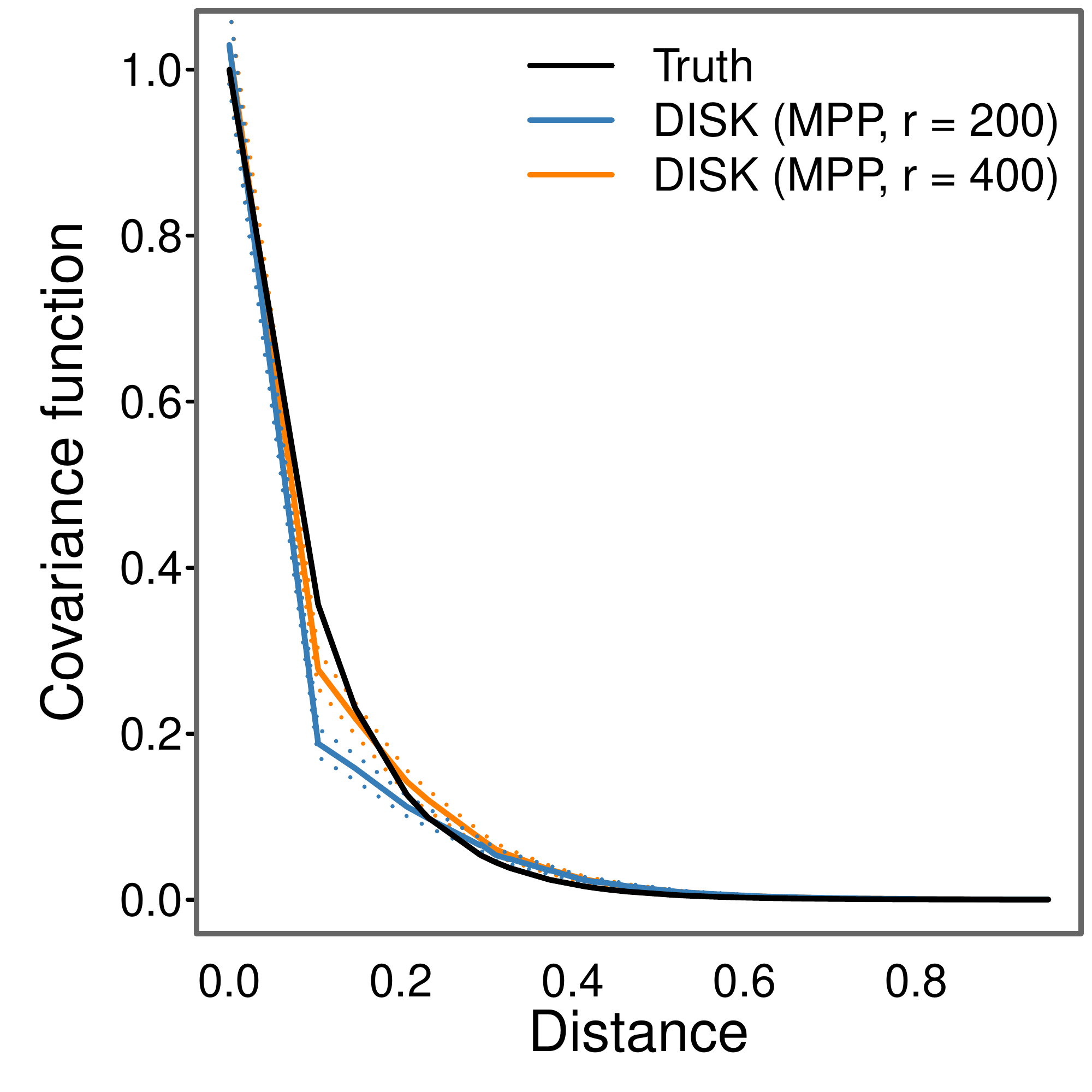}
  \caption{Covariance function of DISK-MPP as a function distance with 200 and 400 knots (r), respectively. The true covariance function is $\cov\{w_0(\sbb_i), w_0(\sbb_j)\} = \exp(-10 \|\sbb_i - \sbb_j \|_2)$.}
  \label{fig:cov_fun}
\end{figure}

Our second simulation example provides performance of DISK (MPP) with various {choices} of the number of knots when data are simulated from a Gaussian random field with nowhere differentiable  surface.  Based on the best results in Simulation 1, we use DISK (MPP) with $k=20$ and laGP as our only competitor. The bias, variance, and $L_2$-risk estimates of DISK (MPP) and laGP {show a similar pattern as in {Simulation 1}} (Table~\ref{inf-i-ns-e4a}). For both DISK (MPP) with $r=200$ and $r=400$ knots, the coverage of 95\% CI turn out to be nominal. The posterior median of all parameters are close to the true value with 95\% CIs covering the true value for $\beta,\sigma^2,\phi$, and the same is true for predictions and 95\% predictive intervals; however, DISK (MPP) slightly overestimates $\tau^2$ (Table~\ref{inf-p-ns-e4a}). This is expected given that the MPP prior applied to the full data tends to overestimate $\tau^2$.  The predictive inference of DISK shows desirable point prediction with precise characterization of uncertainty. We also emphasize that the full Bayesian inference from DISK (MPP) is computationally extremely efficient and takes less than 2 hours; {see supplementary materials for detailed comparisons}. 

{The DISK posterior of the covariance function from modified predictive process is plotted and contrasted with the true correlation function in Figure \ref{fig:cov_fun}. While earlier articles \citep{SanHua12} point out discrepancy between estimated correlation function of the modified predictive process and the true data generating correlation function, Figure \ref{fig:cov_fun} shows that DISK posterior of the correlation function recovers the true correlation quite accurately. Moreover, the accuracy increases once we employ more knots in each subset.} Table~\ref{inf-i-ns-e4a} reveals that laGP offers better $L_2$-risk estimates with comparable quantification of uncertainties for  surface estimation; however, it needs to be emphasized that while running laGP using the available package, the $\beta$ parameter is fixed at its true value. In fact, as noted before, laGP does not offer posterior estimates of parameters which are readily available from DISK. Additionally, strong local variability in the generated nowhere differentiable true spatial surface naturally prefers a nearest neighbor approach over a low-rank approach. {Since DISK (MPP) can be conceptualized as a computationally convenient alternative to MPP, a low-rank approach, with large number of knots, it is understandable that it may appear to be less effective in estimating the  surface than laGP in this case. Nevertheless, the model free nature of the DISK approach allows us to fit a nearest neighbor approach, including NNGP, on each subset to improve inference and expedite computations by multiple folds. We plan to carefully investigate this DISK feature elsewhere.}

\begin{table}[ht]
\caption{Inference on the values of spatial surface at the locations in $\Scal^*$ in Simulation 2.
    The numbers in parentheses are standard deviations over 10 simulation replications. The bias, variance, and Bayes $L_2$-risk in the estimation of $w_0$ are defined in \eqref{eq:sim2} and the coverage and length of 95\% credible intervals are calculated pointwise for the locations in $\Scal^*$}\label{inf-i-ns-e4a}
\centering
{\tiny
\begin{tabular}{|r|l|l|l|l|l|}
  \hline
& Bias$^2$ & Variance& $L_2$-Risk & Coverage & Length \\
  \hline
laGP                   & 0.4059 (0.0130) & 0.4910 (0.0086) & 0.8969 (0.0189) & 0.9670 (0.0029) & 2.7216 (0.0237) \\
DISK (MPP),$k=20,r=200$ & 0.8133 (0.0264) & 0.8791 (0.0619) & 1.6923 (0.0793) & 0.9598 (0.0067) & 3.6735 (0.1311) \\
DISK (MPP),$k=20,r=400$ & 0.7295 (0.0256) & 0.8219 (0.0571) & 1.5515 (0.0756) & 0.9645 (0.0066) & 3.5531 (0.1245) \\
   \hline
\end{tabular}
}
\end{table}

\begin{table}[ht]
\caption{Parameter estimates along with 95\% credible intervals from DISK(MPP) with $r=200$ and $r=400$ in Simulation 2. For parametric inference posterior medians are provided along with the 95\% credible intervals (CIs) in the parentheses , where available. Similarly mean squared prediction errors (MSPEs) along with length and coverage of 95\% predictive intervals (PIs) are presented, where available. The upper and lower quantiles of 95\% CIs and PIs are averaged over 10 simulation replications, with the numbers in parentheses for the last three columns denoting standard deviations across replications;  `-' indicates that the parameter estimate or prediction is not provided by the software or the competitor.}\label{inf-p-ns-e4a}
\centering
{\tiny
\begin{tabular}{|r| l |l|l|l|l|l|l|l}
  \hline
 & $\beta$ & $\sigma^2$ & $\tau^2$ & $\phi$ & MSPE & Coverage & Length & $\log_{10}$(Time) \\
  \hline
Truth & 1 & 1 & 0.1 & 9 & - & - & - & -\\
laGP & -  &  -  &  -   &   -  & 0.50 (0.0120) & 0.21 (0.0115) & 0.38 (0.0003) & 0.79 (0.0193) \\
   DISK, $k=20,r=200$ & 1.00 (0.96, 1.03) & 0.97 (0.92, 1.03) & 0.14 (0.12, 0.17) & 9.28 (8.94, 9.68) & 0.90 (0.0210) & 0.96 (0.0027) & 3.97 (0.0773) & 2.39 (0.0011) \\
   DISK, $k=20,r=400$ & 1.00 (0.96, 1.03) & 0.97 (0.92, 1.03) & 0.14 (0.11, 0.17) & 9.32 (8.96, 9.73) & 0.82 (0.0194) & 0.96 (0.0035) & 3.86 (0.0717) & 2.51 (0.0262) \\
   \hline
\end{tabular}
}
\end{table}

\subsubsection{Simulation 3: large spatial data.}
\label{sec:simulation-2:-large}

Our ultimate goal is to apply DISK in massive data settings, so we evaluate its performance when $n=10^6$. As mentioned earlier, massive size of the data in Simulation 3 prohibits the fitting of models based on full-rank and low-rank GPs, including MPP, LatticeKrig, and NNGP due to numerical issues, leaving only laGP as a feasible competitor. Since Simulation 1 demonstrates similar performance of DISK (MPP) and DISK (GP) with DISK (MPP) having a smaller run time, we use only DISK (MPP) for comparisons with laGP in Simulation 3.
An identical three-step strategy for fitting DISK (MPP) is employed as in Simulation 1 but with $k=500$ and with $r=400$ and $r=600$. Notably, $r$ is increased from $200$ and $400$ in Simulation 1 to $400$ and $ 600$ to account for the larger subset size in this simulation, maintaining a high $r/m$ ratio.

The results for DISK's uncertainty quantification in parameter estimation and prediction agree with those observed in Simulation 1, but, unlike Simulation 1, DISK outperforms laGP in the estimation of $w_0(\sbb^*)$ for $\sbb^* \in \Scal^*$ for both choices of $r$ (Tables \ref{inf-w-e5} and \ref{inf-p-e6}). The point estimates of $\beta$ and $\tau^2$ are close to their true values with narrow 95\% CIs. For $\tau^2$, the CI misses the truth, which is expected given that the full data GP in Simulation 1 underestimates $\tau^2$. The bias, variance, and Bayes $L_2$-risk of DISK (MPP) for both $r$s are lower than those of laGP. The coverages of 95\% CIs for laGP and DISK (MPP) are the same but the lengths of 95\% CIs in DISK (MPP) are smaller than those of laGP for both $r$s (see Table~\ref{inf-w-e5}). We conclude that DISK (MPP) matches the performance of laGP in delivering predictive inference, while it outperforms laGP in terms of estimation of $w_0(\sbb^*)$ for $\sbb^* \in \Scal^*$.

\begin{table}[ht]
  \caption{Inference on the values of spatial surface at the locations in $\Scal^*$ in Simulation 3. The numbers in parentheses are standard deviations over 10 simulation replications. The bias, variance, and $L_2$-risk in the estimation of $w_0$ are defined in \eqref{eq:sim2} and the coverage and length of 95\% credible intervals are calculated pointwise for the locations in $\Scal^*$.}
  \label{inf-w-e5}
  \centering
  {\tiny
    \begin{tabular}{|r|c|c|c|c|c|}
      \hline
      & Bias$^2$ & Variance& $L_2$-Risk & Coverage & Length \\
      \hline
      laGP & 0.0002 (0.0000) & 0.0100 (0.0000) & 0.0102 (0.0000) & 1.0000 (0.0000) & 0.3905 (0.0006) \\
      DISK, $k=500$, $r=400$ & 0.0002 (0.0000) & 0.0030 (0.0000) & 0.0061 (0.0000) & 1.0000 (0.0000) & 0.2132 (0.0002) \\
      DISK, $k=500$, $r=600$ & 0.0001 (0.0000) & 0.0026 (0.0000) & 0.0052 (0.0000) & 1.0000 (0.0000) & 0.1977 (0.0002) \\
      \hline
    \end{tabular}
  }
\end{table}

\begin{table}[h]
  \caption{
    Parametric inference and prediction in Simulation 3. For parametric inference posterior medians are provided along with
    the 95\% credible intervals (CIs) in the parentheses, where available. Similarly mean squared prediction errors (MSPEs) along with length and coverage of 95\% predictive intervals (PIs) are presented, where available. The upper and lower quantiles of 95\% CIs and PIs are averaged over 10 simulation replications, with the numbers in parentheses for the last three columns denoting standard deviations across replications;  `-' indicates that the parameter estimate or prediction is not provided by the software or the competitor.
  }
  \label{inf-p-e6}
\centering
{\tiny
\begin{tabular}{|r|c|c|c|c|c|c|c|}
  \hline
  & $\beta$ & $\sigma^2$ & $\tau^2$ & $\phi$  \\
  \hline
  Truth &  1.00 &   -  &    0.01 &    -  \\
   DISK (MPP) ($r=400$, $k=500$)  & 1.01 (0.98, 1.04) & 0.16 (0.15, 0.17) & 0.008 (0.008, 0.008) & 0.13 (0.13, 0.13) \\
   DISK (MPP) ($r=600$, $k=500$)  & 1.01 (0.98, 1.04) & 0.16 (0.15, 0.17) & 0.008 (0.008, 0.008) & 0.13 (0.13, 0.13) \\
  \hline
  & MSPE & Coverage & Length   & \\
  \hline
   laGP                         &  0.010 (0.0003) & 0.94 (0.0040) & 0.39 (0.0000) & \\
   DISK (MPP) ($r=400$, $k=500$)  &  0.009 (0.0003) & 0.96 (0.0030) & 0.41 (0.0000) & \\
   DISK (MPP) ($r=600$, $k=500$)  &  0.009 (0.0002) & 0.95 (0.0040) & 0.40 (0.0000) & \\
  \hline
\end{tabular}
}%
\end{table}

\subsection{Real data: Sea Surface Temperature data}
\label{sst}

A description of the evolution and dynamics of the SST is a key component of the study of the earth's climate. SST data  (in centigrade) from ocean samples have been collected by voluntary observing ships, buoys, and military and scientific cruises for decades. During the last 20 years or so, the SST database has been complemented by regular streams of remotely sensed observations from satellite orbiting the earth. A careful quantification of variability of SST data is important for climatological research, which includes determining the formation of sea breezes and sea fog and calibrating measurements from weather satellites \citep{Dietal08}. A number of articles have appeared to address this issue in recent years; see \cite{berliner2000long,LemSan09,wikle2011polynomial}.

We consider the problem of capturing the spatial trend and characterizing the uncertainties in the SST in the west coast of mainland U.S.A., Canada, and Alaska between $40^\circ$--$65^\circ$ north latitudes and $100^\circ$--$180^\circ$ west longitudes. The dataset is obtained from NODC World Ocean Database (\url{https://www.nodc.noaa.gov/OC5/WOD/pr_wod.html}). Due to our focus on spatial modeling, we ignore the temporal component. After screening the data for quality control, we choose a random subset of about $1,000,800$ spatial observations over the selected domain. From the selected observations, we randomly select $10^6$ observations as training data and the remaining observations are used to compare the performance of DISK and its competitors.  We replicate this setup ten times. The selected domain is large enough to allow considerable spatial variation in SST from north to south and provides an important first step in extending these models for analyzing global-scale SST database.

The SST data in the selected domain shows a clear decreasing trend in SST with increasing latitude (Figure~\ref{fig:sst-pred}). Based on this observation, we add latitude as a linear predictor in the univariate spatial regression model \eqref{parent_proc} to explain the long-range directional variability in the SST.
The setup is identical to Simulation 2, except for the presence of the latitude predictor with the corresponding coefficient $\beta_1$.
We assign $N(0, 100)$ prior to $\beta_1$, and the remaining priors and DISK competitors are identical to those in Simulation 2. For application of CMC (MPP), DISK (MPP), and SDP (MPP), we follow the three-step strategy used in Simulation 1 but with $k=300$ and $r=400$. Due to the lack of ground truth for estimating $w_0(\sbb^*)$, we compare the four methods in terms of their inference on $\Omegab$ and prediction of $y(\sbb^*)$ for $\sbb^* \in \Scal^*$ in terms of MSPE and the length and coverage of 95\% posterior PIs.

DISK (MPP) outperforms CMC (MPP) and SDP (MPP) in predictions while closely matching the results of laGP, the current state-of-the-art method for modeling massive spatial data.
The  50\%, 2.5\%, and 97.5\% quantiles of the posterior distributions for $\Omegab$, $w(\sbb^*)$ and $y(\sbb^*)$ for every $\sbb^* \in \Scal^*$ are used for estimation and  uncertainty quantification. CMC (MPP), SDP (MPP) and DISK (MPP)  agree closely in their inference on $\Omegab$, but SDP (MPP) fails to provide any result for $w(\sbb^*)$ or $y(\sbb^*)$ due to the large size of $\Scal^*$ (Table \ref{tbl:sst-pred}). For every $\sbb^* \in \Scal^*$, CMC's and DISK's estimates of $w(\sbb^*)$ and $y(\sbb^*)$ agree closely, but CMC severely underestimates uncertainty in $w(\sbb^*)$ and $y(\sbb^*)$ (Figures \ref{fig:sst-pred} and \ref{fig:sst-w} and Table \ref{tbl:sst-pred}). The pointwise predictive coverages of laGP and DISK match their nominal levels; however, the 95\% posterior PIs of DISK are wider than those of laGP because DISK accounts for uncertainty due to the error term (Figure \ref{fig:sst-pred} and Table \ref{tbl:sst-pred}).
As a whole, SST data analysis reinforces our findings on DISK as a computationally efficient, flexible, and fully Bayesian inferential tool.

\begin{table}[h]
  \caption{Parametric inference and prediction in SST data. CMC, SDP, and DISK use  MPP-based  modeling with $r=400$ on $k=300$ subsets.
    For parametric inference posterior medians are provided along with
    The 95\% credible intervals (CIs) in the parentheses, where available. Similarly mean squared prediction errors (MSPEs) along with length and coverage of 95\% predictive intervals (PIs) are presented, where available. The upper and lower quantiles of 95\% CIs and PIs are averaged over 10 simulation replications, with the numbers in parentheses for the last three columns denoting standard deviations across replications;  `-' indicates that the parameter estimate or prediction is not provided by the software or the competitor}
  \label{tbl:sst-pred}
  \centering
  {\tiny
    \begin{tabular}{|r|c|c|c|c|c|}
      \hline
      & $\beta_0$ & $\beta_1$ & $\sigma^2$ & $\tau^2$ & $\phi$ \\
      \hline
      laGP & - & - & - & - & - \\
      CMC & 31.78 (31.19, 32.37) & -0.35 (-0.36, -0.34) & 12.22 (11.78, 12.69) & 0.110 (0.108, 0.112) & 0.021 (0.020, 0.022) \\
      SDP & 31.67 (31.45, 31.82) & -0.35 (-0.36, -0.34) & 14.42 (13.29, 14.80) & 0.110 (0.108, 0.112) & 0.021 (0.020, 0.022) \\
      DISK & 32.34 (31.74, 32.95) & -0.32 (-0.33, -0.31) & 11.83 (11.23, 12.45) & 0.184 (0.182, 0.185) & 0.039 (0.037, 0.041) \\
      \hline
      & MSPE & Coverage & Length & & \\
      \hline
      laGP & {0.25} (0.00) & {0.95} (0.00) & {2.35} (0.00) & & \\
      CMC & 0.41 (0.00) & 0.13 (0.00) & 0.14 (0.00) &  & \\
      SDP & - & - & - & &  \\
      DISK & 0.41 (0.00) & {0.95} (0.00) & 2.67 (0.00) &  &  \\
      \hline
    \end{tabular}
  }
\end{table}

\begin{figure}[h]
  \centering
  \includegraphics[width=\textwidth]{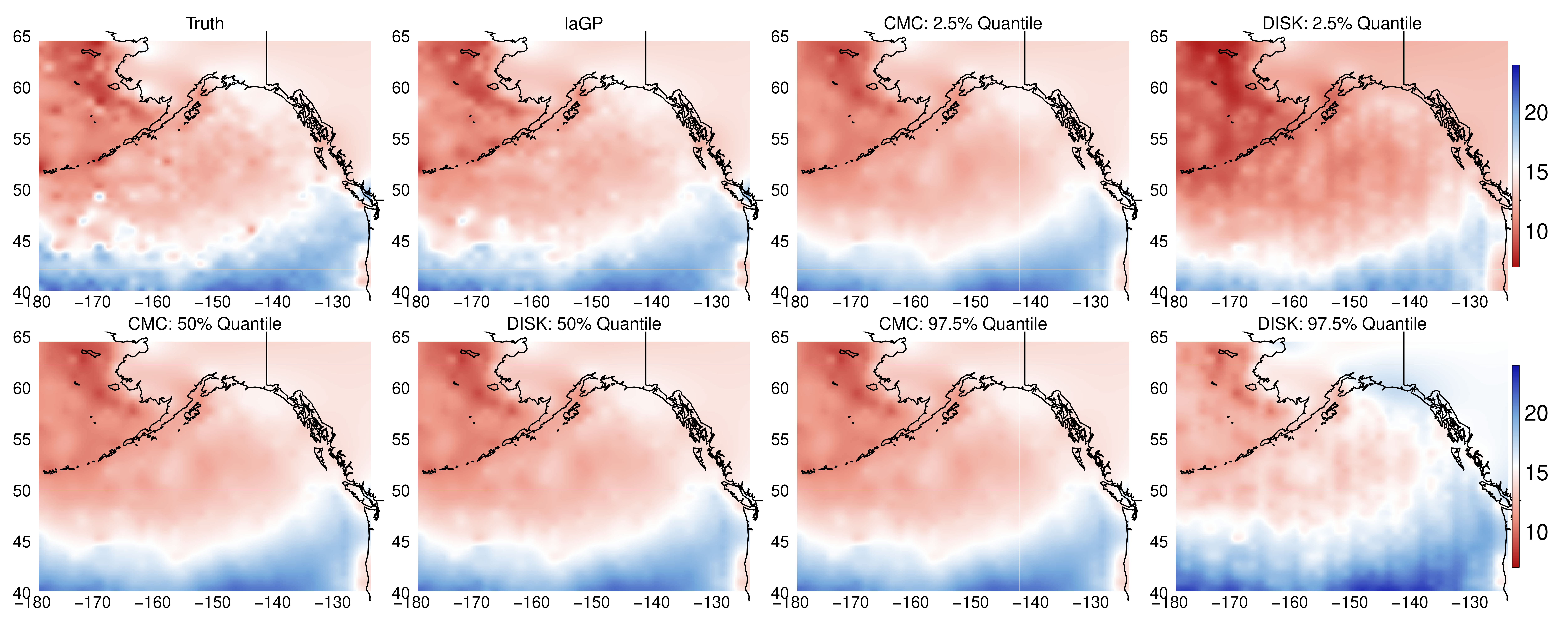}
  \caption{Predication of sea surface temperatures at the locations in $\Scal^*$. Negative longitude means degree west from Greenwich.
    CMC and DISK use  MPP-based  modeling with $r=400$ on $k=300$ subsets and laGP uses the `nn' method. The 2.5\%, 50\%, and 97.5\% quantile surfaces, respectively, represent pointwise quantiles of the posterior distribution for $y(\sbb^*)$ for every $\sbb^* \in \Scal^*$. }
  \label{fig:sst-pred}
\end{figure}

\begin{figure}[h]
  \centering
  \includegraphics[width=\textwidth]{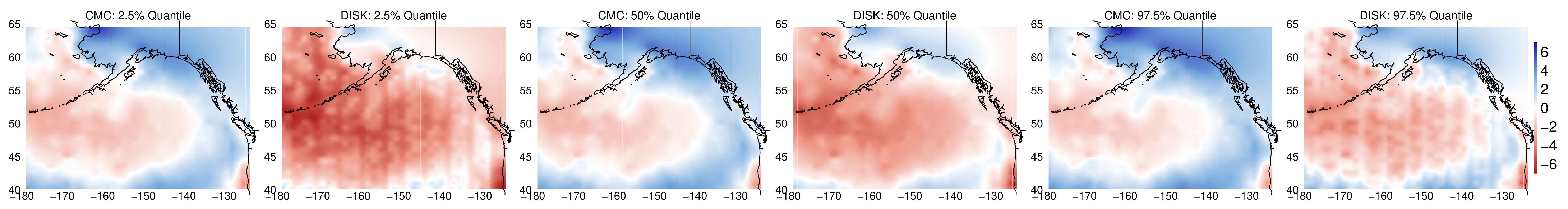}
  \caption{Interpolated spatial surface $w$ at the locations in $\Scal^*$. Negative longitude means degree west from Greenwich.
    CMC and DISK use  MPP-based  modeling with $r=400$ on $k=300$ subsets. The 2.5\%, 50\%, and 97.5\% quantile surfaces, respectively, represent pointwise quantiles of the posterior distribution for $w(\sbb^*)$ for every $\sbb^* \in \Scal^*$.}
  \label{fig:sst-w}
\end{figure}

\section{Discussion}
\label{sec:discussion}

This article presents a novel distributed Bayesian approach for kriging with massive data using the divide-and-conquer technique.
 We provide explicit upper bound on the number of subsets $k$ depending on the analytic properties of the spatial surface, so that the Bayes $L_2$-risk of the DISK posterior is nearly minimax optimal. We have confirmed this empirically via simulated and real data analyses, where DISK compares well with state-of-the-art methods. Additional theoretical results in the supplementary material shed light on the posterior convergence rate of the DISK posterior.


The simplicity and generality of the DISK framework enable scaling of any spatial model. For example, recent applications have confirmed that the NNGP prior requires modifications if scalability is desired for even a few millions of locations \citep{Finetal17}. In future, we aim to scale ordinary NNGP and other multiscale approaches to tens of millions of locations with the DISK framework. Another important future work is to extend the DISK framework for scalable modeling of multiple correlated outcomes observed over massive number of locations.

This article focuses on developing the DISK framework for spatial modeling due to the motivating applications from massive geostatistical data. The DISK framework, however, is applicable to any mixed effects model where the random effects are assigned a GP prior, which includes Bayesian nonparametric regression using GP prior.
We plan to explore more general applications in the future with high dimensional covariates.

\section*{Acknowledgment}

We thank Professor David B. Dunson of Duke University for inspiring many questions that we have addressed in this work and Professor Sudipto Banerjee of UCLA for helpful conversations. Cheng Li's research is supported by the National University of Singapore start-up grant (R155000172133). This material is based on research supported by the Office of Naval Research under Award Number ONR-BAA N00014-18-1-2741. 

\clearpage

\renewcommand\thesection{\arabic{section}}
\renewcommand\thesubsection{\thesection.\arabic{subsection}}
\renewcommand\thesubsubsection{\thesubsection.\arabic{subsubsection}}

\setcounter{section}{0}

\begin{center}
\textbf{\LARGE Supplementary Material for A Divide-and-Conquer Bayesian Approach to Large-Scale Kriging}
\end{center}

\section{Proof of Theorems in Section 3.4}

Recall that the spatial regression model with a GP prior considered in Section 2 is
\begin{align}
y(\sbb_i)= w(\sbb_i) + \epsilon(\sbb_i), \quad \epsilon(\sbb_i) \sim N(0, \tau^2), \quad i = 1, \ldots, n, \quad  w(\cdot) \sim \text{GP}\{0, C_{\alphab}(\cdot, \cdot) \}. \label{atheory-mdl}
\end{align}
Writing this model for the $n$ locations in $\Scal$ gives
\begin{align}
  \label{eq:ap1}
  \yb = \wb_0 + \epsilonb, \quad \epsilonb \mid \Scal \sim N(\zero, \tau^2 \Ib), \quad \yb \mid \Scal \sim N(\wb_0, \tau^2 \Ib),
\end{align}
where $\wb_0 = \{w_0(\sbb_1), \ldots, w_0(\sbb_n)\}$ and $\epsilonb = \{\epsilon(\sbb_1), \ldots, \epsilon(\sbb_n)\}$ are the true value of the residual spatial surface and white noise realized at the locations in $\Scal$. We can write the model in a similar format for each data subset. Let $\sbb \in \Dcal$ be a location, $w_0(\sbb)$ be the true value of the residual spatial surface, $\EE_{\sbb^*}$, $\EE_{0}$, $\EE_{\Scal}$, and $\EE_{0 \mid \Scal}$  respectively be the expectations with respect to the distributions of $\sbb^*$, $(\Scal, \yb)$, $\Scal$, and $\yb$ given $\Scal$.
If $\overline w(\sbb^*)$ is a random variable that follows the DISK posterior for estimating $w_0(\sbb^*)$, then $\overline w(\sbb^*)$ has the density $N(\overline m, \overline v)$, where
\begin{align}
  \label{eq:ae3}
  \overline m = \frac{1}{k} \sum_{j=1}^k \cb_{j,*}^T (\Cb_{j, j} + \tfrac{\tau^2}{k} \Ib)^{-1}\yb_j , \;
  \overline v^{1/2} = \frac{1}{k}\sum_{j=1}^k v_j^{1/2}, \;
  v_j = c_{*,*} - \cb_{j,*}^T (\Cb_{j,j} + \tfrac{\tau^2}{k} \Ib)^{-1} \cb_{j, *},
\end{align}
where $c_{*,*}=C_{\alphab}(\sbb^*,\sbb^*)$, and $\cb_{j,*}^T = \cb_j^T(\sbb^*) = [C_{\alphab}(\sbb_{j1}, \sbb^*), \ldots, C_{\alphab}(\sbb_{jm}, \sbb^*)]$. In the proofs below, without confusion, we use the notation $\cb_{j,*}$ and $\cb_j(\sbb^*)$ interchangeably.


The Bayes $L_2$-risk in estimating $w_0$ using the DISK posterior is defined as
\begin{align} \label{eq:ap2}
\EE_{0} \left[ \EE_{\sbb^*} \{\overline w(\sbb^*) - w_0(\sbb^*)\}^2 \right]  \overset{(i)}{=} \EE_{\Scal}  \int_{\Dcal} \EE_{0 \mid \Scal}\{\overline w(\sbb^*) - w_0(\sbb^*)\}^2 \PP_{\sbb}(d \sbb^*),
\end{align}
where $(i)$ follows from Fubini's theorem. Using bias-variance decomposition,
\begin{align*}
\EE_{0 \mid \Scal} \{ \overline w(\sbb^*) - w_0(\sbb^*)\}^2  &= \EE_{0 \mid \Scal} \left[\overline w(\sbb^*) - \EE_{0 \mid \Scal} \{\overline w(\sbb^*)\} +  \EE_{0 \mid \Scal} \{\overline w(\sbb^*)\} - w_0(\sbb^*) \right]^2\\
&= \left[ \EE_{0 \mid \Scal} \{\overline w(\sbb^*)\} - w_0(\sbb^*) \right]^2  + \EE_{0 \mid \Scal} \left[\overline w(\sbb^*) - \EE_{0 \mid \Scal} \{\overline w(\sbb^*)\}\right]^2   \\
&\equiv  \text{bias}_{0 \mid \Scal}^2\{\overline w(\sbb^*)\} + \var_{0 \mid \Scal}\{\overline w(\sbb^*)\} .
\end{align*}
If $\cb_j^T(\cdot) = [\cov\{w(\cdot), w(\sbb_{j1})\}, \ldots, \cov\{w(\cdot), w(\sbb_{jm})\}] = \{C_{\alphab}(\sbb_{j1}, \cdot), \ldots, C_{\alphab}(\sbb_{jm}, \cdot)\}$,  $\cb^T(\cdot) = \{\cb_1^T(\cdot), \ldots, \cb_k^T(\cdot)\}$, $\wb_{0j}^T = \{w_0(\sbb_{j1}), \ldots, w_0(\sbb_{jm})\}$, and $\wb_{0}^T = \{\wb^T_{01}, \ldots, \wb_{0k}^T\}$,
then the distribution of $\overline w(\sbb^*)$ in \eqref{eq:ae3} implies that
\begin{align*}
\EE_{0 \mid \Scal} \{\overline w(\sbb^*) \}  &= \frac{1}{k} \sum_{j=1}^k \cb_j^T (\sbb^*) \left(\Cb_{j, j} + \tfrac{\tau^2 }{k} I \right)^{-1} \wb_{0j} = \cb_j^T(\sbb^*) (k \Lb + \tau^2 \Ib)^{-1} \wb_0,\\
\var_{0 \mid \Scal}\{\overline w(\sbb^*)\} &= \var_{0 \mid \Scal} \left[ \EE \{\overline w(\sbb^*) \mid \yb\} \right] + \EE_{0 \mid \Scal} \left[ \var \{\overline w(\sbb^*) \mid \yb \} \right] \\
&= \tau^2 \cb^T(\sbb^*) (k \Lb + \tau^2 \Ib)^{-2} \cb(\sbb^*) + \overline v(\sbb^*),
\end{align*}
where $\Lb$ is a block-diagonal matrix with $\Cb_{1,1}, \ldots, \Cb_{k,k}$ along the diagonal;
therefore, the Bayes $L_2$-risk in \eqref{eq:ap2} can be decomposed into three parts:
\begin{align}
  \label{eq:eap4}
\EE_{\sbb^*} \EE_{\Scal} \{  \cb_{*}^T(k \Lb + \tau^2 \Ib)^{-1}\wb_0 - w_0(\sbb^*)\}^2 + \tau^2 \EE_{\sbb^*} \EE_{\Scal} \left\{\cb^T_* (k \Lb + \tau^2 \Ib)^{-2} \cb_* \right\}  + \EE_{\sbb^*}\EE_{\Scal}  \overline v(\sbb^*),
\end{align}
which correspond to $\text{bias}^2$, ${\var}_{\text{mean}}$ and ${\var}_{\text{DISK}}$ in Theorem 1.

\subsection{Proof of Theorem 1}

The next three sections find upper bounds for each of the three terms in \eqref{eq:eap4}. The conclusion of Theorem 1 follows directly by combining the three upper bounds.


\subsubsection{An upper bound for the squared bias}
\label{sec:lemma-6-zhangs}

Consider the squared-bias term in \eqref{eq:eap4}. For ease of presentation, assume that $\{\sbb_1, \ldots, \sbb_n\}$ are relabeled to $$\{\sbb_{11}, \ldots, \sbb_{1m}, \ldots, \sbb_{k1}, \ldots, \sbb_{km}\}$$ corresponding to the $k$ subsets. Define $\xi_{\sbb_{ji}}(\cdot) = C_{\alphab}(\sbb_{ji}, \cdot)$,
\begin{align}
  \label{eq:14}
  \wb^T_0 &= \left( \langle w_0, \xi_{\sbb_{11}} \rangle_{\HH}, \ldots, \langle w_0, \xi_{\sbb_{1m}} \rangle_{\HH}, \ldots, \langle w_0, \xi_{\sbb_{k1}} \rangle_{\HH}, \ldots, \langle w_0, \xi_{\sbb_{km}} \rangle_{\HH} \right) \nonumber \\
   &\equiv  (\wb_{01}^T, \ldots,  \wb_{0k}^T), \nonumber\\
  \cb^T(\cdot) &= (\xi_{\sbb_{11}}, \ldots, \xi_{\sbb_{1m}}, \ldots, \xi_{\sbb_{k1}}, \ldots, \xi_{\sbb_{km}})  \nonumber \\
  &= \{\cb_1^T(\cdot), \ldots,  \cb_k^T(\cdot)\} \equiv (\cb_1^T, \ldots,  \cb_k^T).
\end{align}
The following lemma provides an upper bound on the squared bias of the DISK posterior.
\begin{lemma}\label{lem-sup-1}
If Assumptions A.1--A.3 in the main paper hold, then for some global constant $A>0$,
\begin{align*}
& \EE_{\sbb^*} \EE_{\Scal} \{  \cb_{*}^T(k \Lb + \tau^2 \Ib)^{-1}\wb_0 - w_0(\sbb^*)\}^2 \leq \\
& \frac{8\tau^2}{n} \| w_0 \|_{\HH}^2 + \| w_0 \|_{\HH}^2 \; \underset{d \in \NN}{\inf} \, \left[ \frac{8n}{\tau^2_0} \rho^4 \tr(C_{\alphab}) \tr(C_{\alphab}^{d}) +  \mu_1 \left\{  \frac{A b(m, d, r) \rho^2 \gamma(\tfrac{\tau^2_0}{n})}{\sqrt{m}} \right\}^r \right].
\end{align*}
\end{lemma}

\begin{proof}
Based on the term $\cb_{*}^T(k \Lb + \tau^2 \Ib)^{-1}\wb_0$ in \eqref{eq:eap4}, we define $\Delta_j$ ($j = 1, \ldots, k$) and $\Delta$ as
\begin{align} \label{eq:12}
\Delta_j(\cdot) &= \yb_j^T (\Cb_{j,j} + \tfrac{\tau^2}{k} \Ib)^{-1} \cb_j (\cdot) - w_0(\cdot) \equiv \tilde w_j (\cdot) - w_0(\cdot), \nonumber\\
\Delta(\cdot) &= \yb^T (k \Lb + \tau^2 \Ib)^{-1} \cb (\cdot) - w_0(\cdot)
= \frac{1}{k} \sum_{j=1}^k \left\{ \tilde w_j (\cdot) - w_0(\cdot) \right\}
= \frac{1}{k} \sum_{j=1}^k \Delta_j(\cdot),
\end{align}
so that $\EE_{0 \mid \Scal} (\Delta) = \wb_0^T(k \Lb + \tau^2 \Ib)^{-1} \cb(\cdot) - w_0(\cdot) = k^{-1} \sum_{j=1}^k  \EE_{0 \mid \Scal} (\Delta_j)$ and $\EE_{\Scal} \|\EE_{0 \mid \Scal} (\Delta) \|_2^2$ yields the $\text{bias}^2$ term in \eqref{eq:eap4}. Jensen's inequality implies that $\|\EE_{0 \mid \Scal} (\Delta) \|_2^2 \leq k^{-1} \sum_{j=1}^k \|\EE_{0 \mid \Scal} (\Delta_j) \|_2^2$, so we only need to find upper bounds for $\| \EE_{0 \mid \Scal} (\Delta_j) \|^2_2$ ($j=1, \ldots, k$).

We can recognize that the optimization problem below has $\tilde w_j(\cdot)$ defined in \eqref{eq:12} as its solution,
\begin{align}
  \label{eq:ba4}
  {\text{argmin}}_{w\in \Hcal}\sum_{i=1}^m \frac{\left\{ w(\sbb_{ji}) - y(\sbb_{ji}) \right\}^2}{2\tau^2/k} + \frac{1}{2} \| w \|_{\HH}^2, \quad j = 1, \ldots, k.
\end{align}
Differentiating \eqref{eq:ba4} and taking expectations with respect to $\EE_{0 \mid \Scal}$ implies that
\begin{align}
  \label{eq:ba6}
  & \sum_{i=1}^m  \EE_{0 \mid \Scal} \left\{ \tilde w_j(\sbb_{ji}) - y(\sbb_{ji}) \right\} \, \xi_{\sbb_{ji}} + \frac{\tau^2}{k} \EE_{0 \mid \Scal} (\tilde w_{j}) \nonumber \\
  & = \sum_{i=1}^m \langle \EE_{0 \mid \Scal} (\Delta_j), \xi_{\sbb_{ji}} \rangle_{\HH} \, \xi_{\sbb_{ji}}  + \frac{\tau^2}{k} \EE_{0 \mid \Scal} (\tilde w_{j}) = 0,
\end{align}
where the last inequality follows because $y(\sbb_{ji}) = \langle w_0, \xi_{\sbb_{ji}} \rangle_{\HH} + \langle \epsilon, \xi_{\sbb_{ji}} \rangle_{\HH}$ and $\langle \EE_{0 \mid \Scal}(\epsilon), \xi_{\sbb_{ji}} \rangle_{\HH} = 0$.  Using \eqref{eq:12},  $\Delta_j = \tilde w_j - w_0$, $\EE_{0 \mid \Scal} (\tilde w_j)   = \EE_{0 \mid \Scal} (\Delta_j) + w_0$, and dividing by $m$ in \eqref{eq:ba6}, we obtain that
\begin{align}
  \label{eq:16}
    \frac{1}{m}\sum_{i=1}^m  \langle \EE_{0 \mid \Scal} (\Delta_j), \xi_{\sbb_{ji}} \rangle_{\HH} \, \xi_{\sbb_{ji}}  + \frac{\tau^2}{km}  \EE_{0 \mid \Scal} (\Delta_j) = - \frac{\tau^2}{km} w_0.
\end{align}
If we define the $j$th sample covariance operator as $\hat \Sigmab_j = \frac{1}{m} \sum_{j=1}^m \xi_{\sbb_{ji}} \otimes \xi_{\sbb_{ji}}$, then \eqref{eq:16} reduces to
\begin{align}
  \label{eq:17}
  \left( \hat \Sigmab_j  +  \tfrac{\tau^2}{km} \Ib \right) \EE_{0 \mid \Scal} (\Delta_j) = - \frac{\tau^2}{km} w_0 \implies \|\EE_{0 \mid \Scal}(\Delta_j) \|_{\HH} \leq \|  w_0 \|_{\HH}, \quad j = 1, \ldots, k,
\end{align}
where the last inequality follows because $\hat \Sigmab_j$ is a positive semi-definite matrix.

The rest of the proof finds an upper bound for $\| \EE_{0 \mid \Scal}(\Delta_j) \|_2^2$. We now reduce this problem to a finite dimensional one indexed by a chosen $d \in \NN$. Let $\deltab_{j} = (\delta_{j1}, \ldots, \delta_{jd}, \delta_{j(d+1)}, \ldots, \delta_{j \infty}) \in L_2(\NN)$ such that
\begin{align}
  \label{eq:ba1}
  &\EE_{0 \mid \Scal} (\Delta_j) = \sum_{i=1}^{\infty} \delta_{ji} \phi_i, \quad \delta_{ji} = \langle \EE_{0 \mid \Scal} (\Delta_j),\nonumber \\
  & \phi_i \rangle_{L^2(\PP)}, \quad \| \EE_{0 \mid \Scal} (\Delta_j) \|_2^2 = \sum_{i=1}^{\infty} \delta^2_{ji}, \quad j =1, \ldots, k.
\end{align}
Define the vectors $\deltab_j^{\downarrow}=(\delta_{j1}, \ldots, \delta_{jd})$ and  $\deltab_j^{\uparrow}=(\delta_{j(d+1)}, \ldots, \delta_{j \infty})$, so $\| \EE_{0 \mid \Scal} (\Delta_j) \|_2^2 = \| \deltab_j^{\downarrow} \|_2^2 + \| \deltab_j^{\uparrow} \|_2^2$ and we upper bound $\| \EE_{0 \mid \Scal} (\Delta_j) \|_2^2$ by separately upper bounding $\| \deltab^{\downarrow}_j \|_2^2$ and $\| \deltab_j^{\uparrow} \|_2^2 $. Using the expansion $C_{\alphab}(\sbb, \sbb') = \sum_{j=1}^{\infty} \mu_j \phi_j(\sbb) \phi_j(\sbb')$ for any $\sbb, \sbb' \in \Dcal$, we have the following upper bound for $\| \deltab_j^{\uparrow} \|_2^2$:

\begin{align}
  \label{eq:ba2}
  \| \deltab_j^{\uparrow} \|_2^2  = \frac{\mu_{d+1}}{\mu_{d+1}} \sum_{i=d+1}^{\infty} \delta_{ji}^2 \leq \mu_{d+1} \sum_{i=d+1}^{\infty} \frac{\delta_{ji}^2} {\mu_{i}} \overset{(i)}{\leq} \mu_{d+1} \| \EE_{0 \mid \Scal} (\Delta_j) \|_{\HH}^2 \overset{(ii)}{\leq} \mu_{d+1} \| w_0 \|^2_{\HH},
\end{align}
where $(i)$ follows because $\| \EE_{0 \mid \Scal} (\Delta_j) \|_{\HH}^2 = \sum_{i=1}^{\infty} \delta_{ji}^2/\mu_{i}$ and $(ii)$ follows from \eqref{eq:17}.

We then derive an upper bound for $\| \deltab_j^{\downarrow} \|_2^2$. Let $\Mb = \diag(\mu_1, \ldots, \mu_d)\in \RR^{d \times d}$, $\Phib^j \in \RR^{m \times d}$ be a matrix such that
\begin{align}
  \label{eq:ba3}
  \Phib^j_{ih} =  \phi_{h}(\sbb_{ji}), \quad i = 1, \ldots, m, \quad h = 1, \ldots, d, \quad j = 1, \ldots, k,
\end{align}
 $w_0 = \sum_{i=1}^{\infty} \theta_i \phi_i$, and the tail error vector $\vb_j = (v_{j1},\ldots, v_{jm})^T \in \RR^m$ ($j=1, \ldots, k$) such that
\begin{align*}
  v_{ji} = \sum_{h = d+1}^{\infty} \delta_{jh} \phi_h(\sbb_{ji}), \quad i = 1, \ldots, m.
\end{align*}
For any $g \in \{1, \ldots, d\}$, taking the $\HH$-inner product with respect $\phi_g$ in \eqref{eq:17} yields
\begin{align}
  \label{eq:18}
  & \left\langle  \left( \frac{1}{m} \sum_{i=1}^m \xi_{\sbb_{ji}} \otimes \xi_{\sbb_{ji}}  +  \tfrac{\tau^2}{km} \Ib \right) \EE_{0 \mid \Scal} (\Delta_j), \phi_g \right\rangle_{\HH}   \nonumber \\
  & = - \frac{\tau^2}{km} \langle w_0, \phi_g  \rangle_{\HH} = - \frac{\tau^2}{km}  \frac{\theta_g}{\mu_g}, \quad j = 1, \ldots, k.
\end{align}
Expanding the left hand side in \eqref{eq:18}, we obtain that
\begin{align*}
 &\frac{1}{m}\sum_{i=1}^m \langle \phi_g, \xi_{\sbb_{ji}} \rangle_{\HH} \EE_{0 \mid \Scal } \left\{ \Delta_j(\sbb_{ji}) \right\} + \frac{\tau^2}{km} \langle \phi_g, \EE_{0 \mid \Scal } (\Delta_j)  \rangle_{\HH}
  \nonumber \\
  &= \frac{1}{m}\sum_{i=1}^m \phi_g(\sbb_{ji}) \EE_{0 \mid \Scal} \left\{ \Delta_j(\sbb_{ji}) \right\} + \frac{\tau^2}{km} \frac{\delta_{jg}}{\mu_g}.
\end{align*}
The term $\frac{1}{m}\sum_{i=1}^m \phi_g(\sbb_{ji}) \EE_{0 \mid \Scal} \left\{ \Delta_j(\sbb_{ji}) \right\} $ on the right hand side is
\begin{align}
  \label{eq:20}
  &=   \frac{1}{m}\sum_{i=1}^m \Phib_{ig}^j \sum_{h = 1}^{d} \delta_{jh} \phi_h(\sbb_{ji}) + \frac{1}{m}\sum_{i=1}^m \Phib_{ig}^j \sum_{h = d+1}^{\infty} \delta_{jh} \phi_h(\sbb_{ji})
  \nonumber \\
  &=  \frac{1}{m} \sum_{h = 1}^{d}  \delta_{jh} \sum_{i=1}^m \Phib_{ig}^j  \Phib_{ih}^j + \frac{1}{m} \sum_{i=1}^m \Phib_{ig}^j v_{ji} \nonumber\\
  &= \frac{1}{m} \sum_{h = 1}^{d}  \delta_{jh} \left( \Phib^{j^T}  \Phib^j \right)_{gh} + \frac{1}{m} \sum_{i=1}^m \left(  \Phib^{j^T} v_{j}\right)_g \nonumber\\
  &= \frac{1}{m}\left(\Phib^{j^T}  \Phib^j \deltab^{\downarrow}  \right)_g + \frac{1}{m} \left(  \Phib^{j^T} \vb_j\right)_g.
\end{align}
Substitute \eqref{eq:20} in \eqref{eq:18} for $g = 1, \ldots, d$ to obtain that
\begin{align}
  \label{eq:21}
  \frac{1}{m}\Phib^{j^T}  \Phib^j \deltab_j^{\downarrow} + \frac{1}{m}\Phib^{j^T} \vb_j + \frac{\tau^2}{km} \Mb^{-1} \deltab_j^{\downarrow} &= - \frac{\tau^2}{km} \Mb^{-1} \thetab^{\downarrow} \nonumber\\
  \left( \frac{1}{m}\Phib^{j^T}  \Phib^j + \frac{\tau^2}{km} \Mb^{-1}  \right)  \deltab_j^{\downarrow} &= - \frac{\tau^2}{km} \Mb^{-1} \thetab^{\downarrow}  -  \frac{1}{m}\Phib^{j^T} \vb_j.
\end{align}

The proof is completed by showing that the right hand side expression in \eqref{eq:21} gives an upper bound for $\| \deltab_j^{\downarrow} \|_2^2$. Define $\Qb = \left( \Ib + \frac{\tau^2}{km} \Mb^{-1} \right)^{1/2}$, then
\begin{align*}
  &\frac{1}{m}\Phib^{j^T}  \Phib^j + \frac{\tau^2}{km} \Mb^{-1} = \Ib +   \frac{\tau^2}{km} \Mb^{-1} +   \frac{1}{m}\Phib^{j^T}  \Phib^j - \Ib \\
  &= \Qb \left\{ \Ib + \Qb^{-1} \left( \frac{1}{m}\Phib^{j^T}  \Phib^j - \Ib \right) \Qb^{-1}\right\} \Qb
\end{align*}
and using this in \eqref{eq:21} gives
\begin{align}
  \label{eq:23}
   \left\{ \Ib + \Qb^{-1} \left( \frac{1}{m}\Phib^{j^T}  \Phib^j - \Ib \right) \Qb^{-1}\right\} \Qb  \deltab_j^{\downarrow} &= - \frac{\tau^2}{km} \Qb^{-1} \Mb^{-1} \thetab^{\downarrow}  -  \frac{1}{m} \Qb^{-1} \Phib^{j^T} \vb_j.
\end{align}
Now we define the $\PP$-measureable event
\begin{align}\label{eset1}
& \Ecal_1 = \left\{ \vertiii{ \Qb^{-1}\left( \frac{1}{m}\Phib^{j^T}  \Phib^j - \Ib \right) \Qb^{-1} } \leq 1/2\right\},
\end{align}
where $\vertiii{\cdot}$ is the matrix operator norm.  We have that $\Ib + \Qb^{-1} \left( \frac{1}{m}\Phib^{j^T}  \Phib^j - \Ib \right) \Qb^{-1} \succeq (1/2) \Ib$  whenever $\Ecal_1$ occurs. Furthermore, when $\Ecal_1$ occurs, \eqref{eq:23} implies that
\begin{align*}
  & \| \deltab_j^{\downarrow} \|_2^2 \leq  \| \Qb \deltab_j^{\downarrow} \|_2^2 \leq 4 \left\| \frac{\tau^2}{km} \Qb^{-1} \Mb^{-1} \thetab^{\downarrow} + \frac{1}{m} \Qb^{-1} \Phib^{j^T} \vb_j\right\|_2^2 \\
  & \leq 8 \left\| \frac{\tau^2}{km} \Qb^{-1} \Mb^{-1} \thetab^{\downarrow} \right\|_2^2 + 8 \left\| \frac{1}{m} \Qb^{-1} \Phib^{j^T} \vb_j\right\|_2^2,
\end{align*}
where the last inequality follows because $(a + b)^2 \leq 2 a^2 + 2 b^2$ for any $a, b \in \RR$.

Since $\Ecal_1$ is $\PP$-measureable,  $\EE_{\Scal} \left(  \| \deltab_j^{\downarrow} \|_2^2 \right)  =   \EE_{\Scal} \left\{ \| \deltab_j^{\downarrow} \|_2^2 \one\left( \Ecal_1 \right)  \right\} + \EE_{\Scal} \left\{ \| \deltab_j^{\downarrow} \|_2^2 \one\left( \Ecal_1^c \right)  \right\} $ and the previous display gives
\begin{align}
  \label{eq:29}
      \EE_{\Scal} \left\{ \| \deltab^{\downarrow}_j \|_2^2 \one\left( \Ecal_1 \right)  \right\} \leq 8 \left\| \frac{\tau^2}{km} \Qb^{-1} \Mb^{-1} \thetab^{\downarrow} \right\|_2^2 + 8  \EE_{\Scal}  \left\| \frac{1}{m} \Qb^{-1} \Phib^{j^T} \vb_j\right\|_2^2.
\end{align}
From Lemma 10 in \citet{Zhaetal15}, we have that under our assumptions A.1-A.3, there exists a universal constant $A>0$ such that
\begin{align}
\label{eq:26}
& \left\| \frac{\tau^2}{km} \Qb^{-1} \Mb^{-1} \thetab^{\downarrow} \right\|_2^2 \leq \frac{\tau^2}{km} \| w_0 \|^2_{\HH}, \quad \nonumber\\
&\EE_{\Scal}  \left\| \frac{1}{m} \Qb^{-1} \Phib^{j^T} \vb_j\right\|_2^2 \leq  \frac{km}{\tau^2}\rho^4 \tr(C_{\alphab}) \tr(C_{\alphab}^{d}) \| w_0 \|^2_{\HH}, \nonumber \\
& \PP \left( \Ecal_1^c \right) \leq \left\{A \max \left( \sqrt{\max(r, \log d)}, \frac{\max(r, \log d)}{m^{1/2 - 1/r}} \right)  \frac{ \rho^2 \gamma(\tfrac{\tau^2_0}{km})}{\sqrt{m}} \right\}^r \nonumber \\
&\quad =  \left\{  \frac{A b(m, d, r) \rho^2 \gamma(\tfrac{\tau^2_0}{km})}{\sqrt{m}} \right\}^r.
\end{align}
Since $\mu_1\geq \mu_2\geq \ldots \geq 0$, the optimality condition in \eqref{eq:17} implies that
\begin{align}
  \label{eq:31}
  \left\| \EE_{0 \mid \Scal}(\Delta_j) \right\|^2_2 = \frac{\mu_1} {\mu_1}\sum_{i=1}^{\infty} \delta_{ji} \phi_i \leq  \mu_1 \sum_{i=1}^{\infty} \frac{\delta_{ji}}{\mu_i} \phi_i =  \mu_1  \| \EE_{0 \mid \Scal}(\Delta_j) \|^2_{\HH} \leq \mu_1 \|w_0 \|_{\HH}^2.
\end{align}
Using the shorthand \eqref{eq:26} and \eqref{eq:31}, we obtain that
\begin{align}
  \label{eq:32}
  \EE_{\Scal} \left\{ \| \deltab_j^{\downarrow} \|_2^2 \one\left( \Ecal_1^c \right)  \right\} \leq    \EE_{\Scal} \left\{ \| \EE_{0 \mid \Scal} (\Delta_j )\|_2^2 \one\left( \Ecal_1^c \right)  \right\} \leq \PP (\Ecal_1^c) \mu_1 \| w_0 \|_{\HH}^2.
\end{align}
Combining  \eqref{eq:29} and  \eqref{eq:32} gives
\begin{align}
  \label{eq:31a}
  \EE_{\Scal}(\| \deltab_j \|_2^2) \leq   & \frac{8\tau^2}{km} \| w_0 \|_{\HH}^2 +  \frac{8km}{\tau^2_0} \rho^4 \tr(C_{\alphab}) \tr(C_{\alphab}^{d}) \| w_0 \|^2_{\HH} \nonumber\\
  & +
  \left\{  \frac{A b(m, d, r) \rho^2 \gamma(\tfrac{\tau^2_0}{km})}{\sqrt{m}} \right\}^r \mu_1 \| w_0 \|^2_{\HH}.
\end{align}
Finally, we use that $\|\EE_{0 \mid \Scal}(\Delta) \|_2^2 \leq k^{-1} \sum_{j=1}^k \|  \EE_{0 \mid \Scal}(\Delta_j) \|_2^2 = k^{-1} \sum_{j=1}^k \| \deltab_j\|_2^2$ to obtain that
\begin{align}
  \label{eq:32a}
 & \EE_{\Scal}(\|\EE_{0 \mid \Scal}(\Delta) \|_2^2) \leq   \frac{8\tau^2}{km} \| w_0 \|_{\HH}^2 +  \frac{8km}{\tau^2_0} \rho^4 \tr(C_{\alphab}) \tr(C_{\alphab}^{d})  \| w_0 \|^2_{\HH} \nonumber\\
  &\qquad +
  \left\{  \frac{A b(m, d, r) \rho^2 \gamma(\tfrac{\tau^2_0}{km})}{\sqrt{m}} \right\}^r \mu_1 \| w_0 \|^2_{\HH} \nonumber\\
  &=  \frac{8\tau^2}{n} \| w_0 \|_{\HH}^2  + \| w_0 \|_{\HH}^2 \left[\frac{8n}{\tau^2_0} \rho^4 \tr(C_{\alphab}) \tr(C_{\alphab}^{d}) +
  \mu_1 \left\{  \frac{A b(m, d, r) \rho^2 \gamma(\tfrac{\tau^2_0}{n})}{\sqrt{m}} \right\}^r \right]  ,
\end{align}
where we have replaced $km$ by $n$ in the last equality. Taking the infimum over $d \in \NN$ leads to the proof.
\end{proof}

\subsubsection{An upper bound for the first variance term}
\label{sec:lemma-7-zhangs}

The following lemma provides an upper bound the first part of the variance term in \eqref{eq:eap4}.
\begin{lemma}\label{lem-sup-2}
If Assumptions A.1--A.3 in the main paper hold, then
\begin{align*}
& \tau^2 \EE_{\sbb^*} \EE_{\Scal} \left\{ \cb^T_* (k \Lb + \tau^2 \Ib)^{-2} \cb_* \right\} \leq \nonumber\\
&\frac{2 n + 4 \| w_0 \|^2_{\HH} }{k} \inf_{d\in \NN} \left[ \mu_{d+1} + 12 \frac{n}{\tau^2} \rho^4 \tr(C_{\alphab}) \tr(C_{\alphab}^d)  + \left\{  \frac{A b(m, d, r) \rho^2 \gamma(\tfrac{\tau^2_0}{n})}{\sqrt{m}} \right\}^r \right]  + \nonumber\\
&\frac{12}{k} \frac{\tau^2}{n} \| w_0 \|_{\HH}^2 + 12 \frac{\tau^2}{n} \gamma \left( \frac{\tau^2}{n} \right).
\end{align*}
\end{lemma}


\begin{proof}
Continuing from \eqref{eq:12}, we start by finding an upper bound for $\EE_{0 \mid \Scal} \| \Delta_j \|_{\HH}^2$, which is required later to upper bound $\EE_{0} \| \Delta_j \|_{\HH}^2$. From \eqref{eq:12} we have
\begin{align}
  \label{eq:vv1}
  \EE_{0 \mid \Scal} \| \Delta_j \|_{\HH}^2  \leq 2 \EE_{0 \mid \Scal} \| \tilde w_j |_{\HH}^2 +  2\| w_0\|_{\HH}^2.
\end{align}
An upper bound for $\EE_{0 \mid \Scal} \| \tilde w_j |_{\HH}^2$ gives the desired bound. Using the objective in \eqref{eq:ba4},
\begin{align}
    \label{eq:vv12}
  \frac{1}{2} \| \tilde w_j \|_{\HH}^2 \overset{(i)}{\leq}  \sum_{i=1}^m \frac{\left\{ \tilde w_j(\sbb_{ji}) - y(\sbb_{ji}) \right\}^2}{2\tau^2/k} + \frac{1}{2} \| \tilde w_j \|_{\HH}^2 \overset{(ii)}{\leq}
  \sum_{i=1}^m \frac{\left\{ w_0(\sbb_{ji}) - y(\sbb_{ji}) \right\}^2} {2\tau^2/k} + \frac{1}{2} \| w_0 \|_{\HH}^2 ,
\end{align}
where $(i)$ follows because the term inside the summation is non-negative and $(ii)$ follows because $\tilde w_j$ minimizes the objective. Since $w(\sbb_{ji}) - y(\sbb_{ji}) = - \epsilon(\sbb_{ji})$ and  $\EE_{0 \mid \Scal}\{\epsilon^2(\sbb_{ji})\}\leq \tau^2$ by Assumption A.2, \eqref{eq:vv12} reduces to
\begin{align}
  \label{eq:vv2}
  \EE_{0 \mid \Scal}  \| \tilde w_j \|_{\HH}^2 \leq \frac{k} {\tau^2}  \sum_{i=1}^m \EE_{0 \mid \Scal} \left\{ \epsilon(\sbb_{ji}) \right\}^2 + \| w_0 \|_{\HH}^2 \leq  km + \| w_0 \|_{\HH}^2 .
\end{align}
 Substituting  \eqref{eq:vv2} in \eqref{eq:vv1} gives
\begin{align}
  \label{eq:vv3}
  \EE_{0 \mid \Scal} \| \Delta_j \|_{\HH}^2  \leq 
   2 km + 4 \| w_0\|_{\HH}^2 .
\end{align}

First notice that
\begin{align}
  \label{eq:34}
  \tau^2  \EE_{\sbb^*} \EE_{\Scal} \left\{ \cb^T_* (k \Lb + \tau^2 \Ib)^{-2} \cb_* \right\} = \frac{1}{k^2} \sum_{j=1}^k \tau^2 \EE_{\sbb^*} \EE_{\Scal} \left\{ \cb_{j*}^T \left(\Cb_{j,j} + \tfrac{\tau^2}{k} \Ib\right)^{-2} \cb_{j*} \right\} .
\end{align}
and from \eqref{eq:eap4} we have
\begin{align}
\label{eqnew:34a}
& \tau^2  \EE_{\sbb^*} \EE_{\Scal} \left\{ \cb_{j*}^T \left(\Cb_{j,j} + \tfrac{\tau^2}{k} \Ib\right)^{-2} \cb_{j*} \right\}
=  \EE_{\sbb^*} \EE_{\Scal} \var_{0|\Scal} \left\{ \cb_{j*}^T \left(\Cb_{j,j} + \tfrac{\tau^2}{k} \Ib\right)^{-1} \yb_j \right\} \nonumber \\
& \leq \EE_{\sbb^*} \EE_{\Scal} \EE_{0|\Scal} \left\{\cb_{j*}^T \left(\Cb_{j,j} + \tfrac{\tau^2}{k} \Ib\right)^{-1}\yb_j - w_0(\sbb^*)\right\}^2  = \EE_{\sbb^*}\EE_{\Scal} \EE_{0|\Scal} \|\Delta_j\|_2^2.
\end{align}
Substituting \eqref{eqnew:34a} to \eqref{eq:34} leads to
\begin{align}
\label{eqnew:34b}
& \tau^2 \EE_{\sbb^*} \EE_{\Scal} \left\{ \cb^T_* (k \Lb + \tau^2 \Ib)^{-2} \cb_* \right\} \leq \EE_{\sbb^*}\left\{\frac{1}{k^2} \sum_{j=1}^k  \EE_{\Scal}\EE_{0|\Scal} \|\Delta_j\|_2^2\right\}.
\end{align}

We then find an upper bound for $\EE_{\Scal}\EE_{0|\Scal} \| \Delta_j \|^2_2$ by following similar steps to the proof of Lemma \ref{lem-sup-1}. Let $\deltab_j \in L_2(\NN)$ be the expansion of $\Delta_j$ in the basis $\{\phi_i\}_{i=1}^{\infty}$, so that $\Delta_j = \sum_{i=1}^{\infty} \delta_{ji} \phi_i$ (the $\deltab_j$  sequence here is different from the one in the previous section). Similar to Section \ref{sec:lemma-6-zhangs}, choose a fixed $d \in \NN$ and truncate $\Delta_j$ by defining $\Delta^{\downarrow}_j$, $\Delta^{\uparrow}_j$, $\deltab^{\downarrow}_j$, and $\deltab^{\uparrow}_j$ as
\begin{align*}
&  \Delta^{\downarrow}_j = \sum_{i=1}^d \delta_{ji} \phi_i, \quad
  \Delta^{\uparrow}_j = \sum_{i=d+1}^{\infty} \delta_{ji} \phi_i = \Delta_j - \Delta_j^{\downarrow}, \quad  \\ 
&  \deltab_{j}^{\downarrow} = (\delta_{j1}, \ldots, \delta_{jd}), \quad \deltab_{j}^{\uparrow} = (\delta_{j(d+1)}, \ldots, \delta_{j \infty}).
\end{align*}
The orthonormality of $\{\phi_i\}_{i=1}^{\infty}$  implies that
\begin{align}
  \label{eq:v3}
 \EE_{\Scal} \EE_{0|\Scal} \| \Delta_j \|_2^2 &= \EE_{\Scal} \EE_{0|\Scal} \| \Delta^{\downarrow}_j \|_2^2  +  \EE_{\Scal} \EE_{0 |\Scal} \| \Delta^{\uparrow}_j \|_2^2 \nonumber \\
 &=
  \EE_{\Scal} \EE_{0|\Scal} \| \deltab^{\downarrow}_j \|_2^2  + \EE_{\Scal} \EE_{0|\Scal} \| \deltab^{\uparrow}_j \|_2^2 .
\end{align}
First, the upper bound for $\EE_{0|\Scal} \| \deltab^{\uparrow}_j \|_2^2 $ follows from \eqref{eq:ba2},
\begin{align*}
  & \EE_{0|\Scal} \| \Delta^{\uparrow}_j \|_2^2 =  \sum_{i=d+1}^{\infty} \EE_{0|\Scal} (\delta_{ji}^2) = \mu_{d+1}\sum_{i=d+1}^{\infty} \frac{\EE_{0|\Scal} (\delta_{ji}^2)}{\mu_{d+1}} \leq \mu_{d+1}\sum_{i=d+1}^{\infty} \frac{\EE_{0|\Scal} (\delta_{ji}^2)}{\mu_{i}}  \\
  &= \mu_{d+1} \EE_{0|\Scal} \| \Delta_j^{\uparrow} \|^2_{\HH} \leq \mu_{d+1} \EE_{0|\Scal} \| \Delta_j \|^2_{\HH},
\end{align*}
and using \eqref{eq:vv3},
\begin{align}
  \label{eq:v31}
    \EE_{0 |\Scal} \| \Delta^{\uparrow}_j \|_2^2 \leq  \mu_{d+1}  (2km + 4 \| w_0\|_{\HH}^2).
\end{align}

We now find an upper bound for $\EE_{\Scal} \EE_{0|\Scal} \| \Delta_j^{\downarrow} \|^2_{2}$.
Following Section \ref{sec:lemma-6-zhangs}, define the error vector  $\vb_j=(v_{j1},\ldots,v_{jm})^T \in \RR^m$ with $v_{ji} = \sum_{h=d+1}^\infty \delta_{ji} \phi_h(\sbb_{ji})$ ($i=1, \ldots, m$), and $\Mb = \diag(\mu_1, \ldots, \mu_d)$. From \eqref{eq:ba4} and \eqref{eq:ba6}, $\tilde w_j(\cdot)$ in \eqref{eq:12} satisfies
\begin{align}
  \label{eq:v1}
  \frac{1}{m}\sum_{i=1}^m  \langle \xi_{\sbb_{ji}}, \tilde w_j - w_0 - \epsilon \rangle_{\HH} \, \xi_{\sbb_{ji}}  + \frac{\tau^2}{km}   \tilde w_j = 0.
\end{align}
For any $g \in \{1, \ldots, d\}$, taking the $\HH$-inner product with respect $\phi_g$ in \eqref{eq:v1} to obtain that
\begin{align*}
  &\frac{1}{m}\sum_{i=1}^m  \langle \xi_{\sbb_{ji}}, \Delta_j - \epsilon \rangle_{\HH} \, \langle\xi_{\sbb_{ji}}, \phi_g \rangle_{\HH}  + \frac{\tau^2}{km}   \langle \Delta_j + w_0, \phi_g \rangle_{\HH} = \\
  & \frac{1}{m}\sum_{i=1}^m \left\{ \Delta_j(\sbb_{ji}) - \epsilon(\sbb_{ji})  \right\}   \phi_g(\sbb_{ji})   + \frac{\tau^2}{km}   \frac{\delta_{jg}}{\mu_g}  + \frac{\tau^2}{km}   \frac{\theta_g}{\mu_g}  = 0,  \\
  &\frac{1}{m}\sum_{i=1}^m \left\{ \sum_{h=1}^d \delta_{jh} \phi_h(\sbb_{ji}) + \sum_{h=d+1}^\infty \delta_{jh} \phi_h(\sbb_{ji}) - \epsilon(\sbb_{ji})  \right\}   \phi_g(\sbb_{ji})   + \frac{\tau^2}{km}   \frac{\delta_{jg}}{\mu_g}   = - \frac{\tau^2}{km}   \frac{\theta_g}{\mu_g} ,  \\
  &\frac{1}{m}\sum_{h=1}^d \left\{ \sum_{i=1}^m \phi_h(\sbb_{ji}) \phi_g(\sbb_{ji})  \right\} \delta_{jh}   +  \frac{1}{m}\sum_{i=1}^m \left\{ v_{ji} - \epsilon(\sbb_{ji})  \right\}   \phi_g(\sbb_{ji})   + \frac{\tau^2}{km}   \frac{\delta_{jg}}{\mu_g}   = - \frac{\tau^2}{km}   \frac{\theta_g}{\mu_g},   \\
  &\frac{1}{m} \left( \Phib^{j^T} \Phib^{j} \deltab_j^{\downarrow} \right)_{g}   +  \frac{1}{m} \left\{  \Phib^{j^T} (\vb_j - \epsilonb_{j}) \right\}_g   + \frac{\tau^2}{km}   (\Mb^{-1} \deltab^{\downarrow}_{j})_g  = - \frac{\tau^2}{km}   (\Mb^{-1} \thetab^{\downarrow} )_g. 
\end{align*}
Writing this equation in the matrix form yields,
\begin{align}
  \label{eq:v5}
  \left(  \frac{1}{m} \Phib^{j^T} \Phib^{j} + \frac{\tau^2}{km}  \Mb^{-1}  \right)\deltab_j^{\downarrow}   = - \frac{\tau^2}{km} \Mb^{-1} \thetab^{\downarrow} - \frac{1}{m} \Phib^{j^T} \vb_j + \frac{1}{m} \Phib^{j^T}  \epsilonb_{j} .
\end{align}
Following Section \ref{sec:lemma-6-zhangs}, by defining $\Qb = (\Ib + \frac{\tau^2}{km} \Mb^{-1})^{1/2}$, \eqref{eq:v5} reduces to
\begin{align}
\label{eq:v15}
 & \left\{ \Ib + \Qb^{-1} \left(  \frac{1}{m} \Phib^{j^T} \Phib^{j} - \Ib \right) \Qb^{-1} \right\} \Qb\deltab_j^{\downarrow}   \nonumber\\
 &= - \frac{\tau^2}{km} \Qb^{-1} \Mb^{-1} \thetab^{\downarrow} - \frac{1}{m} \Qb^{-1} \Phib^{j^T} \vb_j + \frac{1}{m} \Qb^{-1} \Phib^{j^T}  \epsilonb_{j} .
\end{align}

On the event $\Ecal_1$ defined as in \eqref{eset1}, we have that $\Ib + \Qb^{-1} \left( \frac{1}{m}\Phib^{j^T}  \Phib^j - \Ib \right) \Qb^{-1} \succeq (1/2) \Ib$. Furthermore, when $\Ecal_1$ occurs, \eqref{eq:v15} implies that
\begin{align*}
& \| \Delta_j^{\downarrow} \|_2^2 \leq  \| \Qb \deltab_j^{\downarrow} \|_2^2 \leq 4 \left\| - \frac{\tau^2}{km} \Qb^{-1} \Mb^{-1} \thetab^{\downarrow} - \frac{1}{m} \Qb^{-1} \Phib^{j^T} \vb_j + \frac{1}{m} \Qb^{-1} \Phib^{j^T}  \epsilonb_{j}  \right\|_2^2\\
&\leq 12  \left\| \frac{\tau^2}{km} \Qb^{-1} \Mb^{-1} \thetab^{\downarrow} \right\|_2^2 + 12 \left\| \frac{1}{m} \Qb^{-1} \Phib^{j^T} \vb_j \right\|_2^2  +  12 \left\| \frac{1}{m} \Qb^{-1} \Phib^{j^T}  \epsilonb_{j}  \right\|_2^2,
\end{align*}
where the last inequality follows because $(a + b + c)^2 \leq 3 a^2 + 3 b^2 + 3 c^2$ for any $a, b, c \in \RR$.
Since $\Ecal_1$ is $\PP$-measureable,  $\EE_{0|\Scal} \left(  \| \Delta_j^{\downarrow} \|_2^2 \right)  =   \EE_{0|\Scal} \left\{ \| \Delta_j^{\downarrow} \|_2^2 \one\left( \Ecal_1 \right)  \right\} + \EE_{0|\Scal} \left\{ \| \Delta_j^{\downarrow} \|_2^2 \one\left( \Ecal_1^c \right)  \right\} $. If the event $\Ecal_1$ occurs,  then the upper bounds for the first term
and the last two terms in the last inequality are given by Lemmas 10 and 7 of \citet{Zhaetal15}, respectively, and we have that
\begin{align}
  \label{eq:v20}
  \left\| \frac{\tau^2}{km} \Qb^{-1} \Mb^{-1} \thetab^{\downarrow} \right\|_2^2 &\leq \frac{\tau^2}{km} \| w_0 \|_{\HH}^2, \nonumber\\
  \EE_{\Scal}  \left\| \frac{1}{m} \Qb^{-1} \Phib^{j^T} \vb_j \right\|_2^2 &\leq \frac{km}{\tau^2} \rho^4 \tr(C_{\alphab}) \tr(C_{\alphab}^d) \left( 2 km  +  4 \| w_0 \|_{\HH}^2 \right), \nonumber \\
  \EE_{\Scal} \EE_{0|\Scal}  \left\| \frac{1}{m} \Qb^{-1} \Phib^{j^T} \epsilonb_j \right\|_2^2 &\leq \frac{1}{m^2} \sum_{h=1}^d \sum_{i=1}^m \frac{1}{1 + \tfrac{\tau^2}{km} \tfrac{1}{\mu_h}} \EE_{\Scal} \EE_{0|\Scal} \left\{ \phi_h^2 (\sbb_{ji}) \epsilon^2(\sbb_{ji}) \right\}.
\end{align}
Since the error $\epsilon(\cdot)$ and $w(\cdot)$ are independent, by Assumption A.3,
$$\EE_{\Scal} \EE_{0|\Scal} \left\{ \phi_h^2 (\sbb_{ji}) \epsilon^2(\sbb_{ji}) \right\} =  \EE_{\Scal}  \left\{ \phi_h^2 (\sbb_{ji})  \right\} \EE_{0|\Scal} \left\{  \epsilon^2(\sbb_{ji}) \right\} \leq \tau^2$$
and the last inequality in \eqref{eq:v20} simplifies to
\begin{align*}
 \EE_{\Scal} \EE_{0|\Scal} \left\| \frac{1}{m} \Qb^{-1} \Phib^{j^T} \epsilonb_j \right\|_2^2 \leq \frac{\tau^2}{m} \sum_{h=1}^d \frac{1}{1 + \tfrac{\tau^2}{km} \tfrac{1}{\mu_h}} \leq
  \frac{\tau^2}{m} \gamma \left( \frac{\tau^2}{km} \right).
\end{align*}
Hence when the event $\Ecal_1$ occurs,
\begin{align}
  \label{eq:v23}
   & \EE_{\Scal} \EE_{0|\Scal} \left\{ \| \Delta_j^{\downarrow} \|_2^2 \one (\Ecal_1) \right\} \leq \nonumber \\
   &12 \frac{\tau^2}{km} \| w_0 \|_{\HH}^2 + 12 \frac{km}{\tau^2} \rho^4 \tr(C_{\alphab}) \tr(C_{\alphab}^d) \left( 2 km  +  4 \| w_0 \|_{\HH}^2 \right) + 12 \frac{\tau^2}{m} \gamma \left( \frac{\tau^2}{km} \right).
\end{align}
If the event $\Ecal_1$ does not occur, then
\begin{align}
  \label{eq:v24}
  \EE_{\Scal} \EE_{0|\Scal} \left\{ \| \Delta_j^{\downarrow} \|_2^2 \one (\Ecal_1^c) \right\}  &\leq   \EE_{\Scal} \left\{ \one (\Ecal_1^c)  \EE_{0 \mid \Scal} \| \Delta_j^{\downarrow} \|_2^2   \right\} \overset{(i)}{\leq}
  \PP(\Ecal_1^c)  \left( 2 km + 4 \| w_0 \|_{\HH}^2 \right) \nonumber\\
  &\overset{(ii)}{=} \left\{  \frac{A b(m, d, r) \rho^2 \gamma(\tfrac{\tau^2_0}{km})}{\sqrt{m}} \right\}^r  \left( 2 km + 4 \| w_0 \|_{\HH}^2 \right)  ,
\end{align}
where $(i)$ follows from \eqref{eq:vv3} and $(ii)$ follows from \eqref{eq:26}. Substituting \eqref{eq:v23}, \eqref{eq:v24}, and \eqref{eq:v31} in \eqref{eq:v3} implies that
\begin{align}
  \label{eq:v25}
  & \EE_{\Scal} \EE_{0|\Scal}  \left\{ \| \Delta_j\|_2^2\right\}  \leq
  12 \frac{\tau^2}{km} \| w_0 \|_{\HH}^2 + 12 \frac{\tau^2}{m} \gamma \left( \frac{\tau^2}{km} \right)\nonumber\\
  & + \left[ \mu_{d+1} + 12 \frac{km}{\tau^2} \rho^4 \tr(C_{\alphab}) \tr(C_{\alphab}^d)  + \left\{  \frac{A b(m, d, r) \rho^2 \gamma(\tfrac{\tau^2_0}{km})}{\sqrt{m}} \right\}^r \right] \left( 2 km + 4 \| w_0 \|^2_{\HH} \right).
\end{align}
Therefore,  substituting \eqref{eq:v25} in \eqref{eqnew:34b} implies that
\begin{align}
  & \tau^2 \EE_{\sbb^*} \EE_{\Scal} \left\{ \cb_*^T (k \Lb + \tau^2 \Ib)^{-2} \cb_* \right\}  \leq \nonumber\\
  &\frac{2 n + 4 \| w_0 \|^2_{\HH}}{k}\left[ \mu_{d+1} + 12 \frac{n}{\tau^2} \rho^4 \tr(C_{\alphab}) \tr(C_{\alphab}^d)  + \left\{  \frac{A b(m, d, r) \rho^2 \gamma(\tfrac{\tau^2_0}{n})}{\sqrt{m}} \right\}^r \right]   + \nonumber\\
  &\frac{12}{k} \frac{\tau^2}{n} \| w_0 \|_{\HH}^2 + 12 \frac{\tau^2}{n} \gamma \left( \frac{\tau^2}{n} \right).
\end{align}
where we have replace $km$ by $n$. Taking the infimum over $d\in \NN$ leads to the proof.
\end{proof}

\subsubsection{An upper bound for the second variance term}
\label{sec:lemma-7-zhangs-2}

The following lemma provides an upper bound the second part of the variance term in \eqref{eq:eap4}.
\begin{lemma}\label{lem-sup-3}
If Assumptions A.1--A.3 in the main paper hold, then
\begin{align*}
&\EE_{\sbb^*} \EE_{\Scal} \overline v (\sbb^*) \leq
3 \frac{\tau^2}{n} \gamma \left( \frac{\tau^2}{n} \right) \\
 &\qquad + \underset{d \in \NN}{\inf} \,\left[ \left\{\frac{4n}{\tau^2}\tr(C_{\alphab})+1\right\}\tr(C_{\alphab}^d) + \tr(C_{\alphab}) \left\{  \frac{A b(m, d, r) \rho^2 \gamma(\tfrac{\tau^2_0}{n})}{\sqrt{m}} \right\}^r  \right].
\end{align*}
\end{lemma}

\begin{proof}
First we have the following relation between $\overline v$ and the subset variance $v_j$:
\begin{align}
  \label{eq:361}
   & \overline v(\sbb^*) = \left( \frac{1}{k} \sum_{j=1}^k v_j^{1/2}(\sbb^*) \right)^2 \leq \frac{1}{k}  \sum_{j=1}^k v_j(\sbb^*) \nonumber \\
   &=  \frac{1}{k}  \sum_{j=1}^k \left\{ C_{\alphab}(\sbb^*, \sbb^*) -  \cb^T_j(\sbb^*) \left(\Cb_{j, j} + \tfrac{\tau^2}{k} \Ib \right)^{-1} \cb_j(\sbb^*) \right\}  .
\end{align}
Since $C_{\alphab}(\sbb, \sbb') = \sum_{i=1}^{\infty} \mu_{i} \phi_i(\sbb) \phi_i(\sbb')$ for $\sbb, \sbb' \in \Dcal$, we have
\begin{align*}
  C_{\alphab}(\sbb^*, \sbb^*)  = \sum_{a=1}^{\infty} \mu_{a} \phi^2_a(\sbb^*), \quad \{\cb_j(\sbb^*)\}_i = \sum_{a=1}^{\infty} \mu_{a} \phi_a(\sbb_{ji}) \phi_a(\sbb^*), \quad i = 1, \ldots, m.
\end{align*}
These together with the orthogonality property of $\left\{\phi_i\right\}_{i=1}^{\infty}$ imply that
\begin{align}
  \label{eq:v2a}
  &~~ \EE_{\sbb^*} \EE_{\Scal} \left\{v_j(\sbb^*)\right\}  = \sum_{a=1}^{\infty} \mu_{a} \EE_{\sbb^*} \phi^2_a(\sbb^*)  \nonumber \\
  & \quad - \sum_{i=1}^{m} \sum_{i'=1}^{m} \sum_{a=1}^{\infty} \sum_{b=1}^{\infty}  \mu_{a} \mu_{b} \left\{ \left(\Cb_{j, j} + \tfrac{\tau^2}{k} \Ib \right)^{-1} \right\}_{i' i''} \nonumber\\
   & \quad \times \EE_{\Scal} \left[ \phi_a(\sbb_{ji}) \phi_b(\sbb_{ji'}) \EE_{\sbb^*} \left\{ \phi_a(\sbb^*) \phi_b(\sbb^*) \right\}\right]\nonumber\\
  &= \tr(C_{\alphab}) - \EE_{\Scal} \sum_{i=1}^{m} \sum_{i'=1}^{m} \sum_{a=1}^{\infty}  \mu^2_{a} \left\{ \left(\Cb_{j, j} + \tfrac{\tau^2}{k} \Ib \right)^{-1} \right\}_{i i'} \phi_a(\sbb_{ji}) \phi_a(\sbb_{ji'})\nonumber\\
  &= \sum_{a=1}^{d} \mu_{a} - \EE_{\Scal} \sum_{a=1}^{d}  \mu^2_{a} \left[ \sum_{i=1}^{m} \sum_{i'=1}^{m}  \left\{ \left(\Cb_{j, j} + \tfrac{\tau^2}{k} \Ib \right)^{-1}  \right\}_{i i'} \phi_a(\sbb_{ji})\phi_a(\sbb_{ji'}) \right]  + \nonumber\\
  &\quad \tr(C_{\alphab}^d) - \EE_{\Scal}
  \sum_{a=d+1}^{\infty}  \mu^2_{a}  \left[\sum_{i=1}^{m} \sum_{i'=1}^{m} \left\{ \left(\Cb_{j, j} + \tfrac{\tau^2}{k} \Ib \right)^{-1} \right\}_{i' i''}  \phi_a(\sbb_{ji})\phi_a(\sbb_{ji'})  \right] \nonumber\\
  &\overset{(i)}{\leq}  \EE_{\Scal} \sum_{a=1}^{d} \left\{ \mu_{a} -  \mu^2_{a} \phib^{j^T}_a (\Cb_{j, j} + \tfrac{\tau^2}{k} \Ib )^{-1} \phib^{j}_a \right\}  + \tr(C_{\alphab}^d),
\end{align}
where $ia$th element of the matrix $\Phib^j$ (defined in the proof of Lemma \ref{lem-sup-1}) is $\phi_a(\sbb_{ji})$, $\phib^{j}_a$ is the $a$th column of  $\Phib^j$, and $(i)$ follows because $\left(\Cb_{j, j} + \tfrac{\tau^2}{k} \Ib \right)$ is a positive definite matrix and $\phib^{j^T}_a \left(\Cb_{j, j} + \tfrac{\tau^2}{k} \Ib \right)^{-1} \phib^{j}_a \geq 0$.

Let $\Mb = \diag(\mu_1, \ldots, \mu_d)$ and $\Qb = \left(\Ib + \frac{\tau^2} {km}\Mb^{-1} \right)^{1/2}$ as defined in the proofs of Lemmas \ref{lem-sup-1} and \ref{lem-sup-2}. Define a $d \times d$ matrix $\Bb\equiv \Mb - \Mb \Phib^{j^T} \left(\Cb_{j, j} + \tfrac{\tau^2}{k} \Ib \right)^{-1} \Phib^{j} \Mb$, so that from \eqref{eq:v2a},
\begin{align}
\label{eq:defvj}
 &\tr(\Bb) = \sum_{a=1}^{d} \left\{ \mu_{a} -  \mu^2_{a} \phib^{j^T}_a \left(\Cb_{j, j} + \tfrac{\tau^2}{k} \Ib \right)^{-1} \phib^{j}_a \right\} , \quad \nonumber \\
 &\EE_{\sbb^*} \EE_{\Scal} \left\{v_j(\sbb^*)\right\} \leq \EE_{\Scal} \tr(\Bb) +  \tr(C_{\alphab}^d).
\end{align}
Let
\begin{align*}
 & \Cb_{j, j} = \Phib^j \Mb \Phib^{j^T} + \Phib^{j \uparrow} \Mb^{\uparrow} \Phib^{j \uparrow ^T} \equiv \Phib^j \Mb \Phib^{j^T} + \Cb_{j,j}^{\uparrow}, \\
 & \Mb^{\uparrow} = \diag(\mu_{d+1}, \ldots, \mu_{\infty}), \quad \Phib^{j \uparrow} = [\phib_{d+1}^j , \cdots, \phib_{\infty}^j],
\end{align*}
then the Woodbury formula \citep{Har97} and the definition of $\Qb$ imply that
\begin{align}
  \label{eq:v2b}
  \Bb &=  \left\{ \Mb^{-1} + \Phib^{j^T} \left( \Cb_{j,j}^{\uparrow} + \tfrac{\tau^2}{k} \Ib \right)^{-1} \Phib^{j}   \right\}^{-1} \nonumber\\
      &=  \frac{\tau^2}{km} \left\{ \Ib + \frac{\tau^2} {km}\Mb^{-1} + \frac{1}{m}\Phib^{j^T} \left(  \tfrac{k}{\tau^2} \Cb_{j,j}^{\uparrow} + \Ib \right)^{-1} \Phib^{j} - \Ib   \right\}^{-1} \nonumber \\
      &=  \frac{\tau^2}{km} \Qb^{-2} \left[ \Ib  +  \Qb^{-1}  \left\{ \frac{1}{m}\Phib^{j^T} \left(  \tfrac{k}{\tau^2} \Cb^{\uparrow}_{j,j} + \Ib \right)^{-1} \Phib^{j} - \Ib \right\} \Qb^{-1} \right]^{-1} .
\end{align}

Define the event $\Ecal_2 = \left\{ \tfrac{k}{\tau^2} \Cb_{j,j}^{\uparrow} \preceq \frac{1}{4}\Ib \right\}$. Since the matrix $ \Cb_{j,j}^{\uparrow}$ is nonnegative definite, we have the relation that
\begin{align*}
\left\{\tr\left(\tfrac{k}{\tau^2} \Cb_{j,j}^{\uparrow}\right)\leq \frac{1}{4} \right\} \subseteq \left\{\smax \left(\tfrac{k}{\tau^2} \Cb_{j,j}^{\uparrow}\right) \leq \frac{1}{4} \right\} \subseteq \Ecal_2,
\end{align*}
$\smax(\Ab)$ is the maximum eigenvalue of the square matrix $\Ab$. Therefore, by Markov's inequality, we have that
\begin{align}\label{traceerr1}
& \PP (\Ecal_2^c) \leq \PP \left\{\tr\left(\tfrac{k}{\tau^2} \Cb_{j,j}^{\uparrow}\right)>\frac{1}{4}\right\} \leq 4\EE_{\Scal} \tr\left(\tfrac{k}{\tau^2} \Cb_{j,j}^{\uparrow}\right)  \nonumber \\
& = \frac{4k}{\tau^2} \sum_{i=1}^m \sum_{a=d+1}^{\infty} \mu_a\EE_{\Scal}\phi^2_a(\sbb_{ji}) = \frac{4km}{\tau^2} \tr\left(C_{\alphab}^d\right).
\end{align}
Now on the event $\Ecal_1\cap \Ecal_2$ (with $\Ecal_1$ defined in \eqref{eset1}), we have that
\begin{align}\label{Bmatbound1}
& \Ib + \Qb^{-1}  \left\{ \frac{1}{m}\Phib^{j^T} \left(  \tfrac{k}{\tau^2} \Cb_{j,j}^{\uparrow} + \Ib \right)^{-1} \Phib^{j} - \Ib \right\} \Qb^{-1} \nonumber \\
& \stackrel{(i)}{\succeq} \Ib + \Qb^{-1}  \left\{ \frac{1}{m}\Phib^{j^T} \left( \frac{1}{4}\Ib + \Ib \right)^{-1} \Phib^{j} - \Ib \right\} \Qb^{-1} \nonumber \\
& = \Ib - \frac{1}{5} \Qb^{-2} + \frac{4}{5} \Qb^{-1}  \left\{ \frac{1}{m}\Phib^{j^T}\Phib^{j} - \Ib \right\} \Qb^{-1} \nonumber \\
& \stackrel{(ii)}{\succeq} \Ib - \frac{1}{5} \Ib - \frac{4}{5}\cdot \frac{1}{2}\Ib = \frac{2}{5}\Ib,
\end{align}
where (i) follows on the event $\Ecal_2$, and (ii) holds on the event $\Ecal_1$ and from the fact $\Qb^{-2}\preceq \Ib$.

Therefore, by combining \eqref{traceerr1}, \eqref{Bmatbound1}, and the upper bound for $\PP(\Ecal_1^c)$ given in \eqref{eq:26} under our assumptions, we obtain that
\begin{align} \label{eq:v2d1}
\EE_{\Scal} \tr(\Bb)  & \leq \EE_{\Scal} \left\{\tr(\Bb) \one(\Ecal_1\cap \Ecal_2)\right\} + \EE_{\Scal} \left[\tr(\Bb) \left\{\one(\Ecal_1^c) + \one(\Ecal_2^c)\right\} \right] \nonumber \\
& \overset{(i)}{\leq} \frac{5}{2} \frac{\tau^2}{km} \tr\left(\Qb^{-2}\right) + \tr(C_{\alphab}) \left\{\PP(\Ecal_1^c)+\PP(\Ecal_2^c)\right\} \nonumber \\
& \overset{(ii)}{\leq}  3 \frac{\tau^2}{n} \gamma\left( \frac{\tau^2}{n} \right) + \frac{4n}{\tau^2}\tr(C_{\alphab})\tr(C_{\alphab}^d) + \tr(C_{\alphab}) \left\{  \frac{A b(m, d, r) \rho^2 \gamma(\tfrac{\tau^2_0}{n})}{\sqrt{m}} \right\}^r,
\end{align}
where (i) follows from \eqref{Bmatbound1}, and (ii) follows from \eqref{traceerr1}, \eqref{eq:26}, and by replacing $km$ with $n$.

\eqref{eq:v2a}, \eqref{eq:v2b}, and \eqref{eq:v2d1} together yield
\begin{align}
\label{eq:vj2e}
& \EE_{\sbb^*}  \EE_{\Scal} \left\{\overline v_j(\sbb^*) \right\} \leq \EE_{\Scal} \tr(\Bb) + \tr \left(C_{\alphab}^d\right) \nonumber \\
& \leq 3 \frac{\tau^2}{n} \gamma\left( \frac{\tau^2}{n} \right) + \left\{\frac{4n}{\tau^2}\tr(C_{\alphab})+1\right\}\tr(C_{\alphab}^d) + \tr(C_{\alphab}) \left\{  \frac{A b(m, d, r) \rho^2 \gamma(\tfrac{\tau^2_0}{n})}{\sqrt{m}} \right\}^r .
\end{align}

Since the righthand side of \eqref{eq:vj2e} does not depend on $j$, a further upper bound for \eqref{eq:361} is given by
\begin{align}
  \label{eq:v2f}
  & \EE_{\sbb^*}  \EE_{\Scal} \left\{\overline v(\sbb^*) \right\}  \leq \frac{1}{k} \sum_{j=1}^k \EE_{\sbb^*} \EE_{\Scal} \left\{\overline v_j(\sbb^*) \right\} \nonumber \\
  & \leq 3 \frac{\tau^2}{n} \gamma \left( \frac{\tau^2}{n} \right) +  \left\{\frac{4n}{\tau^2}\tr(C_{\alphab})+1\right\}\tr(C_{\alphab}^d) + \tr(C_{\alphab}) \left\{  \frac{A b(m, d, r) \rho^2 \gamma(\tfrac{\tau^2_0}{n})}{\sqrt{m}} \right\}^r .
\end{align}
Taking the infimum over $d\in \NN$ leads to the proof.
\end{proof}

\subsection{Proof of Theorem 2}
The proof of parts (i)--(iii) are as follows. (i) Since $d^*$ is a constant integer and $k=o(n)$, we can take $m$ sufficiently large such that $n\geq m>\max(d^*,e^r)$. In the upper bounds of Theorem 1, we choose $d=n$ in every infimum to make the upper bounds larger. This implies that $\tr\left(C_{\alphab}^d\right)=0$, $\mu_{d+1}=0$, and $b(m,d,r)\leq \log n$. Also notice that in this case, $\gamma(a)\leq d^*$ for any $a>0$. Then, Theorem 1 implies that
\begin{align} \label{eq:case1a}
& \EE_{\sbb^*} \EE_{\Scal} \EE_{0|\Scal} \{\overline w(\sbb^*) - w_0(\sbb^*)\}^2  \nonumber \\
& \leq \left(8\|w_0\|_{\HH}^2 + 12k^{-1}\|w_0\|_{\HH}^2 + 15d^*\right)\frac{\tau^2}{n}  \nonumber \\
& ~~ + \left\{\mu_1\|w_0\|_{\HH}^2 + \frac{4\|w_0\|_{\HH}^2}{k} + \frac{2n}{k} + \tr(C_{\alphab}) \right\} \left(\frac{A\rho^2 d^* \log n}{\sqrt{n/k}}\right)^r \nonumber \\
&\leq O(n^{-1}) + \{1+o(1)\} \frac{2\left(A\rho^2 d^* \log n\right)^{r} k^{r/2-1}}{n^{r/2-1}} \nonumber \\
&= O(n^{-1}),
\end{align}
where the last equality follows from the condition on $k$.
\vspace{5mm}

\noindent (ii) In the upper bounds of Theorem 1, we choose $d=n^2$ in every infimum for sufficiently large $n$ such that $\log d = 2\log n > r$. Then
\begin{align}
\label{eq:case2a}
\mu_{d+1} & \leq c_{1\mu} \exp\left(-c_{2\mu} n^{2\kappa} \right) = O(n^{-4}), \nonumber \\
b(m,d,r) & \leq \max\left(\sqrt{\log d}, \frac{\log d}{m^{1/2-1/r}}\right) \leq \log d \leq 2\log n, \nonumber \\
\tr\left(C_{\alphab}^d\right) & = \sum_{i=n^2+1}^{\infty} \mu_i \leq \sum_{i=n^2+1}^{\infty} c_{1\mu} \exp\left(-c_{2\mu} i^\kappa \right) \leq c_{1\mu} \int_{n^2}^{\infty}   \exp\left(-c_{2\mu} z^\kappa \right) dz \nonumber  \\
& = c_{1\mu} \int_{n^{2\kappa}}^{\infty} \frac{1}{\kappa} t^{\frac{1}{\kappa}-1} \exp\left(-c_{2\mu} t \right) dt,
\end{align}
where in the last step, we use the change of variable $t=z^\kappa$. If $\kappa\geq 1$, then since $t\geq n^{2\kappa}\geq 1$, we have $t^{\frac{1}{\kappa}-1}\leq 1$. If $0<\kappa<1$, then there exists a large $n_0\in \NN$ that depends on only $c_{2\mu}$ and $\kappa$, such that for all $n\geq n_0$ and $t\geq n^{2\kappa}$, we have $t^{\frac{1}{\kappa}-1}\leq \exp(c_{2\mu}t/2)$. Therefore, in all cases,
\begin{align}
\label{eq:case2b}
\tr\left(C_{\alphab}^d\right) &\leq  \frac{c_{1\mu}}{\kappa} \int_{n^{2\kappa}}^{\infty}  \exp\left(-c_{2\mu} t /2 \right) dt = \frac{2c_{1\mu}}{c_{2\mu}\kappa} \exp\left(-c_{2\mu}n^{2\kappa} /2 \right) = O(n^{-4}).
\end{align}
Let $d_1= \left(\tfrac{2}{c_{2\mu}} \log n\right)^{1/\kappa}$. For sufficiently large $n$, based on the similar argument as above, $\gamma(\tau^2/n)$ can be bounded as
\begin{align}
\label{eq:case2c}
\gamma(\tau^2/n) & = \sum_{i=1}^{\infty} \frac{\mu_i}{\mu_i+ \frac{\tau^2}{n}} = \sum_{i=1}^{\lfloor d_1\rfloor + 1} \frac{\mu_i}{\mu_i+ \frac{\tau^2}{n}} + \sum_{i=\lfloor d_1 \rfloor +2}^{\infty} \frac{\mu_i}{\mu_i+ \frac{\tau^2}{n}}\nonumber \\
& \leq d_1 + 1 + \frac{n}{\tau^2} \sum_{i=\lfloor d_1 \rfloor +1}^{\infty} c_{1\mu} \exp\left(-c_{2\mu} i^{\kappa} \right) \nonumber \\
& \leq d_1 + 1  +  \frac{n}{\tau^2} \int_{d_1}^{\infty} c_{1\mu} \exp\left(-c_{2\mu} z^{\kappa} \right) dz \nonumber  \\
& = d_1 + 1  +  \frac{n c_{1\mu}}{\tau^2 \kappa } \int_{d_1^{\kappa}}^{\infty} t^{\frac{1}{\kappa}-1}\exp\left(-c_{2\mu} t \right) dt \nonumber \\
& \leq d_1 + 1  +  \frac{n c_{1\mu}}{\tau^2 \kappa } \int_{d_1^{\kappa}}^{\infty} \exp\left(-c_{2\mu} t /2 \right) dt \nonumber \\
& = d_1 + 1  +   \frac{n c_{1\mu}}{c_{2\mu} \tau^2 \kappa} \exp\left(-c_{2\mu} d_1^\kappa /2 \right) \nonumber \\
& = \left(\tfrac{2}{c_{2\mu}} \log n\right)^{1/\kappa} + 1  + \frac{c_{1\mu}}{c_{2\mu} \tau^2 \kappa} = O\left((\log n)^{1/\kappa}\right).
\end{align}
Therefore, from \eqref{eq:case2a}, \eqref{eq:case2b}, \eqref{eq:case2c}, and the bounds in Theorem 1, we obtain that
\begin{align*}
& \EE_{\sbb^*} \EE_{\Scal} \EE_{0|\Scal} \{\overline w(\sbb^*) - w_0(\sbb^*)\}^2  \\
& \leq O(n^{-1}) + 15 \frac{\tau^2}{n} \gamma \left( \frac{\tau^2}{n} \right) + \{1+o(1)\} \frac{2n}{k} \left\{  \frac{A b(m, d, r) \rho^2 \gamma(\tfrac{\tau^2_0}{n})}{\sqrt{m}} \right\}^r \\
& \leq O(n^{-1}) + O\left((\log n)^{1/\kappa}/n\right) + O(1)\cdot
\frac{n}{k} \left\{ \frac{(\log n)^{1/\kappa} \cdot \log n}{\sqrt{n/k}} \right\}^r \\
& \leq O\left((\log n)^{1/\kappa}/n\right) + O(1)\cdot  \frac{k^{\frac{r}{2}-1} (\log n)^{\frac{r(1+\kappa)}{\kappa}}}{n^{\frac{r}{2}-1}} \\
& = O\left((\log n)^{1/\kappa}/n\right),
\end{align*}
where the last equality follows from the condition on $k$.

\vspace{5mm}

\noindent (iii) In the upper bounds of Theorem 1, we choose $d=\lfloor n^{3/(2\nu-1)}\rfloor $ in every infimum for sufficiently large $n$ such that $\log d \geq  \log \left(n^{\frac{3}{2\nu-1}}-1\right) > r$. Then
\begin{align}
\label{eq:case3a}
\mu_{d+1} & \leq c_{\mu} n^{-6\nu/(2\nu-1)} \leq c_{\mu}n^{-3}, \nonumber \\
\tr\left(C_{\alphab}^d\right) & = \sum_{i=d+1}^{\infty} \mu_i \leq \sum_{i=d+1}^{\infty} c_{\mu} i^{-2\nu} \leq c_{\mu} \int_{d}^{\infty} \frac{1}{z^{2\nu}} dz \nonumber  \\
& = \frac{c_{\mu}}{2\nu-1} d^{-(2\nu-1)} \leq \frac{c_{\mu}}{2\nu-1} n^{-6\nu/(2\nu-1)} \leq \frac{c_{\mu}}{2\nu-1} n^{-3}, \nonumber \\
b(m,d,r) & \leq \max\left(\sqrt{\log d}, \frac{\log d}{m^{1/2-1/r}}\right) \leq \log d \leq \frac{3}{2\nu-1} \log n.
\end{align}
$\gamma(\tau^2/n)$ can be bounded as
\begin{align}
\label{eq:case3b}
\gamma(\tau^2/n) & = \sum_{i=1}^{\infty} \frac{1}{1+ \frac{\tau^2}{n\mu_i}} \leq \sum_{i=1}^{\infty} \frac{1}{1+ \frac{\tau^2i^{2\nu}}{c_{\mu}n}} \nonumber \\
& \leq n^{1/(2\nu)} + 1 + \frac{c_{\mu}n}{\tau^2} \sum_{i=\lfloor n^{1/(2\nu)}\rfloor+2}^{\infty} \frac{1}{i^{2\nu}}  \nonumber \\
& \leq n^{1/(2\nu)} + 1 +  \frac{c_{\mu}n}{\tau^2} \int_{n^{1/(2\nu)}}^{\infty} \frac{1}{z^{2\nu}} dz \nonumber  \\
& \leq n^{1/(2\nu)} + 1 +  \frac{c_{\mu}n}{\tau^2(2\nu-1)n^{(2\nu-1)/(2\nu)}} = \left(2+\frac{c_{\mu}}{\tau^2(2\nu-1)}\right) n^{1/(2\nu)} .
\end{align}
From \eqref{eq:case3a}, \eqref{eq:case3b}, and the bounds in Theorem 1, we obtain that
\begin{align*}
& \EE_{\sbb^*} \EE_{\Scal} \EE_{0|\Scal} \{\overline w(\sbb^*) - w_0(\sbb^*)\}^2 \\
&\leq O(n^{-1}) + 15 \frac{\tau^2}{n} \gamma \left( \frac{\tau^2}{n} \right) + \{1+o(1)\} \frac{2n}{k} \left\{  \frac{A b(m, d, r) \rho^2 \gamma(\tfrac{\tau^2_0}{n})}{\sqrt{m}} \right\}^r \\
& \leq O(n^{-1}) + \frac{15\tau^2\left(2+\frac{c_{\mu}}{\tau^2(2\nu-1)}\right) n^{1/(2\nu)}}{n} \\
&~~~ + \{1+o(1)\}
\frac{2n}{k} \left\{  \frac{3A\rho^2 \left(2+\frac{c_{\mu}}{\tau^2(2\nu-1)}\right) n^{1/(2\nu)} \log n}{(2\nu-1)\sqrt{n/k}} \right\}^r \\
& \leq O(n^{-1}) + O\left( n^{-\frac{2\nu-1}{2\nu}} \right) + O(1)\cdot \frac{k^{\frac{r}{2}-1}(\log n)^r}{n^{\frac{r}{2}-1-\frac{r}{2\nu}}} \\
& = O\left( n^{-\frac{2\nu-1}{2\nu}} \right),
\end{align*}
where the last equality follows from the condition on $k$.

\section{General posterior convergence rates for DISK}
\label{post-risk}

In this section, we provide some theoretical results for the posterior convergence rates of DISK posterior when $w_0$ can belong to a function class larger than the RKHS and the number of subsets $k$ grows relatively slower compared to Section 3.4. We only present results for the simplified model in equation (11) of the main manuscript. Recall that $\mathcal{S}^*$ is the set of $l$ reference locations in $\mathcal{D}$ and $\Scal^* \cap \Scal = \emptyset$. Let $\wb_0^* = \{w_0(\sbb_1^*), \ldots, w_0(\sbb_l^*)\}^T$ be the true residual spatial surface generating the data at the locations in $\Scal^*$ and $\wb^* = \{w(\sbb_1^*), \ldots, w(\sbb_l^*)\}^T$ be the realization of  GP $w(\cdot)$ at the locations in $\Scal^*$. Adapting our discussion in Section 3.2 of the main manuscript for the models in equation (5) of the main manuscript to the one for the model in equation (11) in the main manuscript, we have that $\yb_j $ given $\wb_j$ is Gaussian with density $N(\wb_j, k^{-1} \tau^2 \Ib)$ after stochastic approximation as in equation (8) of the main manuscript and the GP prior on $w(\cdot)$ implies that after integrating over $\wb_j$
\begin{align}
  \label{eq:4}
  \yb_j \mid \wb_j^* \sim N(\Ab_j \wb_j^*, \Sigmab_j), \quad \Ab_j = \Cb_{*j}^T \Cb_{*,*}^{-1}, \quad
  \Sigmab_j = k^{-1} \tau^2 \Ib + \Cb_{j,j} - \Cb_{*j}^T \Cb_{*,*}^{-1} \Cb_{*j},
\end{align}
where $\Cb_{*,*}$, $\Cb_{j,j}$, and $\Cb_{*j}$ are defined in equation (8) of the main manuscript. Let $\Ab$  and $\Sigmab$ represent the full data versions of $\Ab_j$ and $\Sigmab_j$ in \eqref{eq:4}. For any $\bb\in \RR^l$, we define two norms
\begin{align}
  \label{eq:5}
   \| \bb \|_{\Scal_j} = \left(\frac{1}{m} \bb^T \Ab_j^T \Sigmab_j^{-1}\Ab_j \bb \right)^{1/2}, \quad
   \| \bb \|_{\Scal} = \left(\frac{1}{n} \bb^T \Ab^T \Sigmab^{-1}\Ab \bb \right)^{1/2}.
\end{align}

Based on the definitions and notation introduced previously, we make the following five assumptions for deriving the general convergence rates of the DISK posterior:
\begin{enumerate}[label={C.\arabic*},ref=C.\arabic*]
\item \label{cd} (Compact domain) The spatial domain $\Dcal$ is a compact space in $\| \cdot \|_2$ metric.
\item \label{ne} (Norm equivalence) The partitions $\Scal_1, \ldots, \Scal_k$ of $\Scal$ are such that there exist universal positive constants $H_l < 1 < H_u$ independent of $j$ such that $H_l \, \| \cdot \|_{\Scal} \leq \| \cdot \|_{\Scal_j} \leq H_u \, \| \cdot \|_{\Scal}$  for $j = 1, \ldots, k$.
\item \label{me} (Metric entropy) Suppose that $\epsilon_m$ is a positive sequence that satisfies (i) $\sqrt{m} \epsilon_{m} \geq 1$ for all $m\geq 1$; (ii) $\epsilon_m\to 0$ as $m\to\infty$; (iii) with a  slight abuse of notation, for every $r > 1$, there is a set $\Fcal_r$ such that for all $m\geq 1$, $D(\epsilon_{m}, \Fcal_r, \|\cdot\|_{\Scal}) \leq e^{m \epsilon_{m}^2 H_l^2 r^2}$ and $\Pi(\Fcal_r) \geq 1 - e^{- 2m \epsilon^2_{m} r^2}$,
  where $D(\epsilon,\Fcal_r,\|\cdot\|_{\Scal})$ is the minimum number of $\|\cdot\|_{\Scal}$-balls of radius $\epsilon$ that cover $\Fcal_r$.
\item \label{pt} (Prior thickness) For the $\epsilon_m$ sequence in Assumption \ref{me} and for all $m\geq 1$, the prior assigns positive mass to any small neighborhood around $\wb_0^*$, $\Pi(w : \| \wb^* - \wb^*_0 \|_{\Scal} \leq \epsilon_m) \geq  e^{- m H_u^{2} \epsilon_{m}^2}$.
\item \label{eb} The metrics $ \| \cdot \|^2_2$ and  $\| \cdot \|^2_{\Scal}$ are equivalent in that $C_l \| \cdot \|^2_{\Scal} \leq \| \cdot \|^2_2  \leq C_u \| \cdot \|^2_{\Scal}$ for some positive universal constants $C_l$ and $C_u$.
\end{enumerate}

Assumption \ref{cd} is common to all models based on GP priors. Assumption \ref{ne} specifies a technical condition on the partitioning scheme so that the realizations of the GP observed in the $j$th subset are similar to those in the full data, where such similarity is described in terms of the norms $\|\cdot \|_{\Scal_j}$ and $\| \cdot \|_{\Scal}$. Assumption \ref{me} regulates the complexity of the sequence of sets $\Fcal_r$ in terms of $\| \cdot \|_{\Scal}$-metric entropy and specifies a condition on the probability assigned by the GP prior to $\Fcal_r$, ensuring that the prior probability of $\Fcal_r$ under the Gaussian measure induced by the GP prior increases with increasing $\| \cdot \|_{\Scal}$-metric entropy of $\Fcal_r$. The subscript $r$ here should not be confused with the number of knots in MPP or other low-rank GP priors. Assumption \ref{pt} says that the GP prior assigns
positive probability to arbitrarily small $\| \cdot \|_{\Scal}$-neighborhood around the true parameter $\wb_0^*$. Assumption \ref{eb} is a technical condition that is used in upper bounding the Bayes $L_2$-risk of Wasserstein barycenter in the estimation of $\wb_0^*$ if we have Bayes $L_2$-risk
upper bounds for the subset posterior distributions.

Similar to Section 3.4, the Bayes $L_2$-risk in the estimation of $\wb_0^*$ using the full data posterior is given by
\begin{align}
  \label{eq:l2-risk-full}
  \EE_{0 \mid \Scal, \Scal^*} \left\{ \EE \left( \| \wb^* - \wb^*_0 \|_{2}^2 \mid \yb \right) \right\}
  = \EE_{0 \mid \Scal, \Scal^*} \left\{ \int \| \wb^* - \wb^*_0 \|_{2}^2 \, d \Pi_n (w \mid \yb) \right\} ,
\end{align}
where $\EE_{0 \mid \Scal, \Scal^*}$ is the expectation under the true space varying function $w_0$ with respect to density of $\yb$ conditional on $\Scal, \Scal^*$ in equation (11) of the main paper. The decay rate of the risk in  \eqref{eq:l2-risk-full} is known under assumptions that are similar to \ref{cd}, \ref{me}, \ref{pt} and are obtained by replacing $m$ by $n$ \citep{VarZan11}.

The theorem below describes the Bayes $L_2$-risk of each subset posterior distribution and the combined DISK posterior distribution $\overline \Pi (\cdot \mid \yb_1, \ldots, \yb_k)$. The proof is given in Section \ref{sec:proof-main-theorem} after some technical lemmas.
\begin{theorem}\label{subsetbound}
If Assumptions \ref{cd}--\ref{eb} hold for the $j$th subset posterior $\Pi_m(\cdot \mid \yb_j)$  with $j=1, \ldots, k$, then there exists a positive constant $c(H_l)$ that only depends on $H_l$, such that
\begin{align*}
& \EE_{0 \mid \Scal_j, \Scal^*} \left\{ \EE ( \| \wb^* - \wb^*_0 \|_{2}^2 \mid \yb_j ) \right\}  \leq C_u c(H_l) \epsilon_m^2, \quad j = 1, \ldots, k, \\
& \EE_{0 \mid \Scal, \Scal^*} \left\{ \int \| \wb^* - \wb^*_0 \|_{2}^2 \, d \overline \Pi (w \mid  \yb_1, \ldots, \yb_k) \right\}
\leq C_u^2 c(H_l) \epsilon_m^2
\end{align*}
as $m \rightarrow \infty$, where $\EE_{0 \mid \Scal_j, \Scal^*}$ is the expectation under the true space varying function $w_0$ with respect to the subset $\yb_j$ of size $m$ conditional on $\Scal_j, \Scal^*$, and $\EE_{0 \mid \Scal, \Scal^*}$ is the expectation under $w_0$ with respect to the full dataset of size $n$ conditional on $\Scal, \Scal^*$.
\end{theorem}

Theorem \ref{subsetbound} holds for any $\epsilon_m$ sequence that satisfies Assumptions \ref{me} and  \ref{pt}. Explicit expressions for $\epsilon_m$ are available if $w_0(\cdot)$ and $\Fcal_r$ are restricted to class of functions with known regularity and $\Pi$ is assumed to be a GP prior with the Mat\'ern or squared exponential covariance kernels. For any $a, b > 0$,  let $C^a[0, 1]^d$ and $H^b[0, 1]^d$ be the H\"older and Sobolev spaces of functions on $[0, 1]^d$ with regularity index $a$ and $b$, respectively. Define $\Dcal = [0, 1]^d$ and $C_{\alphab}$ to be the Mat\'ern kernel with
$C_{\alphab}(\sbb, \sbb') = \frac{\sigma^2_0}{2^{\nu_{0}-1} \Gamma(\nu_{0})}\left(\phi_0\| \sbb - \sbb' \|_2\right)^{\nu_{0}}\mathcal{K}_{\nu_{0}}\left(\phi_0\| \sbb - \sbb' \|_2\right)$ for $\sbb, \sbb' \in \Dcal$, where $\mathcal{K}_{\nu_{0}}$ is a modified Bessel function of the second kind with order, $\nu_{0}$, that controls the process smoothness, and $\Gamma$ is the Gamma function.
If $w_0 \in C^{b^*}[0, 1]^d \cap H^{b^*}[0, 1]^d$ and $\Fcal_r \subset C^{b^*}[0, 1]^d \cap H^{b^*}[0, 1]^d$ for  $b^* > 0$ and $r > 1$, then $\epsilon_m = m^{- \min(\nu_0, b^*) / (2 \nu_0 + d)}$, provided $\min(\nu_0, b^*) > d/2$. Similarly, if $\Dcal = [0, 1]^d$, $C_{\alphab}$ is the squared exponential kernel with $C_{\alphab}(\sbb, \sbb') = \sigma^2_0 e^{-\phi_0 \| \sbb - \sbb' \|^2_2}$, and $w_0$ is an analytic function on $\Dcal$, then $\epsilon_m = (\log m)^{1/2} / \sqrt{m}$; see Theorems 5 and 10 in \cite{VarZan11} for detailed proofs.

If $k\approx \log^a n$ for some $a>0$, then $m \approx n \log^{-a} n$. With this choice of $(m,k)$,  discussion in the previous paragraph implies that $\epsilon_m =  n^{-c^*} \log^{a c^*} n$ for the Mat\'ern, where $c^* = \frac{\min(\nu_0, b^*)}{2 a^* + d}$, and $\epsilon_m = (\log n)^{a/2 + 1/2} / \sqrt{n}$  for the squared exponential covariance kernels. Both these rates are minimax optimal up to $\log$ factors \citep{VarZan11}.

In applications, we are also interested in estimating functions of $\wb_0^*$. An attractive property of the DISK  posterior is that its theoretical guarantees extend to a large class of functions of $\wb^*$. Let $f$ be any function that maps $\wb^*$ to $f(\wb^*)$ and that $f$ is bounded almost linearly by the $\| \cdot \|_{2}$ metric. Then, we have the following corollary from a direct application of Lemma 8.5 in \cite{BicFre81}.
\begin{corollary}\label{funcrate}
Suppose that Assumptions \ref{cd}--\ref{eb} hold for all subset posteriors $\Pi_m(\cdot \mid \yb_j)$  with $j=1, \ldots, k$. Let $f$ be a continuous function that maps $\RR^l$ to $\RR^{l'}$ and satisfies $\| f(\wb^*) \|^2_2 \leq C_f (1 + \|\wb^* -  \wb^*_0 \|^2_2)$ for any $\wb^* \in \RR^l$, where $C_f>0$ is a fixed constant. Let $\overline{ f\sharp\Pi} (\cdot \mid \yb_1, \ldots, \yb_k)$ represent the DISK posterior of $f(\wb^*)$, then as $m\to \infty$,
$$ \int \| \fb - \fb(\wb^*_0) \|_2^2 \, d \overline{ f\sharp\Pi} ( \fb \mid \yb_1, \ldots, \yb_k) = O_{p} \left( \epsilon_m^2 \right),$$
where $O_p$ is in the probability measure under the true space varying function $w_0$ with respect to the full dataset of size $n$ conditional on $\Scal, \Scal^*$.
\end{corollary}

\subsection{Technical Lemmas}
\label{sec:technical-lemmas}

For notational convenience, we define two additional ``Hilbert'' norms $\|\cdot\|_{\HH_j}$ and $\|\cdot\|_{\HH}$, which are rescaled versions of $\|\cdot\|_{\Scal_j}$ and $\|\cdot\|_{\Scal}$. The inner products can also be defined accordingly.
\begin{align*}
  &\langle \hb_{j1}, \hb_{j2} \rangle_{\HH_j} = \hb_{j1}^T \Ab_j^T \Sigmab_j^{-1} \Ab_j \hb_{j2}, \quad
  \| \hb_{j1} \|_{\HH_j}^2 = \langle \hb_{j1}, \hb_{j1} \rangle_{\HH_j}, \quad \nonumber\\
  &\hb_{j1}, \hb_{j2} \in \RR^l, \quad \| \cdot \|_{\HH_j} = \sqrt{m} \| \cdot \|_{\Scal_j},\nonumber\\
  &\langle \hb_1, \hb_2 \rangle_{\HH} = \hb_1^T \Ab^T \Sigmab^{-1} \Ab \hb_2, \quad
  \| \hb_1 \|_{\HH}^2 = \langle \hb_1, \hb_1 \rangle_{\HH} , \quad \nonumber\\
  &\hb_1, \hb_2 \in \RR^l, \quad \| \cdot \|_{\HH} = \sqrt{n} \| \cdot \|_{\Scal}.
\end{align*}

We first prove a series of technical lemmas under our model setup, similar to the lemmas in \citet{VarZan11}.
\begin{lemma}\label{lemma12}
  Suppose that Assumption A.2 holds. Let $\yb_j \sim N_m(\Ab_j \thetab, \Sigmab_j)$, where $\thetab \in \RR^{l}$ is such that $ \| \thetab - \thetab_1 \|_{\HH_j}  \leq \| \thetab_0 - \thetab_1 \|_{\HH_j} / 2$ for  any $\thetab_0, \thetab_1 \in \RR^{l}$. Then, there exists a test $\phi(\yb_j)$ such that $\max\left(\EE_{\thetab_0 } \{\phi(\yb_j)\}, \EE_{\thetab } \{1 - \phi(\yb_j)\}\right) \leq e^{- \| \thetab_0-\thetab_1 \|^2_{\HH_j} / 32} \leq e^{- H_l \| \thetab_1 \|^2_{\HH} / 32} $, where $\EE_{\thetab }$ is the expectation with respect to the measure $N_m(\yb_j \mid \Ab_j \thetab, \Sigmab_j)$.
\end{lemma}

\begin{proof}
  Choose $\thetab_0=0$ for simplicity, and define the test function $\phi(\yb_j) = \one({\thetab_1^T \Ab^T_j \Sigmab_j^{-1} \yb_j > D \| \thetab_1 \|_{\HH_j} }) $. If $\thetab_0=0$, then $ \| \thetab - \thetab_1 \|_{\HH_j}  \leq \| \thetab_1 \|_{\HH_j} / 2$ and the triangular inequality gives $ \| \thetab_1 \|_{\HH_j} / 2 \leq \| \thetab \|_{\HH_j}$. The type I error probability of $\phi(\yb_j)$ is
  \begin{align}
    \label{eq:61}
    \EE_{\thetab_0} \{\phi(\yb_j)\}  &= \PP_{\thetab_0} \left( \thetab_1^T \Ab^T_j \Sigmab_j^{-1} \yb_j > D \| \thetab_1 \|_{\HH_j} \right).
  \end{align}
  Since $\yb_j \sim  N_m(\zero, \Sigmab_j)$,  $\thetab_1^T \Ab^T_j \Sigmab_j^{-1} \yb_j \sim N_m(0, \thetab_1^T \Ab_j^T \Sigmab_j^{-1} \Ab_j \thetab_1) = N(0, \| \thetab_1 \|_{\HH_j}^2)$ and type I error probability in \eqref{eq:61} is
  \begin{align*}
    \EE_{\thetab_0} \{\phi(\yb_j)\} = 1 - \PP_{\thetab_0} \left( \frac{\thetab_1^T \Ab^T_j \Sigmab_j^{-1} \yb_j - 0} {\| \thetab_1 \|_{\HH_j}}  \leq D  \right) = 1-\Phi(D).
  \end{align*}
  For $\thetab \neq \thetab_0$,  $\thetab_1^T \Ab_j^T \Sigmab_j^{-1} \yb_j \sim N(\thetab_1^T \Ab^T_j \Sigmab^{-1}_j \Ab_j \thetab, \thetab_1^T \Ab_j^T \Sigmab_j^{-1} \Ab_j \thetab_1) = N(\langle \thetab_1, \thetab  \rangle_{\HH_j}, \| \thetab_1 \|_{\HH_j}^2)$ and
  \begin{align}
    \EE_{\thetab }\{1 - \phi(\yb_j)\} &= \PP_{\thetab} \left(\thetab_1^T \Ab_j^T \Sigmab_j^{-1} \yb_j < D \| \thetab_1\|_{\HH_j} \right) \nonumber\\
    &= \PP_{\thetab} \left( \frac{\thetab_1^T \Ab^T \Sigmab^{-1} \yb_j -  \langle \thetab_1, \thetab  \rangle_{\HH_j}} { \| \thetab_1 \|_{\HH_j}} < D -   \frac{\langle \thetab_1, \thetab  \rangle_{\HH_j}}{\| \thetab_1 \|_{\HH_j}} \right) \nonumber \\
                            &=     \Phi\left(D - \frac{\langle \thetab_1, \thetab  \rangle_{\HH_j}}{\| \thetab_1 \|_{\HH_j}}  \right). \label{type2}
  \end{align}
  To find an upper bound for $\EE_{\thetab }\{1 - \phi(\yb_j)\}$, notice that
  \begin{align}
    \label{eq:type2-bd}
    &\langle \thetab_1, \thetab  \rangle_{\HH_j} = \frac{\| \thetab_1 \|_{\HH_j}^2 + \| \thetab \|_{\HH_j}^2 - \| \thetab - \thetab_1 \|_{\HH_j}^2}{2}  \nonumber\\
    &\geq
    \frac{ \| \thetab_1 \|_{\HH_j}^2 + \| \thetab_1 \|_{\HH_j}^2 / 4 - \| \thetab_1 \|_{\HH_j}^2 / 4}{2}  = \frac{\| \thetab_1 \|_{\HH_j}^2}{2} .
  \end{align}
  Substituting \eqref{eq:type2-bd} in \eqref{type2} implies that $\EE_{\thetab} \{1 - \phi(\yb_j)\} \leq \Phi\left(D - \frac{\|  \thetab_1 \|_{\HH_j}}{2}  \right)$.
  Under 0-1 loss, an upper on the risk of the decision rule based on $\phi$ is
  \begin{align*}
    1-\Phi(D) + \Phi\left(D - \frac{\|  \thetab_1 \|_{\HH_j}}{2}  \right) = \Phi(-D) + \Phi\left(D - \frac{\|  \thetab_1 \|_{\HH_j}}{2}  \right),
  \end{align*}
  which attains its minimum at $D = \|  \thetab_1 \|_{\HH_j} / 4 $. Substituting this in the upper bounds for Type I and II error probabilities and using $\Phi(-x) \leq \exp(-x^2/2)$, we get
  \begin{align}
    \label{eq:7}
    \Phi(- \| \thetab_1 \|_{\HH_j} / 4) \leq e^{- (\| \thetab_1 \|_{\HH_j} / 4)^2 /2} = e^{- \| \thetab_1 \|^2_{\HH_j} / 32} \leq e^{- H_l \| \thetab_1 \|^2_{\HH} / 32} .
  \end{align}
\end{proof}

Define, $D(\epsilon, \Thetab, \| \cdot \|_{\HH_j})$ is the maximal number of points that can be placed inside the set $\Thetab \subset \RR^l$ such that $ \| \thetab_0-\thetab_1 \|_{\HH_j} >\epsilon$ for any two different points $\thetab_0$ and $\thetab_1$ in $\Thetab$.
\begin{lemma}\label{lemma13}
  Suppose that Assumptions A.1 and A.2 hold. Let $\yb_j \sim  N_m(\Ab_j \thetab, \Sigmab_j)$ for any $\thetab \in \RR^l$. Then, there exists a test $\phi(\yb_j)$ such that for every $r > 1$ and every $i \geq 1$,
  \begin{align*}
    \EE_{\thetab_0} \{\phi(\yb_j)\}  &\leq 33 \, D( r/2, \Thetab, \| \cdot \|_{\HH_j}) e^{-r^2 / 32} ,\\
    \underset{\{\thetab \in \Thetab : \|\thetab - \thetab_0 \|_{\HH_j} \geq ir \}}{\sup} \EE_{\thetab} \{1 - \phi(\yb_j)\}  &\leq  e^{- i^2 r^2 /32};
  \end{align*}
  and
  \begin{align*}
    \EE_{\thetab_0} \{\phi(\yb_j)\}  &\leq 33 \, D( \sqrt{k} H_l^{-1} r/2, \Thetab, \| \cdot \|_{\HH}) e^{-r^2 / 32},\\
    \underset{\{\thetab \in \Thetab : \|\thetab - \thetab_0 \|_{\HH} \geq ir \sqrt{k}H_l^{-1} \}}{\sup} \EE_{\thetab} \{1 - \phi(\yb_j)\}  &\leq  e^{- i^2 r^2 / 32}.
  \end{align*}
\end{lemma}
\begin{proof}
  Partition $\Thetab$ into disjoint shells defined as $E_{i,r} = \{\thetab : ir \leq \| \thetab - \thetab_0 \|_{\HH_j} \leq (i+1)r \}$ $(i=0, 1, \ldots)$. For any $i \geq 1$,
  if $\Thetab_i = D(i r / 2, E_{i,r}, \| \cdot \|_{\HH_j})$, then $\| \thetab_{ia} - \thetab_{ib} \|_{\HH_j} > i r / 2$ for any $\thetab_{ia}, \thetab_{ib} \in \Thetab_{i}$. Furthermore, for any $\thetab \in E_{i,r}$, there is some $\thetab_{i1} \in \Thetab_i$ such that $\| \thetab - \thetab_{i1} \|_{\HH_j} \leq i r / 2 \leq \| \thetab_{i1} - \thetab_0 \|_{\HH_j} / 2$; therefore, Lemma \ref{lemma12} implies that there exists a test $\phi^*_i(\yb_j)$ such that  $\max\{ \EE_{\thetab_0 } \{\phi^*_i(\yb_j)\}, \EE_{\thetab} \{1 - \phi^*_i(\yb_j)\}\} \leq e^{- \| \theta_0-\theta_{i1} \|^2_{\HH_j} / 32}$. Define the test $\phi_i(\yb_j) = \underset{\thetab_{i1} \in \Thetab_i} {\sup} \phi^*_i(\yb_j)$, so the union bound implies that
  \begin{align*}
    \EE_{\thetab_0}\{\phi_i(\yb_j)\}   &\leq \sum_{\thetab_{i1} \in \Thetab_i}e^{- \| \theta_0-\theta_{i1} \|^2_{\HH_j} / 32}
    \leq D( ir/2, E_{i,r}, \| \cdot \|_{\HH_j})  e^{- i^2 r^2 / 32}, \\
    \underset{\thetab \in E_{i,r}}{\sup} \EE_{\thetab} \{1 - \phi_i(\yb_j)\}  &\leq  e^{\sup_{\thetab_{i1} \in \Thetab_i} (- \| \theta_0-\theta_{i1} \|^2_{\HH_j} / 32)} \leq e^{- i^2 r^2 / 32}.
  \end{align*}
  Define $\phi(\yb_j) = \sup_{i \geq 1} \phi_i(\yb_j)$. Again, union bound implies that
  \begin{align*}
    &\EE_{\thetab_0} \{\phi(\yb_j)\}  \leq \sum_{i \geq 1} D( ir/2, E_{i,r}, \| \cdot \|_{\HH_j})  e^{- i^2 r^2 / 32} \\
    &\qquad \qquad \quad \leq D( r/2, \Thetab, \| \cdot \|_{\HH_j}) e^{-r^2 / 32} ( 1 - e^{-1 / 32})^{-1} \\
    &\qquad \qquad \quad \leq 33 D( r/2, \Thetab, \| \cdot \|_{\HH_j}) e^{-r^2 / 32} ,\\
    &\sup_{\{\thetab \in \Thetab : \|\thetab - \thetab_0 \|_{\HH_j} \geq ir \}} \EE_{\thetab} \{1 - \phi(\yb_j)\}  \leq  e^{- i^2 r^2 / 32},
  \end{align*}
  for any $r>1$ and every $i \geq 1$. Since $H_l \| \cdot \|_{\HH} \leq \sqrt{k} \| \cdot \|_{\HH_j} \leq H_u \| \cdot \|_{\HH} $ according to Assumption A.2, we have
  \begin{align*}
   & \EE_{\thetab_0} \{\phi(\yb_j)\}  \leq 33\, D( \sqrt{k} H_l^{-1} r/2, \Thetab, \| \cdot \|_{\HH}) e^{-r^2 / 32},\\
   & \underset{\{\thetab \in \Thetab : \|\thetab - \thetab_0 \|_{\HH} \geq irH_l^{-1}  \sqrt{k} \}}{\sup} \EE_{\thetab} \{1 - \phi(\yb_j)\}  \leq  e^{- i^2 r^2 / 32}.
  \end{align*}
\end{proof}

\begin{lemma}\label{lemma:set}
  Let $p_{m, \thetab}(\yb_j)$ be the pdf of $\yb_j \sim  N_m(\Ab_j \thetab, \Sigmab_j)$ for any $\thetab \in \RR^l$. Then, for any probability distribution $\Pi$ on $\RR^l$ and $x > 0$,
  \begin{align}
    \PP_{\thetab_0} \left\{\yb_j :  \int \left( \frac{p_{m,\thetab}}{p_{m,\thetab_0}} \right)(\yb_j) \, d\Pi(\thetab) \leq e^{- \sigma_{0j}^2 / 2 - \| \mub_{0} \|_{\HH_j} x} \right\} \leq e^{-x^2/2},\label{dino}
  \end{align}
  where $\mub_{0} = \int (\thetab - \thetab_0) d\Pi(\thetab)$ and $\sigma_{0j}^2 = \int \| \thetab -\thetab_0 \|^2_{\HH_j} d\Pi(\thetab).$ Consequently, for any probability distribution $\Pi$ on $\RR^l$ and any $r>1$
  \begin{align}
    \PP_{\thetab_0} \left\{ \yb_j : \int  \left( \frac{p_{m,\thetab}} {p_{m,\thetab_0}} \right) (\yb_j) d \Pi (\thetab) \geq e^{- r^2} \Pi \left( \thetab: \| \thetab-\thetab_0 \|_{\HH_j}^2 < r^2 \right)  \right\} \geq 1-e^{-r^2/8}.\label{dino2}
  \end{align}
\end{lemma}
\begin{proof}
  The pdf of $\yb_j$ implies that
  \begin{align}
    &\log \left( \frac{p_{m,\thetab}}{p_{m,\thetab_0}} \right) (\yb_j) \nonumber \\
    &= \frac{1}{2} \left\{ (\yb_j - \Ab_j \thetab_0)^T \Sigmab_j^{-1} (\yb_j - \Ab_j \thetab_0) - (\yb_j - \Ab_j \thetab)^T \Sigmab_j^{-1} (\yb_j - \Ab_j \thetab_j) \right\}  \nonumber \\
    &= \frac{1}{2} \left\{ \| \thetab_0 \|^2_{\HH_j} - \| \thetab \|^2_{\HH_j}  - 2 \thetab_0^T \Ab_j^T \Sigmab_j^{-1} \yb_j + 2 \thetab^T \Ab_j^T \Sigmab_j^{-1} \yb_j \right\}  \nonumber\\
    &= \frac{1}{2} \Big\{ 2 \| \thetab_0 \|^2_{\HH_j} - 2 \langle \thetab_0, \thetab \rangle_{\HH_j} + 2 \yb_j^T \Sigmab_j^{-1} \Ab_j (\thetab  - \thetab_0) \nonumber\\
    &~~ - \left(\| \thetab \|^2_{\HH_j} + \| \thetab_0 \|^2_{\HH_j} - 2 \langle \thetab_0, \thetab \rangle_{\HH_j} \right) \Big\} \nonumber\\
    &= \langle \thetab_0, \thetab_0 - \thetab \rangle_{\HH_j} +  \yb_j^T \Sigmab_j^{-1} \Ab_j (\thetab  - \thetab_0)  - \frac{1}{2}\| \thetab - \thetab_0 \|^2_{\HH_j}. \label{pow-lik}
  \end{align}
Integrating with respect to $\Pi$ on both sides and using the definitions of $\mub_{0}$ and $\sigma_{0j}^2$,
  \begin{align*}
    \int \log \left( \frac{p_{m,\thetab}}{p_{m,\thetab_0}} \right) (\yb_j) \, d \Pi(\thetab) &= -   \langle \thetab_0,  \mub_{0} \rangle_{\HH_j}   +  \yb_j^T \Sigmab_j^{-1} \Ab_j  \mub_{0} -  \sigma_{0j}^2 / 2 \nonumber \\
    &=  \mub_{0}^T \Ab_j^T \Sigmab_j^{-1}  \left(\yb_j - \Ab_j \thetab_0 \right)  - \sigma_{0j}^2 / 2 .
  \end{align*}
  If $\yb_j \sim N_m (\Ab_j \thetab_0, \Sigmab_j)$, then $ \mub_{0}^T \Ab_j^T \Sigmab_j^{-1}  \left(\yb_j - \Ab_j \thetab_0 \right) \sim N_m (\zero, \| \mub_{0} \|_{\HH_j}^2)$ and
  \begin{align*}
    \mub_{0}^T \Ab_j^T \Sigmab_j^{-1}  \left(\yb_j - \Ab_j \thetab_0 \right)  - \sigma_{0j}^2 / 2\sim N (- \sigma_{0j}^2 / 2,  \| \mub_{0} \|_{\HH_j}^2).
  \end{align*}
  An application of Jensen's inequality implies that
  \begin{align*}
    & \PP_{\thetab_0} \left\{ \int\left(\frac{p_{m,\thetab}}{p_{m,\thetab_0}}\right) (\yb_j) \, d\Pi(\thetab)\leq e^{-\sigma_{0j}^2/2- \| \mub_{0} \|_{\HH_j} x} \right\} \\
    &\leq \PP_{\thetab_0}\left( \int \log \left( \frac{p_{m,\thetab}}{p_{m,\thetab_0}} \right) (\yb_j) \, d\Pi(\thetab) \leq - \sigma_{0j}^2/2-  \| \mub_{0} \|_{\HH_j} x \right)\\
    &=\PP_{\thetab_0}\left( \frac{\int \log \left( \frac{p_{m,\thetab}}{p_{m,\thetab_0}} \right) (\yb_j) \, d\Pi(\thetab) +  \sigma_{0j}^2 / 2 } { \| \mub_{0} \|_{\HH_j} }  \leq \frac{- \sigma_{0j}^2/2-  \| \mub_{0} \|_{\HH_j} x +  \sigma_{0j}^2 / 2}{ \| \mub_{0} \|_{\HH_j}} \right) \\
    &= \Phi(-x) \leq e^{-x^2/2}.
  \end{align*}

  Suppose the integration in  \eqref{dino} is restricted to the set $\widetilde \Thetab = \left\{\thetab: \| \thetab - \thetab_0 \|_{\HH_j} \leq r \right\}$. The prior $\Pi$ in \eqref{dino} can be renormalized to the truncated prior $\widetilde \Pi = \Pi / \Pi(\widetilde \Thetab)$. Using  \eqref{dino} for $\widetilde \Pi$ implies that
  \begin{align}
    \label{eq:b8}
    \PP_{\thetab_0} \left\{ \{\Pi(\widetilde \Thetab) \}^{-1}\int_{\widetilde \Thetab}\left(\frac{p_{m,\thetab}}{p_{m,\thetab_0}}\right) (\yb_j) d \Pi(\thetab) \leq e^{-\sigma_{0j}^2 / 2 - \| \mub_{0} \|_{\HH_j} x} \right\} \leq e^{-x^2/2}.
  \end{align}
  On the other hand, 
  \begin{align}
    \label{eq:9}
    \| \mub_{0} \|^2_{\HH_j} &= \bigg\| \int_{\widetilde \Thetab} (\thetab - \thetab_0) d \widetilde \Pi(\thetab) \bigg\|_{\HH_j}^2 \overset{(i)}{\leq} \int_{\widetilde \Thetab} \| \thetab - \thetab_0 \|_{\HH_j}^2  d \widetilde \Pi(\thetab) \leq r^2 \; \widetilde \Pi(\widetilde \Thetab) = r^2, \nonumber \\
    \sigma_{0j}^2 &= \int_{\widetilde \Thetab} \| \thetab - \thetab_0 \|^2_{\HH_j} d \widetilde \Pi(\thetab) \leq r^2 \; \widetilde \Pi(\widetilde \Thetab) = r^2,
  \end{align}
  where $(i)$ follows from Jensen's inequality. Substituting \eqref{eq:9} in \eqref{eq:b8} and setting $x = r/2$,
  \begin{align}
    \label{eq:10}
    e^{-\sigma_{0j}^2 / 2 - \| \mub_{0} \|_{\HH_j}  r / 2}  \geq e^{- r^2 / 2 -  r^2 /2} = e^{- r^2}
  \end{align}
  and
  \begin{align}
    \label{eq:11}
    &\PP_{\thetab_0} \left\{ \int\left(\frac{p_{m,\thetab}}{p_{m,\thetab_0}}\right) (\yb_j) d \Pi(\thetab) \leq e^{- r^2 } \Pi(\thetab :  \| \thetab - \thetab_0 \|_{\HH_j} \leq r) \right\} \nonumber\\
    \leq &\PP_{\thetab_0} \left\{ \int\left(\frac{p_{m,\thetab}}{p_{m,\thetab_0}}\right) (\yb_j) d \Pi(\thetab) \leq e^{-\sigma_{0j}^2 / 2 - \| \mub_{0} \|_{\HH_j}  r/2} \Pi(\thetab :  \| \thetab - \thetab_0 \|_{\HH_j} \leq r) \right\} \nonumber\\
    \leq & e^{-r^2/8}.
  \end{align}
\end{proof}

\subsection{Proof of Theorem \ref{subsetbound}} \label{sec:proof-main-theorem}

We first prove the first relation in Theorem \ref{subsetbound} for the $j$th subset posterior distribution. For clarity of notation, in the rest of the proof, we use $\PP_{m,w_0}$ to denote the expectation $\EE_{0\mid \Scal_j, \Scal^*}$ in Theorem \ref{subsetbound}, and $\PP_{m,w}$ to denote the expectation under a possibly different space varying function $w$.
 Let $\underline r = 1+(8H_l)^{-1}+ H_l^{-2} + 1/(1-H_l^2/8)$, and let $r$ be any number such that $r>\underline r > 1$. Define $\Wcal_{r} = \{w : \| \wb^* - \wb_0^* \|_{\Scal} > 8 r\epsilon_m  \}$. For any event $\Acal$ and any test $\phi(\yb_j)$,
  \begin{align}
    \label{eq:15}
    &\PP_{m,w_0} \left\{ \Pi_m(\Wcal_r \mid \yb_j) \right\} =     \PP_{m,w_0} \left\{ \Pi_m(\Wcal_r \mid \yb_j) (\one_{\Acal^c} + \one_{\Acal})  \right\} \nonumber\\
    =& \PP_{m,w_0} \left\{ \Pi_m(\Wcal_r \mid \yb_j) \one_{\Acal^c}   \right\} + \PP_{m,w_0} \left\{ \Pi_m(\Wcal_r \mid \yb_j) \one_{\Acal}   \right\} \nonumber \\
    \leq &\PP_{m,w_0}\left( \Acal^c \right)  + \PP_{m,w_0} \left\{ \Pi_m(\Wcal_r \mid \yb_j) \one_{\Acal}   \right\}
    \nonumber \\
    = &\PP_{m,w_0}\left( \Acal^c \right)  + \PP_{m,w_0} \left\{ \Pi_m(\Wcal_r \mid \yb_j) \one_{\Acal} \phi(\yb_j)  \right\} + \PP_{m,w_0} \left[ \Pi_m(\Wcal_r \mid \yb_j) \one_{\Acal} \{1 - \phi(\yb_j)\} \right] \nonumber \\
    \leq &\PP_{m,w_0}\left( \Acal^c \right)  + \PP_{m,w_0} \{\phi(\yb_j)\} + \PP_{m,w_0} \left[  \Pi_m(\Wcal_r \mid \yb_j) \one_{\Acal} \{1 - \phi(\yb_j)\}  \right]  \nonumber \\
   \leq &\PP_{m,w_0}\left( \Acal^c \right)  +  \PP_{m,w_0} \{\phi(\yb_j)\} + \PP_{m,w_0} \left\{ \Pi_m(\Fcal_r^c \mid \yb_j) \one_{\Acal}  \right\} \nonumber\\
    &+  \PP_{m,w_0} \left[  \Pi_m(\Wcal_r \cap \Fcal_r \mid \yb_j) \one_{\Acal} \{1 - \phi(\yb_j)\}  \right]\nonumber \\
    \equiv & A_1 + A_2 + A_3 + A_4.
  \end{align}
  We find upper bound $A_1, A_2, A_3$ and $A_4$, respectively.\\

  \noindent \underline{\emph{Bounding $A_1$}}\\
  Use \eqref{dino2} in Lemma \ref{lemma:set} to define
  \begin{align}
    \Acal = \left\{\yb_j : \int \left(\frac{p_{m, w }} {p_{m, w_0}} \right) (\yb_j) d \Pi (w) \geq e^{- m \epsilon^2_m r^2} \Pi \left( w: \| \wb^*-\wb_0^* \|_{\HH_j} < \sqrt{m} \epsilon_m r \right) \right\}    . \label{event}
  \end{align}
  Since $\sqrt{m}\epsilon_m>1$ by Assumption A.3, setting $r^2$ to be $m \epsilon^2_{m} r^2$ in \eqref{dino2} implies that
  \begin{align*}
     &A_1= \PP_{m, w_0}\left( \Acal^c \right) \\
     &= \PP_{m,w_0} \left\{\yb_j : \int \left(\frac{p_{m, w }} {p_{m, w_0}} \right) (\yb_j) d \Pi (w) < e^{-m \epsilon^2_m r^2} \Pi \left( w: \| \wb^* - \wb^*_0 \|^2_{\HH_j} < m \epsilon^2_{m} r^2 \right) \right\} \\
    &\leq e^{- m \epsilon^2_{m} r^2 / 8}.
  \end{align*}

\noindent \underline{\emph{Bounding $A_2$}}\\
  Let $\Thetab = \{\wb^* = \{w(\sbb_{1}^*), \ldots, w(\sbb_{l}^*)\}^T: w \in \Fcal_r\} \subset \RR^l$ and $\thetab_0 = \wb_0^* = \{w_0(\sbb_{1}^*), \ldots, w_0(\sbb_{l}^*)\}^T $. Set $r$ to be $ 8 H_l \sqrt{m} \epsilon_{m} r$ in Lemma~\ref{lemma13} and use the test $\phi(\yb_j)$ defined for every $8 H_l \sqrt{m} \epsilon_{m} r > 8H_l\underline r > 1 $  and every integer $i \geq 1$ such that
  \begin{align}
    \label{eq:tst13}
    &\PP_{m, w_0}\{\phi(\yb_j)\}  \leq 33 D(4 \sqrt{mk} \epsilon_{m} r, \Thetab, \| \cdot \|_{\HH}) e^{- 2m \epsilon^2_{m} H_l^2 r^2}\nonumber \\
    &\underset{\{w:~ \wb^* =w(\sbb^*)\in \Thetab, \| \wb^* - \wb^*_0 \|_{\HH} \geq 8 i  \sqrt{mk} \epsilon_{m} r \}}{\sup} \PP_{m, w} \{1 - \phi(\yb_j)\} \leq e^{- 2 i^2 m\epsilon^2_{m} H_l^2r^2}.
  \end{align}
  Noting that $A_2 = \PP_{m,w_0} \{\phi(\yb_j) \} = \PP_{m, w_0}\{\phi(\yb_j)\} \leq 33 D(4 \sqrt{mk} \epsilon_{m} r, \Thetab, \| \cdot \|_{\HH}) e^{- 2m \epsilon^2_{m} H_l^2 r^2}$ gives
  \begin{align*}
    A_2  & \overset{(i)}{\leq} 33 D(4 \epsilon_{m} r, \Thetab, \| \cdot \|_{\Scal}) e^{- 2 m \epsilon^2_{m} H_l^2 r^2}
         \overset{(ii)}{\leq} 33 D(\epsilon_{m}, \Thetab, \| \cdot \|_{\Scal}) e^{- 2 m \epsilon^2_{m} H_l^2 r^2} \\
       & \overset{(iii)}{\leq} 33 e^{m \epsilon^2_{m} H_l^2  r^2} e^{- 2 m \epsilon^2_{m} H_l^2 r^2} = 33 e^{- m \epsilon^2_{m} H_l^2  r^2},
  \end{align*}
  where $(i)$ follows from the definitions of $\|\cdot\|_{\HH_j}$ and $\|\cdot\|_{\Scal_j}$ norms, $(ii)$ follows from the property of covering number because $4r > 1$,  and $(iii) $ follows from Assumption A.3.\\

  \noindent \underline{\emph{Bounding $A_3$}}\\
  If the event $\Acal$ occurs, then
  \begin{align*}
    \int \left(\frac{p_{m, w }} {p_{m, w_0}} \right) (\yb_j) d \Pi (w) &\geq e^{-m \epsilon^2_m r^2} \Pi \left( w: \| \wb^*-\wb_0^* \|_{\HH_j} < \sqrt{m} \epsilon_m r \right) \\
    &\geq e^{- m \epsilon^2_{m} r^2} \Pi \left( w:  \frac{\| \wb^*-\wb_0^* \|_{\HH}}{\sqrt{mk}} <  H_u^{-1} \epsilon_mr   \right)\\
    &= e^{- m \epsilon^2_{m} r^2} \Pi \left( w: \| \wb^*-\wb_0^* \|_{\Scal} <  H_u^{-1} \epsilon_m r   \right)\\
    &\geq e^{- m \epsilon^2_{m} r^2} \Pi \left( w: \| \wb^*-\wb_0^* \|_{\Scal} < H_u^{-1} \epsilon_m   \right)\\
    &\geq e^{- m \epsilon^2_m r^2 - m \epsilon^2_m } = e^{- m \epsilon^2_m (r^2 + 1)} , 
  \end{align*}
  where the last inequality follows from Assumption A.4. Assuming that the event $\Acal$ occurs, the posterior probability of any event $\Bcal \subset \Thetab$ implied by the $j$th subset posterior of $\wb^*$ is bounded by
  \begin{align}
    \label{eq:8}
    \Pi_m (\Bcal \mid \yb_j )\leq e^{m \epsilon^2_m (r^2 + 1) } {\int_{\Bcal} \left(\frac{p_{m, w }} {p_{m, w_0}} \right) (\yb_j) d \Pi (w)}.
  \end{align}
  Substituting $\Bcal = \Fcal_r^c$ in \eqref{eq:8} implies that
  \begin{align*}
    A_3 &= \PP_{m, w_0} \left\{ \Pi_m (\Fcal_r^c \mid \yb_j) \one_{\Acal} \right\}
    \leq e^{m \epsilon^2_m (r^2 + 1) } \PP_{m, w_0} \left\{ \int_{\Fcal_r^c} \left(\frac{p_{m, w }} {p_{m, w_0}} \right) (\yb_j) d \Pi (w) \right\} \\
      &= e^{m \epsilon^2_m(r^2 + 1) } \int_{\Fcal_r^c}  \PP_{m, w_0} \left\{  \left(\frac{p_{m, w }} {p_{m, w_0}} \right) (\yb_j)  \right\}  d \Pi (w)  \leq  e^{mk \epsilon^2_m(r^2 + 1) } \int_{\Fcal_r^c} d \Pi(w) \\
      &=e^{m \epsilon^2_m (r^2 + 1)} \Pi(\Fcal_r^c) \overset{(i)}{ \leq}  e^{m \epsilon^2_m (r^2 + 1)}  e^{- 2 m  \epsilon^2_m r^2} = e^{- m \epsilon_m^2 (r^2 - 1)},
  \end{align*}
  where $(i)$ follows from Assumption A.3.\\

  \noindent \underline{\emph{Bounding $A_4$}}\\
  With a little abuse of notation, let $E_{i,r}=\left\{w \in \Fcal_r: 8 i  \epsilon_m r \leq \| \wb^* - \wb^*_0 \| _{\BB} \leq 8 (i+1)\epsilon_m r \right\} $ for $i \geq 1$, then
  \begin{align*}
  A_4 &\leq e^{m \epsilon^2_m (r^2 + 1)} \sum_{i \geq 1} \PP_{m, w_0} \int_{E_{i,r}} \left( \frac{p_{m,w}}{p_{m,w_0}} \right)(\yb_j) \{1 - \phi(\yb_j)\} d\Pi(w) \\
      &= e^{m \epsilon^2_m (r^2 + 1)} \sum_{i \geq 1}  \int_{E_{i,r}}  \left[ \PP_{m,w_0} \left( \frac{p_{m,w}}{p_{m,w_0}} \right) (\yb_j) \{1 - \phi(\yb_j)\} \right] d\Pi(w) \\
      &= e^{m \epsilon^2_m (r^2 + 1)} \sum_{i \geq 1}  \int_{E_{i,r}}   \PP_{m,w} \left\{  1 - \phi(\yb_j) \right\}   d\Pi(w) \\
      &\leq e^{m \epsilon^2_m (r^2 + 1)} \sum_{i \geq 1}  \int_{E_{i,r}}   \underset{\{w \in \Fcal_r : \| \wb^* - \wb^*_0 \|_{\HH} \geq 8 i  \sqrt{mk} \epsilon_{m} r \}} {\sup} \PP_{m,w} \left\{  1 - \phi(\yb_j) \right\}   d\Pi(w) \\
      &\overset{(i)}{\leq} e^{m \epsilon^2_m (r^2 + 1)} \sum_{i \geq 1}  \int_{E_{i,r}}   e^{- 2 i^2 m\epsilon^2_{m} H_l^2r^2}   d\Pi(w)
      \leq e^{m \epsilon^2_m (r^2 + 1 )} \sum_{i \geq 1}  e^{- 2 i^2 m\epsilon^2_{m} H_l^2r^2}   \\
      &\overset{(ii)}{\leq}  e^{-m \epsilon^2_m (r^2 - 1)} \left( 1 + \sum_{i \geq 1} e^{- 2 i}  \right)
      =  e^{-m \epsilon^2_m (r^2 - 1)} (1 - e^{-2})^{-1} \overset{(iii)}{\leq} 2 e^{-m \epsilon^2_m (r^2 - 1)},
  \end{align*}
  where $(i)$ follows from \eqref{eq:tst13}, $(ii)$ follows since $r > \underline r > H_l^{-2} > 1$ so that $ m\epsilon^2_{m}H_l^2r^2 > 1 $, and $(iii)$ follows because $1 / (1 - e^{-2}) < 2$.\\

  We use the upper bounds for $A_1, A_2, A_3$ and $A_4$ to obtain a general upper bound for the risk of $j$th subset posterior distribution. The expectation of the posterior risk can be bounded as
\begin{align*}
& \PP_{m, w_0} \left\{ \int \| \wb^* - \wb^*_0 \|_{\Scal}^2 d\Pi_m (w \mid \yb_j) \right\} \\
= &\PP_{m, w_0} \left\{ \int_{0}^{\infty} \Pi_m \left( \| \wb^* - \wb^*_0 \|_{\Scal}^2 > t \mid \yb_j \right)  d t \right\} \\
= & ~ \PP_{m, w_0} \left\{ \int_{0}^{(8 \epsilon_m \underline r)^2} \Pi_m \left( \| \wb^* - \wb^*_0 \|_{\Scal}^2 > t \mid \yb_j \right)  d t \right\} \\
&+ \PP_{m, w_0} \left\{ \int_{(8 \epsilon_m \underline r)^2}^{\infty} \Pi_m \left( \| \wb^* - \wb^*_0 \|_{\Scal}^2 > t \mid \yb_j \right)  d t \right\} \\
\stackrel{(i)}{\leq} & ~ 64 \underline r^2 \epsilon_m^2 + 128\epsilon_m^2\PP_{m, w_0} \left\{ \int_{\underline r}^{\infty} r \Pi_m \left( \| \wb^* - \wb^*_0 \|_{\Scal} > 8r\epsilon_m \mid \yb_j \right)  d r \right\} \\
\leq & ~  64 \underline r^2 \epsilon_m^2 + 128\epsilon_m^2 \PP_{m, w_0} \left\{ \int_{\underline r}^{\infty} r \left( A_1+A_2+A_3+A_4 \right)  d r \right\} \\
\leq & ~ 64 \underline r^2 \epsilon_m^2 + 128\epsilon_m^2 \nonumber\\
&\times \PP_{m, w_0} \Big\{ \int_{\underline r}^{\infty} r \Big( e^{-m \epsilon_m^2 r^2 / 8} + 33 e^{-m \epsilon_m^2  H_l^2r^2} \\
&+ e^{-m \epsilon_m^2 (r^2-1)} + 2  e^{-m \epsilon^2_m (r^2 -1)}\Big)  d r \Big\} \\
\stackrel{(ii)}{\leq} & 64 \underline r^2 \epsilon_m^2  + 128\epsilon_m^2 \nonumber\\
&\times \PP_{m, w_0} \Big\{ \int_{\underline r}^{\infty} r \Big( e^{-m \epsilon_m^2 H_l^2 r^2 / 8} + 33 e^{-m \epsilon_m^2 H_l^2r^2/8} \\
&+  e^{-m \epsilon_m^2 H_l^2 r^2 / 8} + 2 e^{-m \epsilon_m^2 H_l^2 r^2 / 8} \Big)  d r \Big\} \\
\stackrel{(iii)}{\leq} & 64 \underline r^2 \epsilon_m^2 + 128\epsilon_m^2\cdot 37 \cdot \frac{8}{m\epsilon_m^2 H_l^2}   \int_{0}^{\infty} z  e^{-z^2} dz \\
\stackrel{(iv)}{\leq}& \left(64 \underline r^2 + 128\cdot 37 \cdot \frac{8}{2m\epsilon_m^2 H_l^2} \right)  \epsilon_m^2 \\
&\stackrel{(v)}{<} \left\{64\left(1+\frac{1}{8H_l}+\frac{1}{H_l^2} + \frac{1}{1-H_l^2/8}\right)^2 + \frac{2^{15}}{H_l^2}\right\}\epsilon_m^2 \equiv c(H_l)\epsilon_m^2.
\end{align*}
In the display above, $(i)$ is because $\Pi_m \left( \| \wb^* - \wb^*_0 \|_{\Scal}^2 > t \mid \yb_j \right)\leq 1$ and we use a change of variable $t=(8r\epsilon_m)^2$. $(ii)$ follows because $1/8 > H_l^2/8$, $H_l^2 > H_l^2/8$, and $r^2-1 = H_l^2r^2/8 + (1-H_l^2/8)r^2-1 > H_l^2r^2/8 + (1-H_l^2/8)\left\{\underline r^2-1/(1-H_l^2/8)\right\} > H_l^2r^2/8$. For $(iii)$ we use a change of variable $z=\sqrt{m\epsilon_m^2 H_l^2/8} \cdot r$. $(iv)$ follows from $\int_{0}^{\infty} z  e^{-z^2} dz=1/2$. $(v)$ follows from $m\epsilon_m^2\geq 1$ and the definition of $\underline r$. Finally, Assumption A.5 implies that
\begin{align*}
  & \PP_{m, w_0} \left\{ \int \| \wb^* - \wb^*_0 \|_{2}^2 d\Pi_m (w \mid \yb_j) \right\}  \\
  \leq & C_u^2 \PP_{m, w_0} \left\{ \int \| \wb^* - \wb^*_0 \|_{\Scal}^2 d\Pi_m (w \mid \yb_j) \right\} \leq C_u^2 c(H_l)\epsilon_m^2.
\end{align*}
This has proved the first relation in Theorem \ref{subsetbound}.

We now use the subset bound to obtain an upper bound on the DISK pseudo posterior distribution. First, we note that
\begin{align*}
  W_2^2 \left\{ \overline \Pi(\cdot \mid \yb_1, \ldots, \yb_k), \delta_{\wb_0^*} \right\} =
  \int \| \wb^* - \wb^*_0 \|_{2}^2 d \overline \Pi (w \mid \yb_1, \ldots, \yb_k). 
\end{align*}
Second, Lemma 1.7 in \citet{Srietal17} implies that
\begin{align*}
  &W_2^2 \left\{ \overline \Pi(\cdot \mid \yb_1, \ldots, \yb_k), \delta_{\wb_0^*} \right\} \\
  \leq & \frac{1}{k} \sum_{j=1}^k  W_2^2 \left\{ \Pi_m(\cdot \mid \yb_j), \delta_{\wb_0^*} \right\} =  \frac{1}{k} \sum_{j=1}^k\int \| \wb^* - \wb^*_0 \|_{2}^2 d \Pi_m (w \mid \yb_j).
\end{align*}
Therefore,
\begin{align*}
& \EE_{0 \mid \Scal, \Scal^*} \int \| \wb^* - \wb^*_0 \|_{2}^2 d \overline \Pi (w \mid \yb_1, \ldots, \yb_k) \\
\leq & \frac{1}{k} \sum_{j=1}^k \EE_{0 \mid \Scal_j, \Scal^*} \int \| \wb^* - \wb^*_0 \|_{2}^2 d \Pi_m (w \mid \yb_j) \\
  \leq & \frac{l}{k} \sum_{j=1}^k C_u^2 c(H_l)\epsilon_m^2 =  C_u^2  c(H_l)\epsilon_m^2, 
\end{align*}
which has proved the second relation in Theorem \ref{subsetbound}.

\section{Sampling from the subset posterior distributions using a full-rank GP prior}
\label{sec:univ-spat-regr}

Recall the univariate spatial regression model for the data observed at the $i$th location in subset $j$ using a GP prior is
\begin{align}\label{ap:parent_proc}
  y(\sbb_{ji}) = \xb(\sbb_{ji})^T \betab + w(\sbb_{ji}) + \epsilon(\sbb_{ji}), \quad j = 1, \ldots, k, \quad i = 1, \ldots, m_j.
\end{align}
For the simulations and real data analysis, we assume that $C_{\alphab}(\sbb_{ji}, \sbb_{ji'}) = \sigma^2 \rho(\sbb_{ji}, \sbb_{ji'}; \phi)$ and $D_{\alphab} (\sbb_{ji}, \sbb_{ji'})= \one(i = i') \tau^2$, where $\sigma^2, \phi, \tau^2$ are positive scalars, $\rho(\cdot, \cdot)$ is a known positive definite correlation function, and $\one(i = i') = 1$ if $i = i'$ and 0 otherwise. This implies that $\alphab = (\sigma^2, \tau^2, \phi)$. The model in  \eqref{ap:parent_proc} is completed by putting priors on the unknown parameters. The priors distributions on $\betab$ and $\alphab$ have the following forms:
\begin{align}
  \label{ap:priors1}
  \betab \sim N (\mub_{\betab}, \Sigmab_{\betab}), \quad \sigma^2 \sim \text{IG}(a_{\sigma}, b_{\sigma}), \quad \tau^2 \sim \text{IG}(a_{\tau}, b_{\tau}), \quad \phi \sim \text{U}(a_{\phi}, b_{\phi}),
\end{align}
where $\mub_{\betab}, \Sigmab_{\betab},  a_{\sigma}, b_{\sigma}, a_{\tau}, b_{\tau}, a_{\phi}$, and $b_{\phi}$ are constants, $N$ represents the multivariate Gaussian distribution of appropriate dimension,  IG($a$, $b$) represents the Inverse-Gamma distribution with mean $a/(b+1)$ and variance $b / \{(a-1)^2 (a-2)\}$ for $a > 2$, and U($a$, $b$) represents the uniform distribution on the interval $[a, b]$. The spatial process $w(\cdot)$ is assigned a GP prior as
\begin{align}
  \label{ap:priors2}
  w(\cdot) \mid \sigma^2, \phi \sim \text{GP}\{0, C_{\alphab}(\cdot, \cdot)\}, \quad  C_{\alphab}(\cdot, \cdot) = \sigma^2 \rho(\cdot, \cdot; \phi).
\end{align}
The training data $\{\xb(\sbb_{j1}), y(\sbb_{j1})\}, \ldots, \{\xb(\sbb_{jm_j}), y(\sbb_{jm_j})\}$ are observed at the $m_j$ spatial locations and $\Scal_j =  \{\sbb_{j1}, \ldots, \sbb_{jm_j}\}$ contains the locations in subset $j$.

Consider the setup for predictions and inferences on subset $j$. Let $\Scal^* = \{\sbb_1^*, \ldots, \sbb_l^*\}$ be the set of locations such that $\Scal^* \cap \Scal_j = \emptyset$. If $\wb_j^T = \{w(\sbb_{j1}), \ldots, w(\sbb_{jm_j})\}$ and $\epsilonb^T_j = \{\epsilon(\sbb_{j1}), \ldots, \epsilon(\sbb_{jm_j})\}$, then \eqref{ap:parent_proc} implies that $\wb_j$ apriori follows $N\{\zero, \Cb_{j,j}(\alphab)\}$, where $\Cb_{j,j}(\alphab)$ is the block of $\Cb(\alphab)$ that corresponds to the locations in $\Scal_j$, and $\epsilonb_j$   follows $N(\zero, \tau^2 \Ib)$, where $\Ib$ is the identity matrix of appropriate dimension. Given the training data on subset $j$,  our goal is to predict $\yb_j^* = \{ y(\sbb^*_{1}), \ldots, y(\sbb^*_{l})\}$ and to perform posterior inference on $\wb_j^* = \{ w(\sbb_{1}), \ldots, w(\sbb_{l})\}$,  $\betab_j$, and $\alphab_j$, where the subscript $j$ denotes that the predictions and inferences condition only on subset $j$. Standard Markov chain Monte Carlo (MCMC) algorithms exist to achieve this goal \citep{Banetal14}, but conditioning only on subset $j$ ignores the information contained in the other $(k-1)$ subsets, resulting in greater posterior uncertainty compared to the full data posterior distribution.

Stochastic approximation is an approach for proper uncertainty quantification that modifies the likelihood used for sampling from the subset posterior distributions for predictions and inferences. The likelihoods for $\betab$, $\alphab$, and $\wb_j$ are raised to the power of $k$ to compensate for the data in the other $(k-1)$ subsets, where we assume that $m_1 = \cdots = m_k = m$ and $k = n / m$. First, consider stochastic approximation for the likelihood of $\betab$ and $\alphab$. Integrating out $\wb_j$ in \eqref{ap:parent_proc} gives
\begin{align}
  \label{sp:stoc1}
  \yb_j = \Xb_j \betab + \etab_j, \quad \etab_j \sim N\{\zero, \Cb_{j,j}(\alphab) + \tau^2 \Ib\},
\end{align}
where $\Xb_j =[\xb(\sbb_{j1}): \cdots : \xb(\sbb_{jm})]^T \in \RR^{m \times p}$ is the design matrix for subset $j$. The likelihood of $\betab$ and $\alphab$ given $\yb_j$, $\Xb_j$ after stochastic approximation is
\begin{align}
  \label{eq:llk-ab}
  \{l_j(\betab, \alphab)\}^k &= (2 \pi)^{-mk/2}\vert \Cb_{j,j}(\alphab) + \tau^2 \Ib \vert^{-k/2} e^{ -\frac{k}{2} \left( \yb_j - \Xb_j \betab \right)^T \left\{ \Cb_{j,j}(\alphab) + \tau^2 \Ib  \right\}^{-1} \left( \yb_j - \Xb_j \betab \right)}.
\end{align}
The prior distribution for $\betab$ in \eqref{ap:priors1}, the pseudo likelihood in \eqref{eq:llk-ab}, and Bayes rule implies that the density of the $j$th subset posterior distribution for $\betab$ given the rest is
\begin{align*}
  \betab \mid \text{rest}  \propto e^{ -\frac{1}{2} \left( \yb_j - \Xb_j \betab \right)^T \left[ k^{-1} \left\{ \Cb_{j,j}(\alphab) + \tau^2 \Ib \right\}\right]^{-1} \left( \yb_j - \Xb_j \betab \right)} \,
                                   e^{ -\frac{1}{2} \left( \betab - \mub_{\betab} \right)^T \Sigmab_{\betab}^{-1} \left( \betab - \mub_{\betab} \right)}.
\end{align*}
This implies that the complete conditional distribution of $\betab_j$ has density $N(\mb_{j \betab}, \Vb_{j \betab})$, where
\begin{align}
  \label{eq:post-b}
  &\Vb_{j \betab} =  \left[ k \Xb_j^T  \{\Cb_{j,j}(\alphab) + \tau^2 \Ib\}^{-1}\Xb_j + \Sigmab_{\betab}^{-1} \right] ^{-1}, \quad \nonumber\\
  &\mb_{j \betab} = \Vb_{j \betab} \left[ k \Xb_j^T \left\{ \Cb_{j,j}(\alphab) + \tau^2 \Ib \right\}^{-1} \yb_j + \Sigmab_{\betab}^{-1}  \mub_{\betab} \right] .
\end{align}
If the density of the prior distribution for $\alphab$ is assumed to be $\pi(\sigma^2) \pi(\tau^2) \pi(\phi)$, where the prior densities  $\pi(\sigma^2)$, $\pi(\tau^2)$, and $\pi (\phi)$ are defined in \eqref{ap:priors1}, then the pseudo likelihood in \eqref{eq:llk-ab}, and Bayes rule implies that the density of the $j$th subset posterior distribution for $\alphab$ given the rest is
\begin{align}
  \alphab \mid \text{rest}  \propto \, & \vert \Cb_{j,j}(\alphab) + \tau^2 \Ib \vert^{-k/2} e^{ -\frac{1}{2} \left( \yb_j - \Xb_j \betab \right)^T  \left[ k^{-1} \left\{ \Cb_{j,j}(\alphab) + \tau^2 \Ib \right\}\right]^{-1} \left( \yb_j - \Xb_j \betab \right)} \nonumber \\
  & \left(\sigma^2  \right)^{- a_{\sigma} - 1} e^{- b_{\sigma} / \sigma^2}  \left(\tau^2  \right)^{- a_{\tau} - 1} e^{- b_{\tau} / \tau^2}  (b_{\phi} - a_{\phi})^{-1}. \label{eq:post-a}
\end{align}
This density does not have a standard form, so we use a Metropolis-Hastings step with a normal random walk proposal and sample $\alphab_j$ using the \texttt{metrop} function in the R package \texttt{mcmc} \citep{R17}.

Second, we derive the posterior predictive distribution of $\wb_j^*$ given the rest.  The GP prior on $(\wb_j, \wb_j^*)$ implies that the density of $\wb_j^*$ given $\wb_j$ is
\begin{align}
  \label{ap:ws}
\wb_{j}^* \mid \wb_j \sim N \left\{ \Cb_{*, j} (\alphab) \Cb_{j, j}^{-1}(\alphab) \wb_j, \Cb_{*,*}(\alphab) -  \Cb_{*, j} (\alphab) \Cb_{j, j}^{-1}(\alphab) \Cb_{j, *}(\alphab)\right\},
\end{align}
where $\cov(\wb_j^*, \wb_j^*) = \Cb_{*,*}(\alphab)$, $\cov(\wb_j^*, \wb_j) = \Cb_{*, j} (\alphab)$, and $\cov(\wb_j, \wb_j^*) = \Cb_{j, *}(\alphab)$. Given $\alphab$, $\betab$, $\yb_j$, and $\Xb_j$, \eqref{ap:parent_proc} implies that the likelihood of $\wb_j$ after stochastic approximation is
\begin{align}
    \label{ap:llk-wj}
    \{l_j(\wb_j)\}^k &= (2 \pi)^{-mk/2} \vert \tau^2 \Ib \vert^{-k/2} e^{ -\frac{k}{2 \tau^2} \left( \yb_j - \Xb_j \betab - \wb_j \right)^T \left( \yb_j - \Xb_j \betab - \wb_j \right)}.
\end{align}
The GP prior on $\wb_j$, the pseudo likelihood in \eqref{ap:llk-wj}, and Bayes rule implies that the density of the subset posterior distribution for $\wb_j$ given the rest is
\begin{align*}
  \wb_j \mid \text{rest}  \propto e^{ -\frac{1}{2 \tau^2 / k} \left( \yb_j - \Xb_j \betab - \wb_j \right)^T \left( \yb_j - \Xb_j \betab - \wb_j \right)}\,
                                   e^{ -\frac{1}{2} \wb_j^T \Cb_{j,j}^{-1}(\alphab) \wb_j }.
\end{align*}
This implies that the complete conditional distribution of $\wb_j$ has density $N(\mb_{\wb_j}, \Vb_{\wb_j})$, where
\begin{align}
  \label{eq:post-wj}
  \Vb_{\wb_j} =  \left\{ \Cb_{j,j}^{-1}(\alphab) + \tfrac{k}{\tau^2} \Ib \right\}^{-1}, \quad  \mb_{\wb_j} = \frac{k}{\tau^2} \Vb_{\wb_j} (\yb_j - \Xb_j \betab);
\end{align}
therefore, \eqref{ap:ws} and \eqref{eq:post-wj} imply that the complete conditional distribution of $\wb_j^*$ has density $N(\mb_{\wb_j^*}, \Vb_{\wb_j^*})$, where
\begin{align}
  \mb_{\wb_j^*} &= \EE(\wb_j^* \mid \text{rest}) = \Cb_{*, j} (\alphab) \Cb_{j, j}^{-1}(\alphab) \EE (\wb_j \mid \text{rest}) \nonumber\\
  &=  \Cb_{*, j} (\alphab)  \left\{ \Cb_{j,j}(\alphab) + \tfrac{\tau^2}{k} \Ib \right\}^{-1} (\yb_j - \Xb_j \betab) \label{mm-ws}
\end{align}
and
\begin{align}
  \label{vv-ws}
  \Vb_{\wb_j^*} &= \var(\wb_j^* \mid \text{rest}) = \EE \left\{ \var(\wb_j^* \mid \wb_j) \mid \text{rest} \right\} + \text{var} \left\{ \EE(\wb_j^* \mid \wb_j) \mid \text{rest} \right\} \nonumber\\
                                 &= \Cb_{*,*}(\alphab) -  \Cb_{*, j} (\alphab) \Cb_{j, j}^{-1}(\alphab) \Cb_{j, *}(\alphab) + \Cb_{*, j} (\alphab) \Cb_{j, j}^{-1}(\alphab)   \Vb_{\wb_j} \Cb_{j, j}^{-1} (\alphab) \Cb_{j, *} (\alphab).
\end{align}

Finally, we derive the posterior predictive distribution of $\yb_j^*$ given the rest. If $\betab_j$, $\tau_j^2$, $\wb_j^*$ are the samples from the $j$th subset posterior distribution of $\betab$, $\tau^2$, and $\wb^*$, then \eqref{ap:parent_proc} implies that $\yb_j^*$ given the rest is sampled as
\begin{align*}
  \yb_j^* = \Xb_j \betab_j + \wb_j^* + \epsilonb^*_j, \quad \epsilonb^*_j \sim N(\zero, \tau_j^2 \Ib);
\end{align*}
therefore, the complete conditional distribution of $\yb_j^*$ is  $N(\mub_{\yb_j^*}, \Vb_{\yb_j^*})$, where
\begin{align}
  \mub_{\yb_j^*} = \Xb_j \betab_j + \wb_j^*, \quad \Vb_{\yb_j^*}  = \tau^2_j\Ib. \label{mm-vv-ys}
\end{align}

All full conditionals except that of $\alphab$ are analytically tractable in terms of standard distributions in subset $j$ ($j=1, \ldots, k$). The Gibbs sampler with a Metropolis-Hastings step iterates between the following four steps until sufficient number of samples of $\betab_j, \alphab_j, \wb_j^*$, and $\yb_j^*$ are drawn post convergence to the stationary distribution:
\begin{enumerate}
\item Sample $\betab_j$ from $N(\mub_{j \betab}, \Vb_{j \betab})$, where $\mub_{j \betab}$ and $\Vb_{j \betab}$ are defined in \eqref{eq:post-b}.
\item Sample $\alphab_j$ using the Metropolis-Hastings algorithm from the $j$th subset posterior density (up to constants) of $\alphab_j$ in \eqref{eq:post-a} with a normal random walk proposal.
\item Sample $\wb_{j}^*$ from $N(\mub_{\wb_j^*}, \Vb_{\wb_j^*})$, where $\mub_{\wb_j^*}$ and $\Vb_{\wb_j^*}$ are defined in \eqref{mm-ws} and \eqref{vv-ws}.
\item Sample $\yb_{j}^*$ from $N(\mub_{\yb_j^*}, \Vb_{\yb_j^*})$, where $\mub_{\yb_j^*}$ and $\Vb_{\yb_j^*}$ are defined in \eqref{mm-vv-ys}.
\end{enumerate}

\section{Sampling from the subset posterior distributions using a low-rank GP prior}
\label{sec:univ-mpp-regr}

For clarity, we focus on the modified predictive process (MPP) prior as a representative example of low-rank GP prior. The Gibbs sampling algorithm derived in this section is easily extended to other low-rank GP priors. Following the setup in Section \ref{sec:univ-spat-regr}, we assume that $C_{\alphab}(\sbb_{ji}, \sbb_{ji'}) = \sigma^2 \rho(\sbb_{ji}, \sbb_{ji'}; \phi)$ and $D_{\alphab} (\sbb_{ji}, \sbb_{ji'})= \one(i = i') \tau^2$, $\alphab = (\sigma^2, \tau^2, \phi)$, the prior distributions on $\betab$ and $\alphab$ have the same forms as in \eqref{ap:priors1}, and $\Scal_j$ contains the locations in subset $j$. Following  the previous section, we assume that $m_1 = \cdots = m_k = m$ and $k = n / m$.
The only change in this section is that the spatial process $w(\cdot)$ in \eqref{ap:parent_proc} is assigned a MPP prior derived from parent GP prior in \eqref{ap:priors2}. MPP projects the parent GP $w(\cdot)$ onto a subspace spanned by its realization over a set of $r$ locations, $\mathcal{S}^{(0)}=\{\sbb_{1}^{(0)}, \ldots, \sbb_{r}^{(0)}\}$, known as the ``knots'', where no conditions are imposed on  $\Scal \cap \Scal^{(0)}$. Let $\cb(\cdot, \mathcal{S}^{(0)})= \left\{ C_{\alphab}(\cdot,\sbb_{1}^{(0)}), \ldots,C_{\alphab}(\cdot,\sbb_{r}^{(0)}) \right\}^T$ and
$\wb^{(0)}= \left\{ w(\sbb_{1}^{(0)}), \ldots,w(\sbb_{r}^{(0)}) \right\}^T$ be $r\times 1$ vectors and $\Cb(\mathcal{S}^{(0)})$ be an $r\times r$ matrix whose
$(i,j)$th entry is $C_{\alphab}(\sbb_{i}^{(0)},\sbb_{j}^{(0)})$. The MPP prior defines
\begin{align}
  \label{mpp-prior}
  \tilde w(\cdot)  = \cb^T(\cdot, \mathcal{S}^{(0)}) \Cb(\mathcal{S}^{(0)})^{-1}\wb^{(0)} + \tilde \epsilon(\cdot),
\end{align}
where the processes $\tilde \epsilon(\cdot)$ and $w(\cdot)$ are mutually independent and $\tilde \epsilon(\cdot)$ is a GP with mean 0, $\cov\{\tilde \epsilon(\ab), \tilde \epsilon(\bb)\} = \delta(\ab) \one({\ab=\bb}) $ for any $\ab, \bb \in \Dcal$, and
\begin{align*}
  \delta(\sbb_{ji}) =  C_{\alphab}(\sbb_{ji},\sbb_{ji}) - \cb^T(\sbb_{ji}, \mathcal{S}^{(0)}) \Cb(\mathcal{S}^{(0)})^{-1}\cb(\sbb_{ji}, \Scal^{(0)}).
\end{align*}
The process $\tilde w (\cdot)$ is a low-rank GP with mean 0 and $$\cov\{\tilde w(\ab), \tilde w(\bb)\}  = \cb^T(\ab, \mathcal{S}^{(0)}) \Cb(\mathcal{S}^{(0)})^{-1}\cb(\bb, \mathcal{S}^{(0)}) + \delta(\ab) \one_{\ab=\bb}$$ for any $\ab, \bb \in \Dcal$.  If we replace $w(\cdot)$ by $\tilde w(\cdot)$ in \eqref{ap:parent_proc}, then
\begin{align}\label{ap:mpp}
  y(\sbb_{ji}) = \xb(\sbb_{ji})^T \betab + \tilde w(\sbb_{ji}) + \epsilon(\sbb_{ji}), \quad j = 1, \ldots, k,\quad i=1, \ldots, m_j.
\end{align}
and our definition in \eqref{mpp-prior} implies that $\tilde w(\cdot)$ is assigned a MPP prior \citep{Finetal09}.

We start by defining mean and covariance functions specific to univariate spatial regression using MPP. Let $\tilde \wb_j = \{ \tilde w(\sbb_{j1}), \ldots, \tilde  w(\sbb_{jm})\}$ and $\tilde \wb_j^* = \{ \tilde w(\sbb_{1}), \ldots, \tilde  w(\sbb_{l})\}$. The MPP prior is identical to the FITC approximation in sparse approximate GP regression, so we use the FITC notations to simplify the description of posterior computations \citep{QuiRas05}. Define $\Qb_{j, j} = \Cb_{j, 0}(\alphab) \Cb^{-1}(\mathcal{S}^{(0)}) \Cb_{0, j}(\alphab)$, where $\cov\{w(\sbb_{ja}), w(\sbb_b^{(0)})\} = \left\{ \Cb_{j, 0} (\alphab) \right\}_{a,b}$ ($a=1, \ldots, m$; $b=1, \ldots, r$) and $\Cb_{0, j}(\alphab) = \Cb_{j,0}^T(\alphab)$. The density of $(\tilde \wb_j, \tilde \wb_j^*)$ under the GP prior implied by MPP is $N\{\zero, \tilde \Cb({\alphab})\}$, where $2 \times 2$ block form of $\tilde \Cb({\alphab})$ is defined using
\begin{align}
  \label{eq:fitc}
  &\tilde \Cb_{j, j}(\alphab) = \Qb_{j,j} + \diag\{\Cb_{j, j}(\alphab) - \Qb_{j, j}\} = \cov(\tilde \wb_j, \tilde \wb_j),  \nonumber\\    
  &\tilde \Cb_{j, *}(\alphab) = \Qb_{j,*} = \cov(\tilde \wb_j, \tilde \wb_j^*), \nonumber\\
  &\tilde \Cb_{*, *}(\alphab) = \Qb_{*,*} + \diag\{\Cb_{*, *}(\alphab) - \Qb_{*, *}\} = \cov(\tilde \wb_j^*, \tilde \wb^*_j), \nonumber\\
  &\tilde \Cb_{*, j}(\alphab) = \Qb_{*,j} = \cov(\tilde \wb_j^*, \tilde \wb_j).
\end{align}

Stochastic approximation is implemented following Section \ref{sec:univ-spat-regr}. First, consider stochastic approximation for the likelihood of $\betab$ and $\alphab$. Integrating out $\tilde \wb_j$ in \eqref{ap:mpp} gives
\begin{align}
  \label{sp:stoc1-mpp}
  \yb_j = \Xb_j \betab + \tilde \etab_j, \quad \tilde \etab_j \sim N\{\zero, \tilde \Cb_{j,j}(\alphab) + \tau^2 \Ib\}.
\end{align}
The likelihood of $\betab$ and $\alphab$ given $\yb_j$, $\Xb_j$ after stochastic approximation is
\begin{align}
  \label{eq:llk-ab-mpp}
  \{l_j(\betab, \alphab)\}^k &= (2 \pi)^{-mk/2}\vert \tilde \Cb_{j,j}(\alphab) + \tau^2 \Ib \vert^{-k/2} e^{ -\frac{k}{2} \left( \yb_j - \Xb_j \betab \right)^T \left\{ \tilde  \Cb_{j,j}(\alphab) + \tau^2 \Ib  \right\}^{-1} \left( \yb_j - \Xb_j \betab \right)}.
\end{align}
The prior distribution for $\betab$ in \eqref{ap:priors1}, the pseudo likelihood in \eqref{eq:llk-ab-mpp}, and Bayes rule implies that the density of the $j$th subset posterior distribution for $\betab$ given the rest is
\begin{align*}
  \betab \mid \text{rest}  \propto e^{ -\frac{1}{2} \left( \yb_j - \Xb_j \betab \right)^T \left[ k^{-1} \left\{\tilde \Cb_{j,j}(\alphab) + \tau^2 \Ib \right\}\right]^{-1} \left( \yb_j - \Xb_j \betab \right)} \,
                                   e^{ -\frac{1}{2} \left( \betab - \mub_{\betab} \right)^T \Sigmab_{\betab}^{-1} \left( \betab - \mub_{\betab} \right)}.
\end{align*}
This implies that the complete conditional distribution of $\betab_j$ has density $N(\tilde \mb_{j \betab}, \tilde \Vb_{j \betab})$, where
\begin{align}
  \label{eq:post-b-mpp}
  &\tilde \Vb_{j \betab} =  \left[ k \Xb_j^T  \{\tilde \Cb_{j,j}(\alphab) + \tau^2 \Ib\}^{-1}\Xb_j + \Sigmab_{\betab}^{-1} \right] ^{-1}, \quad  \nonumber\\
  &\tilde \mb_{j \betab} = \tilde \Vb_{j \betab} \left[ k \Xb_j^T \left\{ \tilde \Cb_{j,j}(\alphab) + \tau^2 \Ib \right\}^{-1} \yb_j + \Sigmab_{\betab}^{-1}  \mub_{\betab} \right] .
\end{align}
Following Section \ref{sec:univ-spat-regr}, the density of the $j$th subset posterior distribution for $\alphab$ given the rest is
\begin{align}
  \alphab \mid \text{rest}  \propto \, & \vert \tilde \Cb_{j,j}(\alphab) + \tau^2 \Ib \vert^{-k/2} e^{ -\frac{1}{2} \left( \yb_j - \Xb_j \betab \right)^T  \left[ k^{-1} \left\{  \tilde  \Cb_{j,j}(\alphab) + \tau^2 \Ib \right\}\right]^{-1} \left( \yb_j - \Xb_j \betab \right)} \nonumber \\
  & \left(\sigma^2  \right)^{- a_{\sigma} - 1} e^{- b_{\sigma} / \sigma^2}  \left(\tau^2  \right)^{- a_{\tau} - 1} e^{- b_{\tau} / \tau^2}  (b_{\phi} - a_{\phi})^{-1}. \label{eq:post-a-mpp}
\end{align}
This density does not have a standard form, so we use a Metropolis-Hastings step with a normal random walk proposal and sample $\alphab_j$ using the \texttt{metrop} function in the R package \texttt{mcmc}.

Second, we derive the posterior predictive distribution of $\tilde \wb_j^*$ given the rest.  The MPP prior on $(\tilde \wb_j, \tilde \wb_j^*)$ implies that the density of $\tilde \wb_j^*$ given $\tilde \wb_j$ is
\begin{align}
  \label{ap:ws-mpp}
  \tilde \wb_{j}^* \mid \tilde \wb_j \sim N \left\{ \tilde \Cb_{*, j} (\alphab) \tilde \Cb_{j, j}^{-1} (\alphab) \tilde \wb_j, \tilde \Cb_{*,*}(\alphab) -  \tilde \Cb_{*, j}(\alphab) \tilde \Cb_{j, j}^{-1}(\alphab) \tilde \Cb_{j, *}(\alphab)\right\}.
\end{align}
Given $\alphab$, $\betab$, $\yb_j$, and $\Xb_j$, \eqref{ap:mpp} implies that the likelihood of $\tilde \wb_j$ after stochastic approximation is
\begin{align}
    \label{ap:llk-wj-mpp}
    \{l_j(\tilde \wb_j)\}^k &= (2 \pi)^{-mk/2} \vert \tau^2 \Ib \vert^{-k/2} e^{ -\frac{k}{2 \tau^2} \left( \yb_j - \Xb_j \betab - \tilde \wb_j \right)^T \left( \yb_j - \Xb_j \betab - \tilde \wb_j \right)}.
\end{align}
The MPP prior on $\tilde \wb_j$, the pseudo likelihood in \eqref{ap:llk-wj-mpp}, and Bayes rule implies that the density of the subset posterior distribution for $\tilde \wb_j$ given the rest is
\begin{align*}
\tilde  \wb_j \mid \text{rest}  \propto e^{ -\frac{1}{2 \tau^2 / k} \left( \yb_j - \Xb_j \betab - \tilde \wb_j \right)^T \left( \yb_j - \Xb_j \betab - \tilde \wb_j \right)}\,
                                   e^{ -\frac{1}{2} \tilde  \wb_j^T \tilde \Cb_{j,j}^{-1}(\alphab) \tilde  \wb_j }.
\end{align*}
This implies that the complete conditional distribution of $\tilde \wb_j$ has density $N(\mb_{\tilde \wb_j}, \Vb_{\tilde \wb_j})$, where
\begin{align}
  \label{eq:post-wj-mpp}
  \Vb_{\tilde \wb_j} =  \left\{\tilde  \Cb_{j,j}^{-1}(\alphab) + \tfrac{k}{\tau^2} \Ib \right\}^{-1}, \quad  \mb_{\tilde \wb_j} = \frac{k}{\tau^2} \Vb_{\tilde \wb_j} (\yb_j - \Xb_j \betab);
\end{align}
therefore, \eqref{ap:ws-mpp} and \eqref{eq:post-wj-mpp} imply that the complete conditional distribution of $\tilde \wb_j^*$ has density $N(\mb_{\tilde \wb_j^*}, \Vb_{\tilde \wb_j^*})$, where
\begin{align}
  \mb_{\tilde \wb_j^*} &= \EE(\tilde \wb_j^* \mid \text{rest}) = \tilde \Cb_{*, j} (\alphab) \tilde \Cb_{j, j}^{-1}(\alphab) \EE (\tilde \wb_j \mid \text{rest}) \nonumber\\
                     &= \tilde  \Cb_{*, j} (\alphab)  \left\{ \tilde \Cb_{j,j}(\alphab) + \tfrac{\tau^2}{k} \Ib \right\}^{-1} (\yb_j - \Xb_j \betab) \label{mm-ws-mpp}
\end{align}
and
\begin{align}
  \label{vv-ws-mpp}
  \Vb_{\tilde \wb_j^*} &= \var(\tilde \wb_j^* \mid \text{rest}) = \EE \left\{ \var(\tilde \wb_j^* \mid \tilde \wb_j) \mid \text{rest} \right\} + \text{var} \left\{ \EE(\tilde \wb_j^* \mid \tilde  \wb_j) \mid \text{rest} \right\} \nonumber\\
                                 &= \tilde \Cb_{*,*}(\alphab) -  \tilde \Cb_{*, j} (\alphab) \tilde \Cb_{j, j}^{-1}(\alphab) \tilde \Cb_{j, *}(\alphab) + \tilde \Cb_{*, j} (\alphab) \tilde \Cb_{j, j}^{-1}(\alphab)   \Vb_{\tilde \wb_j} \tilde \Cb_{j, j}^{-1} (\alphab) \tilde \Cb_{j, *} (\alphab).
\end{align}

Finally, we derive the posterior predictive distribution of $\yb_j^*$ given the rest. If $\betab_j$, $\tau_j^2$, $\tilde \wb_j^*$ are the samples from the $j$th subset posterior distribution of $\betab$, $\tau^2$, and $\tilde \wb^*$, then \eqref{ap:mpp} implies that $\yb_j^*$ given the rest is sampled as
\begin{align*}
  \yb_j^* = \Xb_j \betab_j + \tilde \wb_j^* + \epsilonb^*_j, \quad \epsilonb^*_j \sim N(\zero, \tau_j^2 \Ib);
\end{align*}
therefore, the complete conditional distribution of $\yb_j^*$ has density $N(\tilde \mub_{\yb_j^*}, \tilde \Vb_{\yb_j^*})$, where
\begin{align}
  \tilde \mub_{\yb_j^*} = \Xb_j \betab_j + \tilde \wb_j^*, \quad   \tilde  \Vb_{\yb_j^*}  = \tau^2_j\Ib. \label{mm-vv-ys-mpp}
\end{align}

All full conditionals except that of $\alphab$ are analytically tractable in terms of standard distributions in subset $j$ ($j=1, \ldots, k$). The Gibbs sampler with a Metropolis-Hastings step iterates between the following four steps until sufficient number of samples of $\betab_j, \alphab_j, \tilde \wb_j^*$, and $\yb_j^*$ are drawn post convergence to the stationary distribution:
\begin{enumerate}
\item Sample $\betab_j$ from $N(  \tilde \mub_{j \betab},   \tilde \Vb_{j \betab})$, where $  \tilde \mub_{j \betab}$ and $  \tilde \Vb_{j \betab}$ are defined in \eqref{eq:post-b-mpp}.
\item Sample $\alphab_j$ using the Metropolis-Hastings algorithm from the $j$th subset posterior density (up to constants) of $\alphab_j$ in \eqref{eq:post-a-mpp} with a normal random walk proposal.
\item Sample $  \tilde \wb_{j}^*$ from $N(\mub_{  \tilde \wb_j^*}, \Vb_{  \tilde \wb_j^*})$, where $\mub_{  \tilde \wb_j^*}$ and $\Vb_{  \tilde \wb_j^*}$ are defined in \eqref{mm-ws-mpp} and \eqref{vv-ws-mpp}.
\item Sample $\yb_{j}^*$ from $N(  \tilde \mub_{\yb_j^*},   \tilde \Vb_{\yb_j^*})$, where $  \tilde \mub_{\yb_j^*}$ and $  \tilde \Vb_{\yb_j^*}$ are defined in \eqref{mm-vv-ys-mpp}.
\end{enumerate}

\bibliographystyle{Chicago}
\bibliography{papers}



\end{document}